\documentclass[12pt]{article}
\usepackage{amsmath}
\usepackage{amsfonts}
\usepackage{amssymb}
\usepackage{graphicx}
\usepackage{tablefootnote}
\usepackage{threeparttable}
\usepackage{float}
\usepackage{caption}
\usepackage{tikz}
\usepackage{booktabs}
\usepackage{pifont}

\usepackage[margin=1in]{geometry}

\usepackage[
authordate,
backend=biber,
date=year,
maxcitenames=2,
mincitenames=1,
doi=false,
url=false,
isbn=false,
eprint=false
]{biblatex-chicago}
\newcommand*{\bibtitle}{References}

\DeclareDelimFormat{nameyeardelim}{\addcomma\space}

\let\cite\textcite

\newcommand\fnote[1]{\captionsetup{font=footnotesize}\caption*{#1}}

\usepackage{graphicx} % Required for inserting images
\usepackage{amsmath,amsthm,amssymb}
\usepackage{mathtools}
\usepackage{dsfont}
\usepackage{statmath}
\usepackage{calrsfs}
\usepackage{booktabs}
\usepackage{enumitem}
\usepackage{marvosym}
\usepackage{accents}
\usepackage{tikz}
\usepackage{float}
\usepackage{threeparttable}
\usepackage{caption}
\usepackage{subcaption}
\usepackage{tabularx}
\newcommand{\ubar}[1]{\underaccent{\bar}{#1}}
\newcommand{\sym}[1]{\ifmmode^{#1}\else$^{#1}$\fi}

\renewcommand{\arraystretch}{1.25} % Increases table row height
\newtheorem{theorem}{Theorem} 
\newtheorem{prop}{Proposition}
\newtheorem{lemma}{Lemma} 
\newtheorem{assumption}{Assumption}
\renewcommand{\theassumption}{\Alph{assumption}}

\begin{document}

\title{The Econometrics of Utility Transferability in Dyadic Network Formation Models}
\author{Joseph Marshall\thanks{joseph.marshall@nuffield.ox.ac.uk} \\
%EndAName
University of Oxford\thanks{Nuffield College, Oxford, United Kingdom, OX1 1NF}}
\maketitle

\begin{abstract}
This paper studies how to estimate an individual's taste for forming a connection with another individual in a network when transfers are not allowed. We provide sufficient conditions for the consistency and asymptotic normality of the maximum likelihood estimator that does not assume transferable utility (NTU-MLE). We also provide a specification test to assess the validity of the transferable utility (TU) assumption. Applied to data on Indian villages from \cite{banerjee_changes_2024}, our test rejects the validity of the transferable utility assumption in 20-44\% of villages, depending on the network measure. Using the maximum likelihood estimator that assumes transferable utility (TU-MLE) instead of the NTU-MLE when the TU assumption is rejected can change the sign and significance of the estimates.
\end{abstract}

\newpage

\section{Introduction}

Given a group of individuals, can we determine how much benefit one individual derives from connecting with another? And, by extension, can we predict who will connect with whom? These questions are central to the study of social and economic networks. For example, suppose that a new technology is invented and trade economists wish to predict which countries will become exporters of the technology and which will become importers. Or imagine an NGO tasked with running group cognitive behavioral therapy sessions to improve the mental health of struggling adolescents. If the NGO wishes for the adolescents to form friendships with people in their groups so that they can discuss their shared experiences, then the NGO needs to know which group designs are most conducive to friendship formation. Would it be better to have separate groups of boys and girls? Or would it be better to construct groups of adolescents with different characteristics? This paper contributes to the econometric models that seek to answer questions like these. 

To understand the contribution of this paper, it is necessary to introduce two econometric models that are used to estimate individuals' taste for linking with one another. The first model assumes that utility is transferable between individuals; the second model does not.

\subsection{Transferable Utility (TU) Model}

The simpler econometric model used to determine how much benefit two individuals derive from forming a connection is known as the transferable utility (TU) model. This is a binary choice model in which the outcome is an indicator for if the two individuals decide to form a link or not. The name of the model derives from the fact that utility is assumed to be transferable between individuals. If individual $i$ wishes to form a link with individual $j$, but $j$ does not wish to reciprocate, then this assumption means that $i$ can provide $j$ with transfers of some kind (e.g. money) such that forming the link becomes profitable for $j$. The practical implication of this assumption is that we only require a single indicator to summarize the joint linking decision of the two individuals. 

Formally, suppose that individuals $i$ and $j$ form a link, denoted $Y_{ij}^{\text{TU}} = 1$, if the joint (additive) utility that they derive from the link is non-negative. Individual $i$'s utility from linking with $j$ is $U_{ij} \coloneq W_{ij}'\beta - \varepsilon_{ij}$, where: 1) $W_{ij} \coloneq w(X_i, X_j) \in \mathbb{R}^k$ is some continuous function that combines $i$'s random vector of characteristics, $X_i \in \mathbb{R}^p$, with $j$'s random vector of characteristics, $X_j \in \mathbb{R}^p$; 2) $\beta \in \Theta_{\beta} \subseteq\mathbb{R}^k$ is a vector of ``taste for linking parameters''; and 3) $\varepsilon_{ij} \in \mathbb{R}$ is an idiosyncratic error term. Likewise, $j$'s utility from linking with $i$ is $U_{ji} \coloneq W_{ji}'\beta - \varepsilon_{ji}$, with $W_{ji} \coloneq w(X_j, X_i) \in \mathbb{R}^k$. We assume that $(\varepsilon_{ij}, \varepsilon_{ji})$ follows a bivariate normal distribution with $\mathrm{Var}(\varepsilon_{ij})$ and $\mathrm{Var}(\varepsilon_{ji})$ both normalized to 1 and $\mathrm{Cov}(\varepsilon_{ij}, \varepsilon_{ji}) = \rho \in (-1, 1)$. Their joint utility is $U_{ij} + U_{ji} = (W_{ij} + W_{ji})^{'}\beta - (\varepsilon_{ij} + \varepsilon_{ji})$, so the TU model is given by:
\begin{equation}\label{TU Model}
    Y_{ij}^{\text{TU}} = \mathds{1} [(W_{ij} + W_{ji})^{'}\beta \geq \varepsilon_{ij} + \varepsilon_{ji}],
\end{equation}
where: \begin{equation*} \label{eq:bnorm}
\begin{pmatrix}
\varepsilon_{ij} \\
\varepsilon_{ji}
\end{pmatrix}
\sim \mathcal{N} \left(
\begin{pmatrix}
0 \\
0
\end{pmatrix},
\begin{pmatrix}
1 & \rho \\
\rho  & 1
\end{pmatrix}
\right),
\end{equation*}
with $\rho \in (-1,1)$.\footnote{There are several reasons to impose joint normality of the errors; we discuss these reasons in Section 1.3.}

For example, suppose that $\beta \in \mathbb{R}$ and $W_{ij} = W_{ji} = |X_i - X_j|$ where $X_i$ is $i$'s income and $X_j$ is $j$'s income. In this case, $\beta$ is the effect of the distance between $i$ and $j$'s incomes on each agent's payoff from the friendship, so the joint surplus changes by $2\beta$ per unit of income distance. The sign of $\beta$ therefore determines whether income distance makes friendship more or less likely. We typically find that individuals who share similar characteristics are more likely to be friends, a phenomenon known as ``homophily'' in the networks literature \parencite{jackson_social_2008}. Therefore, we might expect $\beta$ to be negative here to capture the fact that people with far apart incomes are less likely to be friends. In this sense, (the negative of) $\beta$ can be thought of as a homophily parameter.

\subsection{Non-Transferable Utility (NTU) Model}

If one is not willing to assume that utility is transferable between individuals, then one cannot express the link decision of the two individuals as a single indicator containing a joint utility term. Instead, for a link to form, denoted $Y_{ij}^{\text{NTU}} = 1$, we require that $i$ and $j$'s separate utilities from forming the link are simultaneously non-negative. Formally, we have the product of two indicators: an indicator containing $i$'s utility from forming the link and an indicator containing $j$'s utility from forming the link. That is, the NTU model is given by:
\begin{equation}\label{NTU Model}
    Y_{ij}^{\text{NTU}} = \mathds{1} [W_{ij}^{'}\beta \geq \varepsilon_{ij}] \cdot \mathds{1} [W_{ji}^{'}\beta \geq \varepsilon_{ji}],
\end{equation}
where again:
\begin{equation*} \label{eq:bnorm}
\begin{pmatrix}
\varepsilon_{ij} \\
\varepsilon_{ji}
\end{pmatrix}
\sim \mathcal{N} \left(
\begin{pmatrix}
0 \\
0
\end{pmatrix},
\begin{pmatrix}
1 & \rho \\
\rho & 1
\end{pmatrix}
\right),
\end{equation*}
with $\rho \in (-1, 1)$.

Let $W_{lij}$ be the $l$-th element of $W_{ij}$ (likewise for $W_{lji}$). We say that $W_{lij}$ is a \textit{symmetric regressor} if $W_{lij} = W_{lji}$ for all pairs $ij$. One common example of a symmetric regressor is $W_{lij} = W_{lji} = |X_i - X_j|$. Alternatively, we say that $W_{lij}$ and $W_{lji}$ are \textit{asymmetric regressors} if there exists a pair $ij$ such that $W_{lij} \neq W_{lji}$, e.g. if $W_{lij} = X_j$ and $W_{lji} = X_i$. Asymmetric regressors are necessary to interpret this model as an NTU model. If it was the case that all regressors were symmetric, i.e. $W_{ij} = W_{ji}$ for all $ij$, then the model would become $Y_{ij}^{\text{NTU}} = \mathds{1} [W_{ij}^{'}\beta \geq \max\{\varepsilon_{ij}, \varepsilon_{ji}\}].$ This is simply the TU model with $\beta$ re-scaled and a different error distribution.

To obtain the log-likelihood function for the NTU model, we sum across unique links (excluding self-links) and normalize by the effective sample size:\footnote{We have duplicate links because $Y_{ij}^{\text{NTU}} = Y_{ji}^{\text{NTU}}$, i.e. the NTU model generates undirected networks (as does the TU model). To see how the double summation ignores self and duplicate links, set $n = 3$. Then the links that are considered are $(1, 2)$, $(1, 3)$, and $(2, 3)$.}
\[
\begin{aligned}
    \hat{Q}_n(\beta, \rho) &\coloneq \frac{1}{N}\sum^{n-1}_{i=1} \sum^{n}_{j=i+1} \Big[Y_{ij} \log \Phi_2^{\rho}\big(W_{ij}^{'}\beta, W_{ji}^{'}\beta\big) + (1 - Y_{ij})\log \big(1 - \Phi_2^{\rho}\big(W_{ij}^{'}\beta, W_{ji}^{'}\beta\big)\big) \Big],
\end{aligned}
\]
where $\Phi_2^{\rho}$ is the bivariate normal CDF with covariance $\rho$, the effective sample size is $N = \binom{n}{2} = \frac{n(n-1)}{2}$, and $n$ is the number of individuals in the network.

Let $\Theta$ be the space containing the parameters $(\beta, \rho)$.\footnote{We place structure on $\Theta$, in particular compactness, in Section 2.} Our first object of interest in this paper is the maximum likelihood estimator for the NTU model (NTU-MLE):
$$
\hat{\theta} = (\hat{\beta}, \hat{\rho}) \coloneq \argmax_{(\beta, \rho) \in \Theta} \hat{Q}_n(\beta, \rho).
$$
Proving that the NTU-MLE, $\hat{\theta}$, is consistent and asymptotically normal under standard assumptions is the focus of Section 2 and the first contribution of this paper.

\subsection{Contribution}

This subsection outlines how this paper is structured and the contributions it makes. 

Most papers in the dyadic network formation literature assume that utility is transferable between agents.\footnote{There is a parallel literature on ``strategic'' network formation models where the decision of two agents to form a link is allowed to depend on the existence of other links in the network. While arguably more realistic in some settings, strategic models are plagued with issues of multiple equilibria. For excellent reviews of strategic models of network formation, see Section 4 of \cite{chandrasekhar_econometrics_2016} and Section 4 of \cite{de_paula_econometric_2020}.} The landmark paper on the TU dyadic network formation model is \cite{graham_econometric_2017}, which adds unobserved individual heterogeneity terms for individual $i$, $A_i$, and individual $j$, $A_j$, to the TU link decision indicator:\footnote{Technically, \cite{graham_econometric_2017} assumes symmetric regressors and a single independent logistic error term such that his $\beta$ is a re-scaled version of ours.}
$$
Y_{ij}^{\text{TU}} = \mathds{1} [(W_{ij} + W_{ji})^{'}\beta + A_i + A_j\geq \varepsilon_{ij} + \varepsilon_{ji}].
$$
The introduction of fixed effects leads to the incidental parameter problem \parencite{neyman_consistent_1948}. To eliminate the fixed effects, \cite{graham_econometric_2017} uses sufficient statistics that depend on the error's parametric distribution. \cite{candelaria_semiparametric_2020} and \cite{toth_semiparametric_2017} extend \cite{graham_econometric_2017} to a semi-parametric setting in which the distribution of $\varepsilon_{ij} + \varepsilon_{ji}$ is left unspecified; \cite{jochmans_semiparametric_2018} and \cite{dzemski_empirical_2019} consider related models that generate directed networks.

At the time of writing, three papers have been written on dyadic network formation models with non-transferable utility: \cite{gao_logical_2023}, \cite{chen_moment_2026}, and \cite{li_bagging_2026}.\footnote{There are strategic network formation models with NTU, such as in \cite{goldsmith-pinkham_social_2013}, but again these are not the focus of this paper.} These papers include unobserved individual heterogeneity terms for $i$ and $j$ in their respective indicators such that their main model is given by:
$$
Y_{ij}^{\text{NTU}} = \mathds{1} [W_{ij}^{'}\beta + A_i \geq \varepsilon_{ij}] \cdot \mathds{1} [W_{ij}^{'}\beta + A_j \geq \varepsilon_{ji}].
$$
Since \cite{gao_logical_2023}, \cite{chen_moment_2026}, and \cite{li_bagging_2026} are the works that are most similar to this paper, it is beneficial to describe what exactly they do and how their contributions are distinct from this paper's contributions. 

In terms of model assumptions, both \cite{chen_moment_2026} and \cite{li_bagging_2026} are parametric; \cite{chen_moment_2026} assume $\varepsilon_{ij}$ and $\varepsilon_{ji}$ are independent logistic random variables, while \cite{li_bagging_2026} keep the joint distribution of $(\varepsilon_{ij}, \varepsilon_{ji})$ general enough to incorporate both the logistic and normal distributions as special cases, but still impose independence of $\varepsilon_{ij}$ and $\varepsilon_{ji}$. On the other hand, \cite{gao_logical_2023} is semi-parametric in that they leave the joint distribution of $(\varepsilon_{ij}, \varepsilon_{ji})$ unspecified. While the semi-parametric set-up is more general than the set-up in \cite{chen_moment_2026} and \cite{li_bagging_2026}, \cite{gao_logical_2023} lack inference results because the asymptotic distribution of their estimator is non-standard and unknown.\footnote{This paper imposes joint normality of the errors. As well as allowing the errors to be correlated through a single interpretable parameter, $\rho$, the normality assumption ensures linear combinations of $\varepsilon_{ij}$ and $\varepsilon_{ji}$ remain normal. This proves useful in the development of the specification test. Furthermore, normality has been well-studied and we are therefore able to rely on several known inequalities, such as \cite{gordon_values_1941}, \cite{birnbaum_inequality_1942}, and \cite{sampford_inequalities_1953}. The normal distribution also closely approximates the logistic distribution, so the choice between the two makes little practical difference in this parametric setting.} Finally, \cite{gao_logical_2023} and \cite{chen_moment_2026} allow for asymmetric regressors, while \cite{li_bagging_2026} do not.\footnote{It is worth noting that while we need asymmetric regressors to interpret our model as an NTU model as discussed in Section 1.2, \cite{gao_logical_2023}, \cite{chen_moment_2026}, and \cite{li_bagging_2026} do not because they include unobserved individual heterogeneity terms, $A_i$ and $A_j$, which distinguish $i$ and $j$'s link decision indicators.} These differences are summarized in Table \ref{tab:literature}. To the best of our knowledge, this is the only dyadic NTU paper that simultaneously allows for asymmetric regressors, dependent errors, and inference (at the cost of excluding fixed effects).

\newcommand{\cmark}{\ding{51}}  % tick
\newcommand{\xmark}{\ding{55}}  % cross

\begin{table}[htbp]
  \centering
  \caption{Comparison with Closely Related Papers}
  \label{tab:literature}
  \begin{threeparttable}
  \begin{tabular}{lcccc}
    \toprule
    Paper
      & \begin{tabular}[b]{@{}c@{}}Fixed\\ effects?\end{tabular}
      & \begin{tabular}[b]{@{}c@{}}Inference?\\ \phantom{x}\end{tabular}
      & \begin{tabular}[b]{@{}c@{}}Asymmetric\\ regressors?\end{tabular}
      & \begin{tabular}[b]{@{}c@{}}Dependent\\ errors?\end{tabular} \\
    \midrule
     \cite{gao_logical_2023} & \cmark & \xmark & \cmark & \cmark \\
     \cite{chen_moment_2026} & \cmark & \cmark & \cmark & \xmark \\
     \cite{li_bagging_2026} & \cmark & \cmark & \xmark & \xmark \\
    \midrule
    This paper & \xmark & \cmark & \cmark & \cmark \\
    \bottomrule
  \end{tabular}
  \begin{tablenotes}[flushleft]
    \footnotesize
    \item \cmark\ indicates that the paper includes the feature and \xmark\ that it does not. ``Fixed effects'' means that $A_i$ is included in $i$'s indicator and $A_j$ in $j$'s; ``Inference'' means that the paper includes a result on the asymptotic distribution of their estimator; ``Asymmetric regressors'' means that there can exist pairs of individuals for which $W_{ij} \neq W_{ji}$; and ``Dependent errors'' means $\mathrm{Cov}(\varepsilon_{ij}, \varepsilon_{ji}) \neq 0$.
  \end{tablenotes}
  \end{threeparttable}
\end{table}

In terms of estimation, \cite{gao_logical_2023} and \cite{chen_moment_2026} rely on elimination of the fixed effects, while \cite{li_bagging_2026} can estimate them. \cite{gao_logical_2023} eliminate the fixed effects, which are not additively separable in the NTU model unlike in the TU model, via a procedure called ``logical differencing''. \cite{chen_moment_2026} is the natural extension of \cite{graham_econometric_2017} to NTU, with elimination of the fixed effects now relying on sets of five nodes (pentads) instead of four (tetrads). The estimation approach in \cite{li_bagging_2026} has multiple steps: first solve a high-dimensional system of moment equations, then plug this solution into a one-step estimator of \cite{le_cam_l_m_theorie_1969}, then refine the estimator to achieve $\sqrt{N}$-consistency via split-network jackknife, and finally perform bootstrap aggregating to achieve Cram\'er-Rao efficiency. One might worry that the results presented in this paper are special cases of the results in \cite{gao_logical_2023}, \cite{chen_moment_2026}, or \cite{li_bagging_2026} with the fixed effects set equal to zero. This is not the case because the logical differencing estimator in \cite{gao_logical_2023} and the moment condition-based estimators in \cite{chen_moment_2026} and \cite{li_bagging_2026} are very different to the standard MLE studied in this paper.

The literature developed from the TU model, to the TU model with fixed effects, to the NTU model with fixed effects (see Figure 1). To the best of our knowledge, no one has studied the NTU model without fixed effects given in Section 1.2. Although unobserved individual heterogeneity is important in many applications, the obvious advantage of ignoring it is that we can avoid the incidental parameter problem and use maximum likelihood estimation, which is considerably easier to implement than the estimation procedures in \cite{gao_logical_2023}, \cite{chen_moment_2026}, and \cite{li_bagging_2026}. 

\begin{figure}[H]
\centering
\caption{Visualizing the Gap in the Literature}
\begin{tikzpicture}
    \path (0,3) node(a) [rectangle, rounded corners, minimum width=3cm, minimum height=2cm, text centered, draw=black, fill=white!30, align=center] {
        \textbf{TU + FE} \\ \cite{graham_econometric_2017} \\ etc.
    };
    \path (0,0) node(b) [rectangle, rounded corners, minimum width=3cm, minimum height=2cm, text centered, draw=black, fill=white!30] {\textbf{TU}};
    \path (9,3) node(c) [rectangle, rounded corners, minimum width=3cm, minimum height=2cm, text centered, draw=black, fill=white!30, align=center] {
        \textbf{NTU + FE} \\ \cite{gao_logical_2023} \\ \cite{chen_moment_2026} \\ \cite{li_bagging_2026}};
    \path (9,0) node(d) [rectangle, rounded corners, minimum width=3cm, minimum height=2cm, text centered, draw=black, fill=white!30, align=center] {
        \textbf{NTU} \\ This paper};

    \draw [teal, line width=1mm,->,>=stealth] (b.north) -- (a.south);
    \draw [teal, line width=1mm,->,>=stealth] (a.east) -- (c.west);
    \draw [dashed, gray, line width=1mm,->,>=stealth] (d.north) -- (c.south);
    \draw [dashed, gray, line width=1mm,->,>=stealth] (b.east) -- (d.west);
\end{tikzpicture}
\fnote{The teal line in the figure depicts the methodological progression that the econometrics literature on dyadic network formation has followed. The gray dashed line depicts the route that this paper will take by considering a non-transferable utility model without unobserved individual heterogeneity.}
\end{figure}

Going directly from the TU model (without fixed effects) to the NTU model (without fixed effects) also reveals some novel theoretical insights. In particular, in Section 2, we establish sufficient conditions for identification of the NTU model (Proposition 1) and for consistency and asymptotic normality of the NTU-MLE (Theorems 1 and 2). Relative to \cite{gao_logical_2023}, \cite{chen_moment_2026}, and \cite{li_bagging_2026}, these sufficient conditions are easier to interpret. And they make clear that relatively little extra needs to be assumed compared with the TU model. The main difference is that identification is trickier to establish in the NTU model, which is related to the work on partial observability by \cite{poirier_partial_1980}. Proposition 1 provides easy-to-check conditions that guarantee identification. For example, a common design that includes a constant, one of $i$'s discrete characteristics, $X_i$, and the absolute difference, $|X_i - X_j|$, is identified if the support of $X_i$ is $\{0,1,2\}$, but not if $\{0,1\}$.

The second contribution of this paper is to provide a specification test for whether a given dataset is better explained by the TU or NTU model;\footnote{\cite{li_bagging_2026} discuss what happens if, say, one believes that the errors follow a logistic distribution when in fact they follow a normal distribution. This form of misspecification is interesting, but arguably not as interesting as if one believes that the TU model best explains link formation when in fact the NTU model does. See also \cite{ding_testing_2026} for testing parametric specifications in this context.} this is the focus of Section 3. The main difficulty is that the error correlation parameter $\rho$ is not identified in the TU model. We provide an LM specification test that accounts for this by taking the supremum over the support of the un-identified parameter in the spirit of \cite{davies_hypothesis_1977} and \cite{andrews_tests_1993}. One further difficulty is that the asymptotic distribution of this test statistic is not pivotal. We overcome this via a multiplier bootstrap procedure in the spirit of \cite{hansen_inference_1996}. Theorem 3 provides conditions for the correct asymptotic size of the test; again, these conditions are easy to interpret. We also provide simulation-based evidence for the strong power of the test.

In Section 4, we apply the estimation and testing procedures developed in this paper to Indian village network data from \cite{banerjee_changes_2024}. \cite{banerjee_changes_2024} collect data from 75 villages and measure social networks on several dimensions. For their dimension based on the question ``who would come to you if he/she needed to borrow kerosene or rice?'', the specification test from Section 3 rejects the TU model in 20\% of villages. Whereas for their dimension based on the prompt ``name the four non-relatives whom you speak to the most'', we reject the TU model in 32\% of villages. This is in line with the expectation that TU is a worse model of friendship network formation than financial network formation. In this application, using the maximum likelihood estimator that assumes transferable utility (TU-MLE) instead of the NTU-MLE when the TU assumption is rejected can change the sign and significance of some of the estimates. For example, the TU-MLE estimates suggest the quality of $j$'s roof (as a proxy for wealth) matters for link formation, while the NTU-MLE estimates suggest no significant effect. Without a testing procedure for TU vs NTU, applied researchers are at risk of similarly applying TU-model-based estimators to NTU-model-based data incorrectly, and vice versa.

Finally, Section 5 concludes. In the Appendix, we provide the proofs of all theoretical results and further simulation-based evidence.

\section{Properties of the NTU-MLE}

Consider again the NTU model \eqref{NTU Model} that was introduced in detail in Section 1.2:\footnote{Now writing $Y_{ij} = Y_{ij}^{\text{NTU}}$ for notational convenience.}
\begin{equation*}
Y_{ij} = \mathds{1} [W_{ij}^{'}\beta \geq \varepsilon_{ij}] \cdot \mathds{1} [W_{ji}^{'}\beta \geq \varepsilon_{ji}],
\qquad
\begin{pmatrix}
\varepsilon_{ij} \\
\varepsilon_{ji}
\end{pmatrix}
\sim \mathcal{N} \left(
\begin{pmatrix}
0 \\
0
\end{pmatrix},
\begin{pmatrix}
1 & \rho \\
\rho & 1
\end{pmatrix}
\right),
\end{equation*}
with $\rho \in (-1, 1)$, $W_{ij} \coloneq w(X_i, X_j) \in \mathbb{R}^{k}$, $W_{ji} \coloneq w(X_j, X_i) \in \mathbb{R}^{k}$, and $\beta \in \Theta_{\beta} \subseteq \mathbb{R}^k$. To place structure on the parameter space below, write $\Theta \coloneq \Theta_{\beta} \times [-1 + \delta_{\rho}, 1 - \delta_{\rho}]$ with $\delta_{\rho} > 0$.

Also recall the corresponding log-likelihood for the NTU model:
\[
\begin{aligned}
    \hat{Q}_n(\beta, \rho) &\coloneq \frac{1}{N}\sum^{n-1}_{i=1} \sum^{n}_{j=i+1} \Big[Y_{ij} \log \Phi_2^{\rho}\big(W_{ij}^{'}\beta, W_{ji}^{'}\beta\big) + (1 - Y_{ij})\log \big(1 - \Phi_2^{\rho}\big(W_{ij}^{'}\beta, W_{ji}^{'}\beta\big)\big) \Big],
\end{aligned}
\]
where $\Phi_2^{\rho}$ is the bivariate normal CDF with covariance $\rho$, the effective sample size is $N = \binom{n}{2} = \frac{n(n-1)}{2}$, and $n$ is the number of individuals in the network. The NTU-MLE is given by:
$$
\hat{\theta} = (\hat{\beta}, \hat{\rho}) \coloneq \argmax_{(\beta, \rho) \in \Theta} \hat{Q}_n(\beta, \rho).
$$

We now present the first contribution of this paper: conditions for identification of $\theta = (\beta, \rho)$ and for consistency and asymptotic normality of the NTU-MLE, $\hat{\theta}$.

\subsection{Identification}

Let $f(y \mid X_i, X_j; \theta)$ denote the conditional probability mass function of $Y_{ij}$ given $(X_i, X_j)$ in the NTU model. That is:
$$
f(y \mid X_i, X_j; \theta) = \Phi_2^\rho(W_{ij}'\beta, W_{ji}'\beta)^{y} \left[1 - \Phi_2^\rho(W_{ij}'\beta, W_{ji}'\beta)\right]^{1-y}, \quad y \in \{0, 1\}.
$$
We say that $\theta = (\beta, \rho)$ is \textit{identified} in the NTU model if, for every $\tilde{\theta} = (\tilde{\beta}, \tilde{\rho}) \in \Theta$ with $\tilde{\theta} \neq \theta$:
$$
\mathbb{P}\left[f(Y_{ij} \mid X_i, X_j; \tilde{\theta}) \neq f(Y_{ij} \mid X_i, X_j; \theta)\right] > 0.
$$

Consider first the case of discrete regressors. Let $S\coloneq\left\{ \big(w(x, x'),\, w(x', x)\big) : x, x' \in \operatorname{supp}(X_i) \right\}$ be the support of the regressors $(W_{ij}, W_{ji})$. Since $Y_{ij}$ is binary and every point of $S$ occurs with positive probability, $\theta$ is identified if and only if, for every $(\tilde{\beta}, \tilde{\rho}) \in \Theta$:
\begin{align*}
    &\Phi_2^\rho(w_1'\beta, w_2'\beta) = \Phi_2^{\tilde{\rho}}(w_1'\tilde{\beta}, w_2'\tilde{\beta}) \ \text{for all } (w_1, w_2) \in S \\
    &\implies (\beta, \rho) = (\tilde{\beta}, \tilde{\rho}).
\end{align*}
Identification is not trivial because $\Phi_2^\rho(w_1'\beta, w_2'\beta) = \Phi_2^{\tilde{\rho}}(w_1'\tilde{\beta}, w_2'\tilde{\beta})$ does not imply $(w_1'\beta, w_2'\beta) = (w_1'\tilde{\beta}, w_2'\tilde{\beta})$, nor $\rho = \tilde{\rho}$.\footnote{This is in contrast with the probit model where the univariate standard normal CDF, $\Phi$, is strictly increasing, so $\Phi(w'\beta) = \Phi(w'\tilde{\beta})$ implies $w'\beta = w'\tilde{\beta}$, and therefore a simple nonsingularity assumption on the regressors buys identification \parencite{newey_chapter_1994}.} However, it is the case that $\Phi_2^\rho(w_1'\beta, w_1'\beta) = \Phi_2^{\rho}(w_1'\tilde{\beta}, w_1'\tilde{\beta})$ implies $w_1'\beta = w_1'\tilde{\beta}$. My identification argument uses this fact, in combination with the existence of regressor values at which symmetric dyads $W_{ij} = W_{ji}$ occur, to identify $\rho$.  Then, given $\rho$, $\beta$ is identified because $\Phi_2^\rho$ is strictly increasing in each of its arguments.\footnote{This is related to the work on partial observability by \cite{poirier_partial_1980}. \cite{poirier_partial_1980} conjectures that the key to identification is allowing the regressors to take on more values than there are unknown parameters, but he does not provide a proof of this claim. Our approach also relies on assumptions about the support of the regressors.}

Formally, let $D \coloneq \{w: (w,w) \in S \}$ be the diagonal support, i.e. the regressor values at which symmetric dyads $W_{ij} = W_{ji}$ occur. We make the following three assumptions for identification:

\renewcommand{\theassumption}{I\arabic{assumption}}

\begin{assumption}[NTU Model]\label{ass:I1}
The model is the NTU model given in \eqref{NTU Model}. The error pairs $(\varepsilon_{ij}, \varepsilon_{ji})_{i<j}$ are i.i.d.\ across dyads and independent of $(X_i)_{i=1}^{n}$, which are themselves i.i.d. across individuals.
\end{assumption}

\begin{assumption}[Non-Cancellation]\label{ass:I2}
For every $b \in \mathbb{R}^{k} \setminus \{0\}$, there exists $(w_1, w_2) \in S$ such that $(w_1'b)(w_2'b) \geq 0$ and $(w_1'b, w_2'b) \neq (0,0)$.
\end{assumption}

\begin{assumption}[Symmetric-Dyad Variation]\label{ass:I3}
At least one of the following holds:
\begin{enumerate}
    \item[(a)] There exist $w_1, w_2 \in D$ with $(w_1, w_2) \in S$ and $\Phi_2^{\rho}(w_1'\beta, w_1'\beta) \neq \Phi_2^{\rho}(w_2'\beta, w_2'\beta)$,
    \item[(b)] There exist $w_0, w_1, w_2 \in D$ with $w_1 = \tfrac{1}{2}(w_0 + w_2)$ and
$\Phi_2^{\rho}(w_0'\beta, w_0'\beta) \neq \Phi_2^{\rho}(w_2'\beta, w_2'\beta)$.
\end{enumerate}

\end{assumption}

\begin{prop}[Identification]
Under \ref{ass:I1}-\ref{ass:I3}, $\theta = (\beta, \rho)$ is identified.
\end{prop}

\begin{proof}
    See Appendix. \renewcommand{\qedsymbol}{}
\end{proof}

\ref{ass:I1} is simply a re-statement of the NTU model.\footnote{Independence across dyads and across nodes is not actually used in the proof of Proposition 1, which only uses that $(\varepsilon_{ij}, \varepsilon_{ji})$ is independent of $(X_i, X_j)$ and that the link probability and the support $S$ are common to all dyads. But independence is used heavily in the consistency and asymptotic normality theorems (Theorems 1 and 2). We include it in \ref{ass:I1} here to avoid re-stating the NTU model assumption again in later sections.} In words, \ref{ass:I2} says that there is no way of changing $\beta$ such that, at every dyad, one member's index increases while the other's falls. \ref{ass:I2} is used to identify $\beta$ given $\rho$. \ref{ass:I3}(a) says that there are two symmetric regressor values that give different link probabilities, and their mixed dyad is observed. \ref{ass:I3}(b) says that there are two symmetric regressor values that give different link probabilities, and they are linearly related via a third value. Either \ref{ass:I3}(a) or \ref{ass:I3}(b) is used to identify $\rho$. 

To make these assumptions more concrete, Table \ref{tab:prop1_assumptions} presents which assumptions hold and which fail for different common regressor specification/designs. The simplest possible design includes a constant and a single binary characteristic of $i$'s, i.e. $w(X_i, X_j) = (1, X_i)'$ with $\operatorname{supp}(X_i) = \{0,1\}$; this is design 1 in Table \ref{tab:prop1_assumptions}. In this case:
$$
S = \left\{ \big((1, x)',\, (1, x')'\big) : x, x' \in \{0, 1\} \right\}, \qquad D = \left\{ (1, 0)',\, (1, 1)' \right\}.
$$
To show that \ref{ass:I2} is satisfied, take a symmetric point $(w,w)\in S$ such that $(w_1'b)(w_2'b) = (w'b)^2 \geq 0$; this product is zero at the two symmetric points only if $b = 0$. To show that \ref{ass:I3}(a) is satisfied, simply note that both $(1, 0)'$ and $(1, 1)'$ are in $D$ and the mixed pair $\big((1, 0)', (1, 1)' \big)$ is in $S$. \ref{ass:I3}(b) fails because $D$ only has two elements, so there is no midpoint triple. This does not matter because we already have identification in this case via Proposition 1 given \ref{ass:I2} and \ref{ass:I3}(a). If instead $\operatorname{supp}(X_i) = \{0,1,2\}$, then there is a midpoint triple and \ref{ass:I3}(b) is satisfied too; this is design 2 in Table \ref{tab:prop1_assumptions}. Now suppose one wishes to include $j$'s discrete characteristic too such that the design becomes $w(X_i, X_j) = (1, X_i, X_j)'$ with $\operatorname{supp}(X_i) = \{0,1,2\}$; this is design 3 in Table \ref{tab:prop1_assumptions}. \ref{ass:I2} fails for this design because swapping the coefficients on own and partner characteristics, $X_i$ and $X_j$, swaps the two members' indices at every dyad, and so leaves every link probability unchanged. Finally, consider adding the absolute difference between $i$ and $j$'s discrete characteristic such that the design becomes $w(X_i, X_j) = (1, X_i, |X_i - X_j|)'$ with $\operatorname{supp}(X_i) = \{0,1,2\}$; this is design 4 in Table \ref{tab:prop1_assumptions}. In this case:
$$
D = \left\{ (1, 0, 0)',\, (1, 1, 0)',\, (1, 2, 0)' \right\}.
$$
Then, since $(1, 1, 0)' = \frac{1}{2}[(1, 0, 0)' + (1, 2, 0)']$, \ref{ass:I3}(b) is satisfied. Designs such as design 4 are where \ref{ass:I3}(b) is particularly useful. This is fortunate because many designs in practice, such as the design considered in the application in Section 4, feature an individual's characteristics as well as similarity terms.

\begin{table}[H]
\centering
\small
\caption{Proposition 1 Assumptions for Different Designs}
\label{tab:prop1_assumptions}
\begin{threeparttable}
\renewcommand{\arraystretch}{1.25}
\begin{tabularx}{\textwidth}{@{}l l c c c c >{\raggedright\arraybackslash}X@{}}
\toprule
& Design $w(X_i,X_j)$ & Support of $X_i$ & I2 & I3(a) & I3(b) & Reason for failure \\
\midrule
1 & $(1,\,X_i)'$ & $\{0,1\}$
  & \cmark & \cmark & \xmark
  & I3(b): $D$ has only two elements, so no midpoint triple exists. \\
2 & $(1,\,X_i)'$ & $\{0,1,2\}$
  & \cmark & \cmark & \cmark
  & --- \\
3 & $(1,\,X_i,\,X_j)'$ & $\{0,1,2\}$
  & \xmark & \xmark & \cmark
  & I2: $b=(0,-1,1)'$ reverses the two indices at every dyad (own/partner swap).\newline
    I3(a): no mixed dyad has both vectors in $D$. \\
4 & $(1,\,X_i,\,|X_i-X_j|)'$ & $\{0,1,2\}$
  & \cmark & \xmark & \cmark
  & I3(a): the third regressor is $0$ on the diagonal and positive off it. \\
\bottomrule
\end{tabularx}
\begin{tablenotes}[flushleft]
  \footnotesize
  \item Assumptions I2 and I3 across four designs, assuming the coefficient on $X_i$ is nonzero. Design 2 satisfies every condition. Design 3 is not identified: the own- and partner-coefficients may be interchanged without changing any link probability. Designs 1 and 4 each satisfy exactly one of the two forms of I3, so each is covered by Proposition 1 through a different route.
\end{tablenotes}
\end{threeparttable}
\end{table}

With continuous regressors, since \ref{ass:I2} and \ref{ass:I3} are conditions on the support $S$, Proposition 1 applies unchanged as long as the definition of $S$ is altered to be the closure of $\left\{ \big(w(x, x'),\, w(x', x)\big) : x, x' \in \operatorname{supp}(X_i) \right\}$. When the regressors include continuous own (or partner) characteristics, symmetric dyads can occur with probability zero, but the corresponding points remain in $S$. For example, if $X_i$ has an interval as its support in design 1 or 4, then \ref{ass:I3}(b) holds automatically. Finally, note that continuity does not rescue design 3: swapping the own- and partner-coefficients leaves every link probability unchanged regardless of the support. 

\subsection{Consistency}

With identification covered by Proposition 1, the next step is to discuss consistency. Let $\theta_0 = (\beta_0, \rho_0)$ be the true value of $\theta = (\beta, \rho)$. The following standard assumptions guarantee consistency:

\renewcommand{\theassumption}{C\arabic{assumption}}
\setcounter{assumption}{0}

\begin{assumption}[Parameter Space]\label{ass:C1}
$\Theta \coloneq \Theta_{\beta} \times [-1 + \delta_{\rho}, 1 - \delta_{\rho}]$ with $\delta_{\rho} > 0$ is compact. And $\theta_0 \in \text{interior}(\Theta)$.
\end{assumption}

\begin{assumption}[Identification]\label{ass:C2}
For every $\theta \in \Theta$ with $\theta \neq \theta_0$:
$$
\mathbb{P}\left[f(Y_{ij} \mid X_i, X_j; \theta) \neq f(Y_{ij} \mid X_i, X_j; \theta_0)\right] > 0.
$$
\end{assumption}

\begin{assumption}[Moments of Order $2 + \delta$]\label{ass:C3}
$\mathbb{E}\| W_{ij}\|^{2 + \delta} < \infty$ for some $\delta > 0$.
\end{assumption}

\begin{theorem}[Consistency]
    Under \ref{ass:I1} and \ref{ass:C1}-\ref{ass:C3}, $\hat{\theta} \xrightarrow{p} \theta_0$ as $n \to \infty$.
\end{theorem}

\begin{proof}
    See Appendix. The proof adapts existing theory to this dyadic setting. Specifically, we prove a uniform (strong) law of large numbers (U-SLLN) that applies to several objects in this paper (Lemma 1). This builds on U-statistic theory in \cite{serfling_approximation_1980}. \renewcommand{\qedsymbol}{}
\end{proof}

\ref{ass:C1} is the usual compactness plus true parameter in the interior assumption. The identification assumption \ref{ass:C2} is of course implied by Proposition 1, and therefore by \ref{ass:I1}-\ref{ass:I3}. We include identification as its own assumption here in case researchers want to verify identification via an alternative result to Proposition 1. \ref{ass:C3} is a typical slightly-stronger-than-second-moments assumption.

\subsection{Asymptotic Normality}

The final property of the NTU-MLE to discuss is asymptotic normality. Beyond the assumptions for consistency, we make two further standard assumptions:

\renewcommand{\theassumption}{N\arabic{assumption}}
\setcounter{assumption}{0}

\begin{assumption}[Moments of Order $4 + \delta$]\label{ass:N1}
$\mathbb{E}\| W_{ij}\|^{4 + \delta} < \infty$ for some $\delta > 0$.
\end{assumption}

\begin{assumption}[Information]\label{ass:N2}
$J \coloneq \mathbb{E} \left[\frac{\nabla_{\theta}\Phi_2^{\rho_0}\big(W_{ij}^{'}\beta_0, W_{ji}^{'}\beta_0\big)\left(\nabla_{\theta}\Phi_2^{\rho_0}\big(W_{ij}^{'}\beta_0, W_{ji}^{'}\beta_0\big)\right)'}{\Phi_2^{\rho_0}\big(W_{ij}^{'}\beta_0, W_{ji}^{'}\beta_0\big)(1-\Phi_2^{\rho_0}\big(W_{ij}^{'}\beta_0, W_{ji}^{'}\beta_0\big))} \right]$ is nonsingular.
\end{assumption}

\begin{theorem}[Asymptotic Normality]
    Under \ref{ass:I1}, \ref{ass:C1}-\ref{ass:C2}, and \ref{ass:N1}-\ref{ass:N2}, as $n \to \infty$ so that $N = \frac{n(n-1)}{2} \to \infty$: 
    \begin{enumerate}
        \item[(i)] $\sqrt{N}(\hat{\theta} - \theta_0) \xrightarrow{d} \mathcal{N}(0, J^{-1})$,
        \item[(ii)] $\hat{J} \coloneq \frac{1}{N}\sum^{n-1}_{i=1} \sum^{n}_{j=i+1} \frac{\nabla_{\theta}\Phi_2^{\hat{\rho}}\big(W_{ij}^{'}\hat{\beta}, W_{ji}^{'}\hat{\beta}\big)\left(\nabla_{\theta}\Phi_2^{\hat{\rho}}\big(W_{ij}^{'}\hat{\beta}, W_{ji}^{'}\hat{\beta}\big)\right)'}{\Phi_2^{\hat{\rho}}\big(W_{ij}^{'}\hat{\beta}, W_{ji}^{'}\hat{\beta}\big)(1-\Phi_2^{\hat{\rho}}\big(W_{ij}^{'}\hat{\beta}, W_{ji}^{'}\hat{\beta}\big))} \xrightarrow{p} J$.
    \end{enumerate} 
\end{theorem}

\begin{proof}
    See Appendix. Again, the proof adapts existing theory to this dyadic setting: Lemma 2 establishes a central limit theorem that applies to several objects in this paper.\renewcommand{\qedsymbol}{}
\end{proof}

Unsurprisingly, for asymptotic normality the slightly-stronger-than-second-moments assumption is upgraded to slightly-stronger-than-fourth-moments via \ref{ass:N1}. \ref{ass:N2} is the typical nonsingularity assumption made with MLE to ensure invertibility of the asymptotic variance. As well as asymptotic normality, Theorem 2 provides a consistent estimator of the (inverse of the) asymptotic variance, $\hat{J}$. 

In the proof of Theorem 3 in the Appendix, we prove analogous results to Theorems 1 and 2 for the TU-MLE in the TU model. The key difference between the sufficient conditions for consistency and asymptotic normality of the TU-MLE and the NTU-MLE is the identification assumption \ref{ass:C2}. However, as discussed in Section 2.1, identification can be established for many common designs via Proposition 1. That is, the NTU-MLE is not much harder to justify than the TU-MLE theoretically.

\section{Testing for Utility Transferability}

This section presents the second contribution of this paper: a specification test for whether a given dataset is better explained by the TU or NTU model. After recognizing that the TU and NTU models can be nested in a more general model, the test is mostly standard. One difficulty is that the error correlation parameter $\rho$ is not identified in the TU model. We account for this by taking the supremum over the support of the un-identified parameter in the spirit of \cite{davies_hypothesis_1977} and \cite{andrews_tests_1993}. This supremum-based test statistic has a non-pivotal asymptotic distribution, which we can simulate from via a multiplier bootstrap procedure in the spirit of \cite{hansen_inference_1996}.

Specifically, recall the TU model given in \eqref{TU Model} and the NTU model given in \eqref{NTU Model}. To test for TU vs NTU, it is useful to introduce a more general model that nests both the TU and NTU models as special cases: the Imperfect Transferable Utility (ITU) model.\footnote{The ITU model has scarcely featured in network econometrics, with footnote 12 in \cite{de_paula_econometric_2020} being a notable exception; the model has appeared in the matching literature, for example in \cite{galichon_costly_2019}.} The ITU model is given by:
$$
Y_{ij}^{\text{ITU}} = \mathds{1} [(W_{ij} + \tau W_{ji})^{'}\beta \geq \varepsilon_{ij} + \tau \varepsilon_{ji}] \cdot \mathds{1} [(\tau W_{ij} + W_{ji})^{'}\beta \geq \tau \varepsilon_{ij} + \varepsilon_{ji}],
$$
where $\tau \in [0,1]$. One can interpret $\tau$ in the ITU model as an exchange rate in that each unit of utility one individual transfers delivers $\tau$ units to the other. The TU model is ITU with $\tau = 1$ and the NTU model is ITU with $\tau = 0$. Therefore, to test the suitability of the TU model, we can test: 
$$
H_0 : \tau = 1 \quad \text{vs} \quad H_A : \tau < 1.
$$

Estimation is notably easier under the null. Specifically, the TU model can be re-written as:
\begin{equation*}
    Y_{ij}^{\text{TU}} = \mathds{1} [V_{ij}^{'}\gamma \geq \nu_{ij}],
\end{equation*}
with $V_{ij} \coloneq W_{ij} + W_{ji} \in \mathbb{R}^{k}$, $\gamma \coloneq \frac{\beta}{\sqrt{2(1 + \rho)}} \in \Theta_\gamma \subseteq \mathbb{R}^{k}$, and $\nu_{ij} \coloneq \frac{\varepsilon_{ij} + \varepsilon_{ji}}{\sqrt{2(1 + \rho)}} \sim \mathcal{N}(0,1)$.\footnote{Recall from Section 1.3 that one reason to assume normal errors is so that linear combinations of them remain normal, as is the case for $\nu_{ij}$ here.} This is the classic probit model, but with dyadic data. Analogous to the NTU-MLE, we have the following log-likelihood function for the TU model:
$$
\hat{Q}_n(\gamma) \coloneq \frac{1}{N}\sum^{n-1}_{i=1} \sum^{n}_{j=i+1} \Big[Y_{ij}^{\text{TU}} \log \big(\Phi(V_{ij}^{'}\gamma) \big) + (1 - Y_{ij}^{\text{TU}})\log \big(1 - \Phi(V_{ij}^{'}\gamma) \big) \Big],
$$
where $\Phi$ is the univariate standard normal CDF. The TU-MLE is given by:
$$
\hat{\gamma} \coloneq \argmax_{\gamma \in \Theta_\gamma} \hat{Q}_n(\gamma).
$$
In the proof of Theorem 3, we show that $\gamma$ is identified and $\sqrt{N}$-consistently estimated by the TU-MLE, $\hat{\gamma}$, under the standard probit assumptions but with finite fourth moment instead of second moments due to the dyadic dependence. Whereas under the alternative, for example if $\tau = 0$, one would need to impose the assumptions for consistency and asymptotic normality of the NTU-MLE from Section 2. Furthermore, since $\tau = 1$ lies on the boundary of its parameter space, Wald and likelihood ratio tests would have non-standard limiting distributions \parencite{andrews_testing_2001}. A score test that only requires estimation under the null avoids this boundary problem. For these reasons, we propose a score/LM test.

\subsection{sup-LM Test}

Consider the following \textit{infeasible} LM test statistic:\footnote{Now writing $Y_{ij} = Y_{ij}^{\text{TU}}$ for notational convenience.}
$$
t_N(c) \coloneq - \frac{\frac{1}{\sqrt{N}} \sum^{n-1}_{i=1} \sum^{n}_{j=i+1} \hat{e}_{ij} \hat{r}_{ij}(c)}{\sqrt{\frac{1}{N} \sum^{n-1}_{i=1} \sum^{n}_{j=i+1} \hat{\omega}_{ij} \hat{r}_{ij}(c)^2}},
$$
where:
\begin{alignat*}{2}
    &\hat{e}_{ij} \coloneq \phi(V_{ij}^{'}\hat{\gamma}) \frac{Y_{ij} - \Phi(V_{ij}^{'}\hat{\gamma})}{\Phi(V_{ij}^{'}\hat{\gamma})(1 - \Phi(V_{ij}^{'}\hat{\gamma}))}, \qquad
    &&\hat{\omega}_{ij} \coloneq \frac{\phi(V_{ij}^{'}\hat{\gamma})^2}{\Phi(V_{ij}^{'}\hat{\gamma})(1 - \Phi(V_{ij}^{'}\hat{\gamma}))}, \\
    &\hat{r}_{ij}(c) \coloneq g(c, \hat{d}_{ij}) - \hat{\Pi}(c)'V_{ij},
    &&g(c, \hat{d}_{ij}) \coloneq c\, \psi\left(\frac{\hat{d}_{ij}}{c}\right), \\
    &\psi(x) \coloneq \frac{e^{-x^2}}{\sqrt{\pi}} + x\, \text{erf}(x),
    &&\hat{d}_{ij} \coloneq (W_{ij} - W_{ji})^{'}\hat{\gamma}, \\
    &\hat{\Pi}(c) \coloneq \left(\sum^{n-1}_{i=1} \sum^{n}_{j=i+1} \hat{\omega}_{ij}V_{ij}V_{ij}^{'}\right)^{-1}\left(\sum^{n-1}_{i=1} \sum^{n}_{j=i+1} \hat{\omega}_{ij}V_{ij}\,g(c, \hat{d}_{ij})\right), \span\span \\
    &c \coloneq \sqrt{\frac{2(1-\rho)}{1+\rho}} \in \mathcal{C},
\end{alignat*}
with $\phi$ being the univariate standard normal PDF and $\text{erf}$ being the Gaussian error function.\footnote{Note that the factor of $c$ in front of $\psi$ can be canceled from the numerator and denominator of $t_N(c)$. The factor exists to a) ensure a cleaner Lipschitz bound in the proof of Theorem \ref{thm:size} and b) attach a bandwidth interpretation to $c$.}

This test statistic is infeasible because, although $\gamma = \frac{\beta}{\sqrt{2(1 + \rho)}}$ is identified under the null, $\beta$ and $\rho$ are not individually identified, so $c = \sqrt{\frac{2(1-\rho)}{1+\rho}}$ is not identified. This means, although one can obtain a consistent estimate $\hat{d}_{ij} = (W_{ij} - W_{ji})^{'}\hat{\gamma}$, one cannot obtain a consistent estimate $g(c, \hat{d}_{ij})$ via the plug-in method.

To generate a feasible test statistic, we follow the suggestion in \cite{davies_hypothesis_1977} and take the supremum of the infeasible statistic over the un-identified parameter, $c$.\footnote{Alternatively, one could average over $c$ instead of taking the supremum as in the exponential-average tests of \cite{andrews_optimal_1994}.} Namely, the \textit{feasible} LM test statistic we propose is:
$$
t_N^{\text{sup}} \coloneq \sup_{c \in \mathcal{C}} t_N(c) = \sup_{c \in \mathcal{C}} \left\{ -\frac{\frac{1}{\sqrt{N}} \sum^{n-1}_{i=1} \sum^{n}_{j=i+1} \hat{e}_{ij} \hat{r}_{ij}(c)}{\sqrt{\frac{1}{N} \sum^{n-1}_{i=1} \sum^{n}_{j=i+1} \hat{\omega}_{ij} \hat{r}_{ij}(c)^2}}\right\}.
$$

One further difficulty is that the asymptotic distribution of this test statistic is not pivotal: it depends on the correlation between $r_{ij}(c_1)$ and $r_{ij}(c_2)$ for generic $c_1, c_2 \in \mathcal{C}$, which depends on the distribution of $V_{ij}$ and $d_{ij}$, i.e. the data. For this reason, we perform a bootstrap procedure to simulate the critical value given the data. Specifically, we perform a multiplier bootstrap procedure, as used for example in \cite{hansen_inference_1996}.\footnote{We opt for a multiplier bootstrap over a resampling bootstrap because resampling nodes or dyads destroys the network structure.} Independently of the data, let $\{\xi_{ij}\}_{i<j}$ be i.i.d.\ $\mathcal{N}(0,1)$ and define the multiplier statistic:
$$
t_N^{*\text{sup}} \coloneqq \sup_{c \in \mathcal{C}} \left\{ -\frac{\frac{1}{\sqrt{N}} \sum^{n-1}_{i=1} \sum^{n}_{j=i+1} \xi_{ij} \hat{e}_{ij} \hat{r}_{ij}(c)}{\sqrt{\frac{1}{N} \sum^{n-1}_{i=1} \sum^{n}_{j=i+1} \hat{\omega}_{ij} \hat{r}_{ij}(c)^2}} \right\}.
$$
Let $c_{1-\alpha}^{*}$ denote the $(1-\alpha)$-th quantile of the distribution of $t_N^{*\text{sup}}$ conditional on the data. The multiplier bootstrap test rejects $H_0$ if $t_N^{\text{sup}} > c_{1-\alpha}^{*}$. In practice, $c_{1-\alpha}^{*}$ is computed by simulation:\footnote{As always, the simulation error is immaterial provided the number of simulations $B \to \infty$.}
\begin{enumerate}
    \item For each $b = 1, \dots, B$, draw $\{\xi_{ij}^{(b)}\}_{i<j}$ i.i.d.\ $\mathcal{N}(0,1)$, independently across $b$ and independently of the data.
    \item For each $b$, compute $t_N^{*(b)}$, defined as $t_N^{*\text{sup}}$ above with $\xi_{ij}$ replaced by $\xi_{ij}^{(b)}$.
    \item Take $\hat{c}_{1-\alpha}$, the empirical $(1-\alpha)$-th quantile of $\{t_N^{*(b)}\}_{b=1}^{B}$, in place of $c_{1-\alpha}^{*}$. I.e. reject $H_0$ if $t_N^{\text{sup}} > \hat{c}_{1-\alpha}$. 
\end{enumerate}

\subsection{Size}

To establish a size theorem, we make the following assumptions, noting that an object without a hat means $\hat{\gamma}$ has been replaced with the true value $\gamma_0$, for example $\omega_{ij} \coloneq  \frac{\phi(V_{ij}^{'}\gamma_0)^2}{\Phi(V_{ij}^{'}\gamma_0)(1 - \Phi(V_{ij}^{'}\gamma_0))}$.

\renewcommand{\theassumption}{T\arabic{assumption}}
\setcounter{assumption}{0}

\begin{assumption}[TU Model]\label{ass:T1}
The null model is the TU model given in \eqref{TU Model}. The error pairs $(\varepsilon_{ij}, \varepsilon_{ji})_{i<j}$ are i.i.d.\ across dyads and independent of $(X_i)_{i=1}^{n}$, which are themselves i.i.d. across individuals.
\end{assumption}

\begin{assumption}[Identification]\label{ass:T2} $\mathbb{E}[\omega_{ij} V_{ij} V_{ij}']$ is nonsingular.
\end{assumption}

\begin{assumption}[Fourth Moments]\label{ass:T3} $\mathbb{E}\| W_{ij}\|^4  < \infty$.
\end{assumption}

\begin{assumption}[Test Relevance]\label{ass:T4} $\inf_{c \in \mathcal{C}} \mathbb{E}[\omega_{ij} r_{ij}(c)^2] > 0$.
\end{assumption}

\begin{assumption}[Compact Search Set]\label{ass:T5} $\mathcal{C} = [\ubar{C}, \bar{C}]$ with $0 < \ubar{C} \leq \bar{C} < \infty$. Equivalently, $\rho \in [\ubar{\rho}, \bar{\rho}]$ with $-1 < \ubar{\rho} \leq \bar{\rho} < 1$.
\end{assumption}

Under these assumptions, the collection of random variables $\{t_N(c) : c \in \mathcal{C}\}$ converges in distribution to a Gaussian process (Theorem \ref{thm:size}(i)) and  $t_N^{\text{sup}}$ converges in distribution to the supremum over $c$ of this Gaussian process (Theorem \ref{thm:size}(ii)). Theorem \ref{thm:size}(iii) guarantees that the bootstrap procedure outlined above will have correct asymptotic size, provided $B$ is taken to be large.

\begin{theorem}[Size]\label{thm:size}
Suppose assumptions \ref{ass:T1}-\ref{ass:T5} hold and let $C(\mathcal{C})$ be the space of continuous real-valued functions on $\mathcal{C}$ equipped with the supremum norm. Then, as $n \to \infty$:
\begin{enumerate}
  \item[(i)] $\{t_N(c) : c \in \mathcal{C}\} \overset{d}{\to} \mathcal{G}$ in $C(\mathcal{C})$, where $\mathcal{G}$ is the mean-zero Gaussian process with unit variance, covariances $\frac{\mathbb{E}[\omega_{ij} r_{ij}(c_1) r_{ij}(c_2)]}{\sqrt{\mathbb{E}[\omega_{ij} r_{ij}(c_1)^2]\mathbb{E}[\omega_{ij} r_{ij}(c_2)^2]}}$, and continuous sample paths.
  \item[(ii)] $t_N^{\text{sup}} \overset{d}{\to} \sup_{c \in \mathcal{C}}\mathcal{G}(c)$, whose distribution function is continuous and strictly increasing on $\mathbb{R}$.
  \item[(iii)] Conditional on the data, $t_N^{*\text{sup}} \overset{d}{\to} \sup_{c \in \mathcal{C}}\mathcal{G}(c)$ in probability. Thus, $c_{1-\alpha}^* \overset{p}{\to} c_{1-\alpha}$ where $c_{1-\alpha}$ is the $(1-\alpha)$-th quantile of the distribution of $\sup_{c \in \mathcal{C}}\mathcal{G}(c)$. Finally, $\mathbb{P}[t_N^{\text{sup}} > c_{1-\alpha}^*] \to \alpha$.
\end{enumerate}
\end{theorem}

\begin{proof}
    See Appendix. As elsewhere in this paper, the proof adapts existing theory to this dyadic setting. For example, we condition on the covariates to apply the functional CLT in \cite{van_der_vaart_weak_1996} and use Lemma 1 repeatedly.\footnote{Note that Lemma 1 assumes \ref{ass:I1} instead of \ref{ass:T1}. An alternative version of Lemma 1 with \ref{ass:I1} replaced with \ref{ass:T1} is immediate; the proof relies on the conditional-independence that is common to both assumptions.} \renewcommand{\qedsymbol}{}
\end{proof}

Assumptions \ref{ass:T1}-\ref{ass:T3} are standard sufficient conditions for $\sqrt{N}$-consistency of the TU-MLE/probit estimator, $\hat{\gamma}$, to the true value, $\gamma_0$.\footnote{We argue this in the proof of Theorem \ref{thm:size} in the Appendix.} Assumption \ref{ass:T4} ensures that the limit of the denominator of $t_N(c)$ is well-defined for all $c$. A useful sufficient condition for \ref{ass:T4} is that $\mathbb{V}ar[|d_{ij}| \mid V_{ij}] > 0$ with positive probability. This condition is verifiable for some designs. For example, if $W_{ij} = (1, aX_i + bX_j)'$, then $V_{ij} = (2, (a+b)(X_i + X_j))'$ and $d_{ij}$ is proportional to $(a-b)(X_i - X_j)$. Therefore, if $a=b$, i.e. the regressors are symmetric, then knowing $V_{ij}$ would imply knowing $d_{ij}$. Or, if $X_i$ is binary, then knowing $V_{ij}$ means knowing $X_i + X_j$. If $X_i + X_j$ is $0$ or $2$, then $X_i = X_j$, so $d_{ij} = 0$. If $X_i + X_j = 1$, then $|d_{ij}|$ is proportional to $|a - b|$. Again, knowing $V_{ij}$ would imply knowing $d_{ij}$. That is, symmetric regressors or insufficient support would violate \ref{ass:T4} in this design; this is closely related to the assumptions imposed on the support in the identification section. Finally, assumption \ref{ass:T5} defines the set over which the supremum is taken; this assumption is used during several analytical arguments in the proof of Theorem \ref{thm:size}. Unlike \ref{ass:T1}-\ref{ass:T4}, assumption \ref{ass:T5} is not a restriction on the DGP but a choice made by the researcher. The restriction $\bar{C} < \infty$ excludes values of $\rho$ near $-1$ and $\ubar{C} > 0$ excludes values of $\rho$ near $1$. Since neither endpoint is attainable in the model, no admissible value of $\rho$ is ruled out by the model itself. We report sensitivity to the choice of $[\ubar{\rho}, \bar{\rho}]$ in the simulation results.

\subsection{Power}

If $c$ (or $\rho$) was known, then the generically infeasible LM statistic $t_N(c)$ would become feasible. Furthermore, $t_N(c)$ is the \textit{locally} asymptotically uniformly most powerful test against the alternative $H_A : \tau < 1$.\footnote{The concept of local asymptotic optimality of a test is as in Addendum 15.5 of \cite{van_der_vaart_asymptotic_1998}. We do not formally prove that $t_N(c)$ is locally asymptotically UMP when $c$ is known because the whole point is that $c$ is not known in general. A proof would rely on a) recognizing that the optimal statistic in Addendum 15.5 is equal to $t_N(c)$ and b) adapting the optimality theorem (Theorem 15.4) to allow for a non-open parameter space and for dyadic dependence. For a), use that the score of the log-likelihood of the ITU model with respect to $\tau$ evaluated under the null ($\tau = 1$) is given by $\frac{1}{2}e_{ij} g(c, d_{ij}).$} Therefore, $t_N(c)$ provides an oracle power envelope against which one can compare the power of the feasible test statistic, $t_N^{\text{sup}}$. In general, $t_N^{\text{sup}}$ will be less powerful than $t_N(c)$ because the supremum trades off power for robustness (of the size) to the value of the un-identified parameter. The question becomes how much power is lost by using the sup-LM test instead of the infeasible/oracle LM test. We answer this question via simulation.

 Specifically, Figure \ref{fig:supLM_test_power_curves} plots power curves, i.e. the proportion of simulations for which the null $H_0: \tau = 1$ is rejected against the true value of $\tau$ used to generate the networks. We vary the error correlation $\rho \in \{-0.8, -0.4, 0, 0.4, 0.8\}$, indicated by different colored lines in the figure, as well as whether the sup-LM test or oracle test is used, indicated by solid vs dashed lines. All points in the figure are constructed by taking the average of rejection decisions at level $\alpha = 0.05$ across 2000 networks generated from the ITU model with the corresponding true values of $\tau$ and $\rho$, $\beta=1$, $n=100$ individuals, and a single asymmetric regressor $W_{ij} \neq W_{ji}$ drawn from a bivariate normal distribution with $\mathrm{Cov}(W_{ij}, W_{ji}) = 0.1$. For the supremum search set, we take $\ubar{\rho}=-0.9$, $\bar{\rho} = 0.9$, and 50 evenly spaced grid points. For the bootstrap, we take $B=999$ draws per test.

 Figure \ref{fig:supLM_test_power_curves} suggests that almost no power is lost by using the sup-LM test and therefore that the test is powerful. As expected, as $\tau$ departs from the null in the direction of the alternative, i.e. from TU to NTU, the power of the sup-LM test increases quickly. For example, the sup-LM test has 100\% power provided $\tau \leq 0.9$ when $\rho = -0.8$ in this simulation. The power of the test is decreasing in $\rho$ for this design.\footnote{This is not always the case. For designs with larger $\mathrm{Cov}(W_{ij}, W_{ji})$, i.e. more regressor symmetry, it is possible to generate a pattern where the power curves are non-monotonic in $\rho$.} Nonetheless, even for the largest value of $\rho$ we consider, the power is practically 100\% provided $\tau \leq 0.8$. In all cases, the oracle test hardly outperforms the sup-LM test. The worst case is when $\tau = 0.9$ and $\rho = 0.8$; at this point, the power of the oracle test is 0.630 and the power of the sup-LM test is 0.603. Provided $\tau \leq 0.8$, the powers of the two tests are indistinguishable in this simulation.

 Figure \ref{fig:supLM_test_power_curves} also empirically confirms the size result from Theorem \ref{thm:size}. The empirical sizes can be read off the figure by looking at the proportion of rejections when $\tau = 1$ (as indicated by the small dotted line). The true size is $\alpha = 0.05$; empirical sizes are 0.0515 when $\rho=-0.8$, 0.0580 when $\rho=-0.4$, 0.0485 when $\rho=0$, 0.0490 when $\rho=0.4$, and 0.0420 when $\rho=0.8$.

 In the Appendix, we show that the patterns presented in Figure \ref{fig:supLM_test_power_curves} are robust to the sample size and the choice of the supremum search set. For example, the oracle test still hardly outperforms the sup-LM test when the sample size is lowered to $n=50$. Of course, the power curves are flatter with this lower sample size, but power is still 100\% for all values of $\rho$ considered provided $\tau \leq 0.6$. The reason for this excellent finite-sample performance is that, due to the dyadic nature of the data, we have $\sqrt{N}$-convergence of the TU-MLE (and NTU-MLE as per Theorem 2) instead of $\sqrt{n}$-convergence. The number of dyads, $N = \frac{n(n-1)}{2}$, is quadratic in the number of individuals, $n$. So, although $n=50$ is small, $N=\frac{50(50-1)}{2}=1225$ is not. Regarding the supremum search set, shrinking to $[\ubar{\rho}, \bar{\rho}] = [-0.6, 0.6]$ seems to have little effect. The power discrepancy between the sup-LM and oracle tests shrinks slightly for $\rho \in \{-0.4, 0, 0.4\}$ and increases slightly for $\rho \in \{-0.8, 0.8\}$. The reason for this effect is that a narrower search set lowers the critical value, but risks missing the true value of $\rho$. The worst case discrepancy now has sup-LM test power equal to 0.599 (down from 0.603 when $[\ubar{\rho}, \bar{\rho}] = [-0.9, 0.9]$).

\begin{figure}[H]
  \centering
  \caption{Specification Test Power Curves}
  \label{fig:supLM_test_power_curves}
  \includegraphics[width=0.9\textwidth]{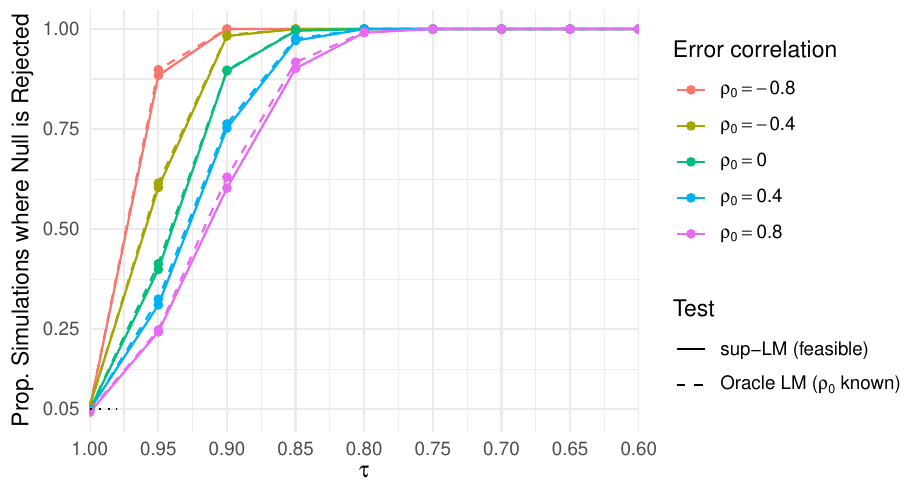}
  \fnote{This figure shows empirical rejection probabilities of the sup-LM and oracle LM tests of $H_0 : \tau = 1$ against the true value of $\tau$ for five different values of the error correlation, $\rho$. For the sup-LM test, reject $H_0$ at level $\alpha = 0.05$ if $t_N^{\text{sup}} > \hat{c}_{0.95}$. For the oracle LM test, reject $H_0$ at level $\alpha = 0.05$ if $t_N(c) > z_{0.95}$, where $z_{0.95} = 1.645$ is the 95th percentile of standard normal distribution. All objects are defined in the text above. All calculations are based on 2000 networks generated from the ITU model with the corresponding true values of $\tau$ and $\rho$, $\beta=1$, $n=100$ individuals, and a single asymmetric regressor $W_{ij} \neq W_{ji}$ drawn from a bivariate normal distribution with $\mathrm{Cov}(W_{ij}, W_{ji}) = 0.1$. For the supremum search set, we take $\ubar{\rho}=-0.9$, $\bar{\rho} = 0.9$, and 50 evenly spaced grid points. For the bootstrap, we take $B=999$ draws per test. Empirical sizes are very close to the nominal size of $\alpha = 0.05$ when $\tau = 1$, and empirical power is essentially 100\% for all values of $\rho$ considered provided $\tau \leq 0.8$. Almost no power is lost by using the sup-LM test relative to the oracle test. In the Appendix, we show that this is robust to using a smaller sample size ($n=50$) and a narrower supremum search set ($[\ubar{\rho}, \bar{\rho}] = [-0.6, 0.6]$).}
\end{figure}

\section{Application: Indian Villages}

In this section, we apply the estimation and testing procedures developed in this paper to Indian village network data from \cite{banerjee_changes_2024}. This is the follow-up paper to \cite{banerjee_diffusion_2013}, who study how information about a new microfinance opportunity spreads through village networks in Karnataka, India. In their follow-up paper, \cite{banerjee_changes_2024} provide more complete network data, with 99\% of possible links recorded. In addition to being near complete, their data provides an ideal testing ground for the theory developed in this paper. In particular, they collect data from 75 villages, each with roughly 200 households, and measure social networks on 12 different dimensions. For some network dimensions one might expect the TU model to be more appropriate, and for other dimensions the NTU model might be more appropriate. For example, with their dimension based on the question ``who would come to you if he/she needed to borrow kerosene or rice?'', it seems plausible that if household A is not willing to lend to household B at zero up-front cost, then household B could provide a monetary transfer to facilitate the lending. On the other hand, with their dimension based on the prompt ``name the four non-relatives whom you speak to the most'', it seems unlikely that household A would pay household B in some way so that they speak to them more.\footnote{\cite{banerjee_changes_2024} construct undirected, unweighted networks, as this paper's theory is built on. For example, for the top four non-relatives dimension, there is a link between $i$ and $j$ if either $i$ lists $j$ in their top four or $j$ lists $i$ in their top four (or both).} Furthermore, the fact that there is data for multiple villages means that we can perform our specification test in each village and record the overall proportion of TU model rejections. Finally, \cite{banerjee_changes_2024} provide many covariates which can be used to explain linking decisions.

Table \ref{tab:india_rejections} reports the share of villages where the sup-LM specification test developed in Section 3 rejects $H_0 : \tau =1$, i.e. that the data generating process is the TU model. We report these shares for three different network dimensions and for two different specifications of the regressors $w(X_i, X_j)$. The dimensions we consider are: who a household would borrow/lend kerosene or rice from/to; who a household names in their top four non-relatives that they speak to the most; and the union of all 12 network dimensions.\footnote{For completeness, the 12 dimensions are measured by asking households who they would: 1) borrow money from, 2) lend money to, 3) give advice to, 4) ask for advice on a difficult personal decision, 5) borrow kerosene or rice from, 6) lend kerosene or rice to, 7) ask for help in a medical emergency, 8) speak to the most, 9) say they are related to, 10) go to temple with, 11) receive as visitors in their home, and 12) visit in another's home.} The first regressor specification we consider includes an intercept, an indicator for whether the two households belong to the same caste, and the following characteristics of the partner household: its number of rooms, whether it has a sheet or reinforced cement concrete roof, whether it has private electricity, whether its primary occupation is small business or government service (as a proxy for leadership), and indicators for its caste.\footnote{Within a village, any regressor that is collinear with the others (such as an indicator for a caste category that is absent from that village) is dropped for that village only.} The idea is that these regressors capture the wealth/social standing of the partner household and wealthier/higher standing households make for more desirable connections. The partner characteristics enter $w(X_i, X_j)$ through $X_j$ and $w(X_j, X_i)$ through $X_i$, so they are the asymmetric regressors that the test relies on: recall from Section 3 that the test has no power when every regressor is symmetric. The second specification adds similarity controls. Specifically, we control for the absolute differences between the two households': number of rooms, roof quality, private electricity indicator, and leadership proxy indicator. The idea is that households are more likely to connect to other similar households, i.e. homophily.

Several observations can be made about the shares reported in Table \ref{tab:india_rejections}. First, the TU model is often rejected: it is rejected in 20-44\% of villages depending on the network dimension and regressor specification used. Under the null, the expected proportion of rejections is 5\%, which is clearly far below our observed proportions.\footnote{In every cell of Table \ref{tab:india_rejections}, a one-sided binomial test of the null that the proportion is 0.05 has
$p < 0.001$.} Second, the rejection rates align with the network dimension in an intuitive way. Specifically, as motivated, one would expect the TU model to explain borrowing/lending decisions better than social engagement. Indeed, the TU model is rejected more often for the ``speak to the most'' networks compared with the ``borrow/lend kerosene or rice'' networks: 43\% vs 39\% without similarity controls and 32\% vs 20\% with similarity control. Third, rejection rates are lower when similarity controls are included. This makes sense because, as discussed in Section 1.2, when regressors are more symmetric, TU and NTU are harder to distinguish, thus lowering power. Nonetheless, even with symmetric similarity controls, we are still able to reject the TU model at least 20\% of the time, thus highlighting the need for NTU-model-based estimators such as the NTU-MLE from Section 2 or the estimators in \cite{gao_logical_2023}, \cite{li_bagging_2026}, and \cite{chen_moment_2026}.

\begin{table}[htbp]
\centering
\caption{Share of Villages where $H_0: \text{TU}$ is Rejected}
\label{tab:india_rejections}
\begin{threeparttable}
\begin{tabular}{lcc}
\toprule
 & \multicolumn{2}{c}{Similarity controls?} \\
\cmidrule(lr){2-3}
Network & No & Yes \\
\midrule
Borrow/lend kerosene or rice & 0.387 & 0.200 \\
Speak to the most & 0.427 & 0.320 \\
Any relationship & 0.440 & 0.373 \\
\bottomrule
\end{tabular}
\begin{tablenotes}[flushleft]
  \footnotesize
  \item Each cell reports the proportion of the 75 villages in which the sup-LM test rejects $H_0: \tau = 1$ at the 5\% level ($B = 9{,}999$ bootstrap draws; supremum taken over $\rho \in [-0.9, 0.9]$ with 50 grid points).
\end{tablenotes}
\end{threeparttable}
\end{table}

To push the analysis further, we compare the TU-MLE and NTU-MLE estimates for a single village where the TU model is rejected. Specifically, with the ``any relationship'' network measure and regressor specification that includes similarity controls, the first village where the TU model is rejected is village 6. Table \ref{tab:india_village} reports the TU-MLE and NTU-MLE estimates for this specification in village 6. Two interesting observations can be made. First, in general, the estimates are in line with the predictions. $i$ is more likely to link with $j$ if $j$ is wealthier and/or has higher social standing as measured by the number of rooms in their house or if their primary occupation is small business owner or government employee. Furthermore, $i$ is more likely to link with $j$ if they share the same caste or occupation (homophily). Second, using the TU-MLE instead of the NTU-MLE when the TU model is rejected can change the sign and significance of the estimates.\footnote{The TU-MLE coefficients are $\beta/\sqrt{2(1+\rho)}$, so signs and significance are comparable across columns but magnitudes are not.} For example, the TU-MLE estimates in Table \ref{tab:india_village} suggest the quality of $j$'s roof matters for link formation, while the NTU-MLE estimates suggest no significant effect. And the TU-MLE coefficient on the Scheduled Caste/Scheduled Tribe vs Other Backward Caste dummy is positive, while the NTU-MLE coefficient is negative (albeit both insignificant). The bottom line is that the choice between the TU-MLE and NTU-MLE matters.

\begin{table}[H]
\centering
\caption{TU-MLE and NTU-MLE Estimates for a Single Village}
\label{tab:india_village}
\begin{threeparttable}
\begin{tabular}{lcc}
\toprule
 & TU-MLE & NTU-MLE \\
\midrule
Same caste & 0.105*** & 0.184*** \\
 & (0.025) & (0.047) \\
Partner number rooms & 0.139*** & 0.282*** \\
 & (0.016) & (0.030) \\
Partner roof quality & 0.117** & 0.103 \\
 & (0.042) & (0.070) \\
Partner private electricity & 0.006 & 0.071 \\
 & (0.037) & (0.063) \\
Partner leader & 0.165** & 0.278** \\
 & (0.052) & (0.088) \\
Partner caste: General & 0.371*** & 1.354*** \\
 & (0.072) & (0.329) \\
Partner caste: SC/ST & 0.004 & $-0.020$ \\
 & (0.034) & (0.059) \\
Partner caste: religious minority & 0.113 & 0.180 \\
 & (0.069) & (0.121) \\
 $|\Delta|$ number rooms & $-0.015$ & 0.004 \\
 & (0.012) & (0.026) \\
$|\Delta|$ roof quality & 0.008 & $-0.024$ \\
 & (0.026) & (0.050) \\
$|\Delta|$ private electricity & $-0.002$ & 0.007 \\
 & (0.023) & (0.046) \\
$|\Delta|$ leader & $-0.091$** & $-0.126$* \\
 & (0.032) & (0.052) \\
\midrule
$\hat{\rho}$ & -- & 0.796 \\
p-value for $H_0: \tau = 1$ & \multicolumn{2}{c}{$<0.001$} \\
Households & 112 & 112 \\
Observations (pairs) & 6,216 & 6,216 \\
\bottomrule
\end{tabular}
\begin{tablenotes}[flushleft]
  \footnotesize
  \item Estimates for the ``any relationship'' network in village 6, with similarity controls. Standard errors are in parentheses: for the TU-MLE they come from the inverse of the information matrix in R; for the NTU-MLE they come from $\hat{J}$ in Theorem 2(ii). The omitted caste category is Other Backward Caste. The constant is omitted from the table. The p-value is from the sup-LM test of $H_0: \tau = 1$ ($B = 9{,}999$ bootstrap draws; supremum taken over $\rho \in [-0.9, 0.9]$ with 50 grid points). * $p<0.05$, ** $p<0.01$, *** $p<0.001$.
\end{tablenotes}
\end{threeparttable}
\end{table}

\section{Conclusion}

This paper has contributed to the literature on dyadic network formation models by providing sufficient conditions for the consistency and asymptotic normality of the NTU-MLE. Furthermore, this paper has provided a specification test to determine whether the TU or NTU model is more appropriate in a given setting. An application to Indian villages has confirmed the value of this specification test, but we have been careful to state that this test assumes that the errors follow a normal distribution and that there is no unobserved individual heterogeneity. We reiterate the merit of a more general specification test in a semi-parametric setting with fixed effects.\footnote{Extending the test to models with fixed effects is not straightforward. Specifically, Theorem 3 relies on plugging a $\sqrt{N}$-consistent estimator into the test statistic. While it is possible to estimate $\beta$ in the TU model with fixed effects at $\sqrt{N}$ rate, for example using the estimator in \cite{graham_econometric_2017}, the fixed effects themselves can only be estimated at $\sqrt{n}$ rate. Plugging in these estimates would generate an incidental parameter bias that would need to be corrected.}

From an applied perspective, this paper has shown that the choice between the TU-MLE and NTU-MLE matters, yet the NTU-MLE is not much harder to justify than the TU-MLE theoretically. Furthermore, relative to other NTU-model-based estimators that allow for fixed effects, the NTU-MLE is very easy to compute using standard maximum likelihood routines. Given this ease of implementation and the possibility of substantive differences in empirical conclusions, researchers should treat the choice between TU and NTU as an empirical question rather than a default assumption. And the specification test developed in this paper provides an answer to this empirical question.

\section*{Acknowledgments}
I thank Francis DiTraglia, Martin Weidner, and seminar participants at the Oxford Econometrics Lunch for insightful comments and guidance. Generative AI tools have been used for proof checking and code checking; all output has been carefully reviewed; any remaining errors are the author’s responsibility.

%%%%% REFERENCES

% JEM: Quote for the top of references (just like a chapter quote if you're using them).  Comment to skip.
%\begin{savequote}[8cm]
%The first kind of intellectual and artistic personality belongs to the hedgehogs, the second to the foxes \dots
%  \qauthor{--- Sir Isaiah Berlin \cite{berlin_hedgehog_2013}}
%\end{savequote}

\setlength{\baselineskip}{0pt} % JEM: Single-space References

{\renewcommand*\MakeUppercase[1]{#1}%
\printbibliography[heading=bibintoc,title={\bibtitle}]}

%% APPENDICES %% 
% Starts lettered appendices, adds a heading in table of contents, and adds a
%    page that just says "Appendices" to signal the end of your main text.
\newpage
\appendix

\section{Lemmas and Proofs}

\setlist[description]{font=\normalfont\itshape}
\makeatletter
\renewenvironment{proof}[1][\proofname]{%
  \par\pushQED{\qed}%
  \normalfont\topsep6\p@\@plus6\p@\relax
  \trivlist
  \item[\hskip\labelsep\bfseries #1.]%
}{%
  \popQED\endtrivlist\@endpefalse
}
\makeatother

\begin{proof}[Proof of Proposition 1]
 \vspace{1ex}  % forces a line break

 The proof follows six steps:
 \begin{enumerate}
     \item Record some known properties of the normal reverse hazard function, $\lambda(z) \coloneq \frac{\phi(z)}{\Phi(z)}$.
     \item Argue that the bivariate normal distribution with two equal arguments, $F_{\rho}(a) \coloneq \Phi_2^{\rho}(a, a)$, is strictly increasing in both $a$ and $\rho$. Also argue that for fixed $p \in (0,1)$ the map $\rho \mapsto F_{\rho}^{-1}(p)$ has derivative:
     $$
     \frac{d}{d\rho}F_{\rho}^{-1}(p) = - \frac{\lambda\left(\sqrt{\frac{1-\rho}{1+\rho}}F_{\rho}^{-1}(p)\right)}{2 \sqrt{1 - \rho^2}}
     $$
     \item Combine steps 1 and 2 to argue that two symmetric dyads and a mixed dyad as in \ref{ass:I3}(a) pin down $\rho$.
     \item Similarly, combine steps 1 and 2 to argue that three collinear symmetric dyads as in \ref{ass:I3}(b) pin down $\rho$.
     \item Use \ref{ass:I2} to pin down $\beta$ given $\rho$.
     \item Conclude formally that $\theta = (\beta, \rho)$ is identified under \ref{ass:I1}-\ref{ass:I3}.
 \end{enumerate}

 \textbf{Step 1:} We note that $\lambda(z) \coloneq \frac{\phi(z)}{\Phi(z)}$ satisfies: (i) $\lambda'(z) \in(-1,0)$, (ii) $\lambda''(z) > 0$, (iii) $(\log\lambda(z))'' < 0$, and (iv) $\lambda(z) \leq 1 + |z|$. For (i), use inequality (3) in \cite{sampford_inequalities_1953} which states that $v(z) \coloneq \frac{\phi(z)}{1 - \Phi(z)} = \lambda(-z)$ satisfies $0 < v'(z) < 1$, and therefore $-1 < \lambda'(z) < 0$. For (ii), use inequality (4) in \cite{sampford_inequalities_1953} which states that $v''(z) > 0$, and therefore $\lambda''(z) > 0$. For (iii), note that $(\log\lambda(z))' = -z - \lambda(z)$, so $(\log\lambda(z))'' = -1 - \lambda'(z) = -1 + v'(-z) < 0$ since $v'(z) < 1$ for all $z$ again. For (iv), use the inequality in \cite{birnbaum_inequality_1942} which states that $\frac{1}{v(z)} \geq \frac{\sqrt{4+z^2} - z}{2}$ for $z > 0$. Then, since $\frac{\sqrt{4+z^2} - z}{2} \geq \frac{1}{1+z}$ for $z > 0$, $1+z \geq v(z) = \lambda(-z)$ for $z>0$, and therefore $\lambda(z) < 1 + |z|$ for $z < 0$. If $z \geq 0$, simply use that $\lambda(z) \leq \lambda(0) = \sqrt{\frac{2}{\pi}} < 1 < 1 + |z|$ since $\lambda(z)$ is strictly decreasing from (i).

 \textbf{Step 2:} By definition:
 $$
 \Phi_2^{\rho}(a, b) = \int_{-\infty}^a \int_{-\infty}^b \phi_2^{\rho}(s, t) \ dt \ ds,
 $$
 with:
 \begin{align*}
     \phi_2^{\rho}(s, t) &= \frac{1}{2 \pi \sqrt{1 - \rho^2}} \text{exp} \left[- \frac{1}{2 (1 - \rho^2)}(s^2 - 2\rho s t + t^2) \right] \\
     &= \frac{1}{\sqrt{1 - \rho^2}} \phi(s) \phi\left(\frac{t - \rho s}{\sqrt{1 - \rho^2}}\right),
 \end{align*}
 where the second line follows from decomposing the joint law into a marginal times conditional. Plugging this into $\partial_a\Phi_2^{\rho}(a, b)$ gives:
  \begin{align*}
     \partial_a\Phi_2^{\rho}(a, b) &= \int_{-\infty}^b \phi_2^{\rho}(a, t) \ dt \\
     &= \frac{1}{\sqrt{1 - \rho^2}} \phi(a) \int_{-\infty}^b\phi\left(\frac{t - \rho a}{\sqrt{1 - \rho^2}}\right) \ dt \\
     &=\phi(a) \Phi\left(\frac{b - \rho a}{\sqrt{1 - \rho^2}}\right).
 \end{align*}
 The result for $\partial_b\Phi_2^{\rho}(a, b)$ is symmetric and therefore via the chain rule:
 $$
 \partial_a F_{\rho}(a) = 2 \phi(a) \Phi\left(\frac{a(1- \rho)}{\sqrt{1 - \rho^2}}\right) > 0.
 $$
 Strict positivity follows from $\phi > 0$ and $\Phi > 0$ given $\rho \in (-1,1)$. That is, $F_{\rho}(a)$ is strictly increasing in $a$, as required.

 For the derivative with respect to $\rho$, apply \cite{plackett_reduction_1954} to the bivariate normal case to obtain $\partial_{\rho}\Phi_2^{\rho}(a, b) = \phi_2^{\rho}(a, b)$. Thus, $\partial_{\rho} F_{\rho}(a) = \phi_2^{\rho}(a, a) > 0$, i.e. $F_{\rho}(a)$ is strictly increasing in $\rho$, as required.

 For the derivative of $F_{\rho}^{-1}(p)$, differentiate $F_{\rho}(F_{\rho}^{-1}(p)) = p$ implicitly with $a = F_{\rho}^{-1}(p)$ to obtain:
  $$
 \frac{d}{d\rho}F_{\rho}^{-1}(p) = - \frac{\partial_{\rho} F_{\rho}(a)}{\partial_{a} F_{\rho}(a)} = - \frac{\phi_2^{\rho}(a, a)}{2 \phi(a) \Phi\left(\frac{a(1- \rho)}{\sqrt{1 - \rho^2}}\right)}.
 $$
 Then use the conditional factorization of $\phi_2^{\rho}(a, b)$ at $a = b$ to obtain:
 $$
 \frac{\phi_2^{\rho}(a, a)}{\phi(a)} = \frac{1}{\sqrt{1 - \rho^2}}\phi\left(\frac{a(1- \rho)}{\sqrt{1 - \rho^2}}\right).
 $$
 Finally, plugging this in yields the desired expression:
 $$
 \frac{d}{d\rho}F_{\rho}^{-1}(p) = - \frac{\phi\left(\frac{a(1- \rho)}{\sqrt{1 - \rho^2}}\right)}{2 \sqrt{1 - \rho^2} \Phi\left(\frac{a(1- \rho)}{\sqrt{1 - \rho^2}}\right)}= - \frac{\lambda\left(\sqrt{\frac{1-\rho}{1+\rho}}a\right)}{2 \sqrt{1 - \rho^2}} = - \frac{\lambda\left(\sqrt{\frac{1-\rho}{1+\rho}}F_{\rho}^{-1}(p)\right)}{2 \sqrt{1 - \rho^2}}.
 $$

 \textbf{Step 3:} Let $(a_1, a_1)$ and $(a_2, a_2)$ be two distinct symmetric dyads with mixed dyad $(a_1, a_2)$. We want to show that the system:
 \[
    \Phi_2^\rho(a_1,a_1)=p_{11},\qquad \Phi_2^\rho(a_2,a_2)=p_{22},\qquad \Phi_2^\rho(a_1,a_2)=p_{12},
    \]
 with $p_{11}, p_{22}, p_{12} \in (0,1)$ has at most one solution $(a_1, a_2, \rho)$. Since $F_{\rho}(a)$ is strictly increasing in $a$ from step 2, we can invert the first two equations in the system to get $a_1 = F_{\rho}^{-1}(p_{11})$ and $a_2 = F_{\rho}^{-1}(p_{22})$. Then plug these into the third equation to define:
 $$
 G(\rho) \coloneq \Phi_2^\rho(F_{\rho}^{-1}(p_{11}),F_{\rho}^{-1}(p_{22}))=p_{12}.
 $$
 Therefore, if we can show that $G(\rho)$ is strictly monotonic, then we can invert $G$ to get the unique solution $(a_1, a_2, \rho) = (F_{\rho}^{-1}(p_{11}), F_{\rho}^{-1}(p_{22}), G^{-1}(p_{12}))$, as required. Via the chain rule, the formula for $\frac{d}{d\rho}F_{\rho}^{-1}(p)$ from step 2, and the fact that $\frac{\phi_2^\rho(a_1, a_1)}{\phi_2^\rho(a_1, a_2)} = \frac{\phi(z_{11})}{\phi(z_{12})}$ with $z_{mn} \coloneq \frac{a_n - \rho a_m}{\sqrt{1-\rho^2}}$, the derivative of $G$ is:
 \begin{align*}
     G'(\rho) &= \phi_2^\rho(a_1, a_2) + \phi(a_1) \Phi\left( \frac{a_2 - \rho a_1}{\sqrt{1 - \rho^2}}\right)a_1' + \phi(a_2) \Phi\left( \frac{a_1 - \rho a_2}{\sqrt{1 - \rho^2}}\right)a_2' \\
     &= \phi_2^\rho(a_1, a_2) - \phi(a_1) \Phi\left( \frac{a_2 - \rho a_1}{\sqrt{1 - \rho^2}}\right) \frac{\lambda\left(\sqrt{\frac{1-\rho}{1+\rho}}a_1\right)}{2 \sqrt{1 - \rho^2}} - \phi(a_2) \Phi\left( \frac{a_1 - \rho a_2}{\sqrt{1 - \rho^2}}\right)\frac{\lambda\left(\sqrt{\frac{1-\rho}{1+\rho}}a_2\right)}{2 \sqrt{1 - \rho^2}} \\
     &= \phi_2^\rho(a_1, a_2) \left[1 - \frac{1}{2}\left(\frac{\lambda(z_{11})}{\lambda(z_{12})} + \frac{\lambda(z_{22})}{\lambda(z_{21})}\right) \right].
 \end{align*}
 So to show that $G'(\rho) < 0$, we require to show that $\frac{\lambda(z_{11})}{\lambda(z_{12})} + \frac{\lambda(z_{22})}{\lambda(z_{21})} > 2$. Since $a_1 \neq a_2$, take $a_2 > a_1$ without loss of generality.

 Case 1: $\rho = 0$. Then $\frac{\lambda(z_{11})}{\lambda(z_{12})} + \frac{\lambda(z_{22})}{\lambda(z_{21})} = \frac{\lambda(a_1)}{\lambda(a_2)} + \frac{\lambda(a_2)}{\lambda(a_1)} = \frac{\lambda(a_1)}{\lambda(a_2)} + \frac{1}{\frac{\lambda(a_1)}{\lambda(a_2)}}> 2$ since $a_1 \neq a_2$ by assumption and $\lambda(z) > 0$.

 Case 2: $\rho > 0$. Because $\lambda(z)$ is strictly log-concave from step 1, $\frac{\lambda(x+h)}{\lambda(x)}$ is strictly decreasing in $x$ for fixed $h > 0$. Taking $x_1 = z_{21}$, $h = z_{22} - z_{21}$ (is positive because $a_2 > a_1$), and $x_2 = z_{11}$ gives that $\frac{\lambda(z_{22})}{\lambda(z_{21})} > \frac{\lambda(z_{12})}{\lambda(z_{11})}$ provided $z_{21} < z_{11}$, which holds because $a_2 > a_1$. Thus, $\frac{\lambda(z_{11})}{\lambda(z_{12})} + \frac{\lambda(z_{22})}{\lambda(z_{21})} > \frac{\lambda(z_{11})}{\lambda(z_{12})} + \frac{\lambda(z_{12})}{\lambda(z_{11})}=\frac{\lambda(z_{11})}{\lambda(z_{12})} + \frac{1}{\frac{\lambda(z_{11})}{\lambda(z_{12})}} > 2$ again, as required.

 Case 3: $\rho < 0$. Apply the same result as in case 2 but now with $x_2 = z_{12}$ to obtain $\frac{\lambda(z_{22})}{\lambda(z_{21})} > \frac{\lambda(2z_{12} - z_{11})}{\lambda(z_{12})}$ provided $z_{21} < z_{12}$, which holds because $a_2 > a_1$ and $\rho > -1$. So $\frac{\lambda(z_{11})}{\lambda(z_{12})} + \frac{\lambda(z_{22})}{\lambda(z_{21})} > \frac{\lambda(z_{11}) + \lambda(2z_{12} - z_{11})}{\lambda(z_{12})}$. Then, since $\lambda$ is strictly convex from step 1 and therefore strictly midpoint convex, $\frac{\lambda(z_{11}) + \lambda(2z_{12} - z_{11})}{2} > \lambda(z_{12})$. Combined, $\frac{\lambda(z_{11})}{\lambda(z_{12})} + \frac{\lambda(z_{22})}{\lambda(z_{21})} > 2$, as required.

 \textbf{Step 4:} Let $(a_0, a_0)$, $(a_1, a_1)$, and $(a_2, a_2)$ be three distinct symmetric dyads with $a_1 = \frac{1}{2}(a_0 + a_2)$. We want to show that the system:
 \[
    \Phi_2^\rho(a_0,a_0)=p_{00},\qquad
    \Phi_2^\rho(a_1,a_1)=p_{11},\qquad \Phi_2^\rho(a_2,a_2)=p_{22},
    \]
 with $p_{00}, p_{11}, p_{22} \in (0,1)$ has at most one solution $(a_0, a_1, a_2, \rho)$. Define:
 $$
 K(\rho) \coloneq F_{\rho}^{-1}(p_{11}) - \frac{1}{2}\left(F_{\rho}^{-1}(p_{00}) + F_{\rho}^{-1}(p_{22}) \right).
 $$
 Since $a_m = F_{\rho}^{-1}(p_{mm})$ for $m=0,1,2$ and $a_1 = \frac{1}{2}(a_0 + a_2)$, we have that $K(\rho) = 0$. We want to show that $K(\rho)$ has at most one root in $(-1,1)$. Using the formula for $\frac{d}{d\rho}F_{\rho}^{-1}(p)$ from step 2:
 $$
 K'(\rho) = - \frac{1}{2 \sqrt{1 - \rho^2}} \left[\lambda\left(\sqrt{\frac{1-\rho}{1+\rho}}a_1\right) - \frac{1}{2}\left(\lambda\left(\sqrt{\frac{1-\rho}{1+\rho}}a_0\right) + \lambda\left(\sqrt{\frac{1-\rho}{1+\rho}}a_2\right)\right) \right].
 $$
 Let $\rho^*$ be any root of $K(\rho)$. By definition of $K(\rho)$, $a_1(\rho^*) = \frac{1}{2}(a_0(\rho^*) + a_2(\rho^*))$. Therefore, by strict convexity of $\lambda$ from step 1:
 $$
 \lambda\left(\sqrt{\frac{1-\rho^*}{1+\rho^*}}a_1(\rho^*)\right) < \frac{1}{2}\left(\lambda\left(\sqrt{\frac{1-\rho^*}{1+\rho^*}}a_0(\rho^*)\right) + \lambda\left(\sqrt{\frac{1-\rho^*}{1+\rho^*}}a_2(\rho^*)\right)\right).
 $$
 So $K'(\rho^*) > 0$. A continuous function (which $K$ is) that is strictly increasing at every root must have at most one root, as required.

 \textbf{Step 5:} Fix $\rho$ and assume \ref{ass:I2} holds. We want to show that if $\Phi_2^\rho(w_1'\beta, w_2'\beta) = \Phi_2^\rho(w_1'\tilde{\beta}, w_2'\tilde{\beta})$ for all $(w_1, w_2) \in S$, then $\beta = \tilde{\beta}$. Suppose for a contradiction $\beta \neq \tilde{\beta}$, so $b \coloneq \beta - \tilde{\beta} \neq 0$. Therefore, by \ref{ass:I2}, there exists $(w_1, w_2) \in S$ such that $(w_1'(\beta - \tilde{\beta}))(w_2'(\beta - \tilde{\beta})) \geq 0$ and $(w_1'(\beta - \tilde{\beta}), w_2'(\beta - \tilde{\beta})) \neq (0,0)$. That is, $w_1'(\beta - \tilde{\beta})$ and $w_2'(\beta - \tilde{\beta})$ are weakly of the same sign with at least one strictly non-zero. Suppose without loss of generality $w_1'(\beta - \tilde{\beta}) > 0$ and $w_2'(\beta - \tilde{\beta}) \geq 0$ (the other cases follow identical logic). Then, since $\Phi_2^{\rho}(\cdot, \cdot)$ is strictly increasing in both arguments from step 2:
 $$
 \Phi_2^\rho(w_1'\beta, w_2'\beta) > \Phi_2^\rho(w_1'\tilde{\beta}, w_2'\beta) \geq \Phi_2^\rho(w_1'\tilde{\beta}, w_2'\tilde{\beta}),
 $$
 which contradicts $\Phi_2^\rho(w_1'\beta, w_2'\beta) = \Phi_2^\rho(w_1'\tilde{\beta}, w_2'\tilde{\beta})$, as required.
 
 \textbf{Step 6:} Consider $(\beta, \rho)$ and $(\tilde{\beta}, \tilde{\rho})$. By the equivalence stated after the definition of identification in the main text, it suffices to show that if $\Phi_2^\rho(w_1'\beta, w_2'\beta) = \Phi_2^{\tilde{\rho}}(w_1'\tilde{\beta}, w_2'\tilde{\beta})$ for all $(w_1, w_2) \in S$, then $(\beta, \rho) = (\tilde{\beta}, \tilde{\rho})$. Under \ref{ass:I3}(a), both $(w_1'\beta, w_2'\beta, \rho)$ and $(w_1'\tilde{\beta}, w_2'\tilde{\beta}, \tilde{\rho})$ are solutions to the system in step 3, so by step 3 $(w_1'\beta, w_2'\beta, \rho) = (w_1'\tilde{\beta}, w_2'\tilde{\beta}, \tilde{\rho})$ and therefore $\rho=\tilde{\rho}$. Under \ref{ass:I3}(b) instead, both $(w_0'\beta, w_1'\beta, w_2'\beta, \rho)$ and $(w_0'\tilde{\beta}, w_1'\tilde{\beta}, w_2'\tilde{\beta}, \tilde{\rho})$ are solutions to the system in step 4, so by this step $(w_0'\beta, w_1'\beta, w_2'\beta, \rho) = (w_0'\tilde{\beta}, w_1'\tilde{\beta}, w_2'\tilde{\beta}, \tilde{\rho})$ and therefore $\rho=\tilde{\rho}$ again. Then, given $\rho=\tilde{\rho}$, $\Phi_2^\rho(w_1'\beta, w_2'\beta) = \Phi_2^{\rho}(w_1'\tilde{\beta}, w_2'\tilde{\beta})$ for all $(w_1, w_2) \in S$, and therefore $\beta = \tilde{\beta}$ by step 5 given \ref{ass:I2}.

\end{proof}

\setcounter{lemma}{0}

\begin{lemma}
    Let $h(y, x, x'; \theta)$ be a symmetric kernel that is continuous in $\theta \in \Theta$ for every $(y, x, x') \in \{0,1\} \times \mathbb{R}^p \times \mathbb{R}^p$, with $\Theta$ compact. Assume \ref{ass:I1} holds and write $h_{ij}(\theta) \coloneq h(Y_{ij}, X_i, X_j; \theta)$. If:
    $$
    \mathbb{E}\left[\sup_{\theta \in \Theta} \left|h_{ij}(\theta) \right|^{1+\eta} \right] < \infty
    $$
    for some $\eta \in (0,1]$, then $\mathbb{E}[h_{ij}(\theta)]$ is continuous on $\Theta$ and:
    $$
    \sup_{\theta \in \Theta} \left|\frac{1}{N} \sum^{n-1}_{i=1} \sum^{n}_{j=i+1}h_{ij}(\theta) - \mathbb{E}[h_{ij}(\theta)]\right| \xrightarrow{p} 0.
    $$
    Furthermore, if $\eta=1$ then we have almost sure convergence.
\end{lemma}

\begin{proof}[Proof of Lemma 1]
   \vspace{1ex} % forces a line break
    The proof follows four steps:
 \begin{enumerate}
     \item Establish pointwise convergence in probability, i.e. $\frac{1}{N} \sum^{n-1}_{i=1} \sum^{n}_{j=i+1}h_{ij}(\theta) - \mathbb{E}[h_{ij}(\theta)] \xrightarrow{p} 0$ for fixed $\theta \in \Theta$.
     \item Establish continuity of $\mathbb{E}[h_{ij}(\theta)]$.
     \item Upgrade to uniform convergence in probability.
     \item Argue that when $\eta=1$ we can upgrade from convergence in probability to almost sure convergence.
 \end{enumerate}

 \textbf{Step 1:} Fix $\theta \in \Theta$ and decompose:
    \begin{align*}
         \frac{1}{N} \sum^{n-1}_{i=1} \sum^{n}_{j=i+1}h_{ij}(\theta) - \mathbb{E}[h_{ij}(\theta)]
         &= \underbrace{\frac{1}{N} \sum^{n-1}_{i=1} \sum^{n}_{j=i+1}\left(h_{ij}(\theta) - \mathbb{E}[h_{ij}(\theta) \mid \mathcal{X}]\right)}_{\text{(I)}} \\
         &+ \underbrace{\frac{1}{N} \sum^{n-1}_{i=1} \sum^{n}_{j=i+1}\mathbb{E}[h_{ij}(\theta) \mid \mathcal{X}] - \mathbb{E}[h_{ij}(\theta)]}_{\text{(II)}}.
    \end{align*}
    
    For term (I), we use a conditional Chebyshev argument. Specifically, we decompose $h_{ij} = h_{ij} \mathds{1}[|h_{ij}| \leq \sqrt{N}] + h_{ij} \mathds{1}[|h_{ij}| > \sqrt{N}]$ and plug into term (I) to obtain:
    \begin{align*}
         \text{(I)}
         &= \underbrace{\frac{1}{N} \sum^{n-1}_{i=1} \sum^{n}_{j=i+1}\left(h_{ij} \mathds{1}[|h_{ij}| \leq \sqrt{N}] - \mathbb{E}[h_{ij} \mathds{1}[|h_{ij}| \leq \sqrt{N}] \mid \mathcal{X}]\right)}_{\text{(Ia)}} \\
         &+ \underbrace{\frac{1}{N} \sum^{n-1}_{i=1} \sum^{n}_{j=i+1}\left(h_{ij} \mathds{1}[|h_{ij}| > \sqrt{N}] - \mathbb{E}[h_{ij} \mathds{1}[|h_{ij}| > \sqrt{N}] \mid \mathcal{X}]\right)}_{\text{(Ib)}}.
    \end{align*}
    Conditional on $\mathcal{X}$, $h_{ij}$ are independent across dyads via assumption \ref{ass:I1}. Therefore, the conditional variance of term (Ia) is bounded as such:
    \begin{align*}
        \mathbb{V}\text{ar} \left[(Ia) \mid \mathcal{X} \right] &= \frac{1}{N^2} \sum^{n-1}_{i=1} \sum^{n}_{j=i+1} \mathbb{V}\text{ar} \left[h_{ij} \mathds{1}[|h_{ij}| \leq \sqrt{N}]\mid \mathcal{X} \right] \\
        &\leq \frac{1}{N^2} \sum^{n-1}_{i=1} \sum^{n}_{j=i+1} \mathbb{E} \left[\left(h_{ij} \mathds{1}[|h_{ij}| \leq \sqrt{N}]\right)^2 \mid \mathcal{X} \right] \\
        &= \frac{1}{N^2} \sum^{n-1}_{i=1} \sum^{n}_{j=i+1} \mathbb{E} \left[\left|h_{ij} \mathds{1}[|h_{ij}| \leq \sqrt{N}]\right|^{1 - \eta}\left|h_{ij} \mathds{1}[|h_{ij}| \leq \sqrt{N}]\right|^{1 + \eta} \mid \mathcal{X} \right] \\
        &\leq \frac{1}{N^2} \sum^{n-1}_{i=1} \sum^{n}_{j=i+1}  \sqrt{N}^{1 - \eta} \mathbb{E} \left[\left|h_{ij} \right|^{1 + \eta} \mid \mathcal{X} \right] \\
        &= N^{-\frac{1 + \eta}{2}}\cdot \frac{1}{N} \sum^{n-1}_{i=1} \sum^{n}_{j=i+1} \mathbb{E} \left[\left|h_{ij} \right|^{1 + \eta} \mid \mathcal{X} \right].
    \end{align*}
    By the strong law of large numbers for U-statistics (Theorem 5.4.A) in \cite{serfling_approximation_1980} applied to $\frac{1}{N} \sum^{n-1}_{i=1} \sum^{n}_{j=i+1} \mathbb{E} \left[\left|h_{ij} \right|^{1 + \eta} \mid \mathcal{X} \right]$, since (i) $\mathbb{E} \left[\left|h_{ij} \right|^{1 + \eta} \mid \mathcal{X} \right]$ does not depend on $Y_{ij}$ or $n$, (ii) $X_i$ are i.i.d., (iii) $\mathbb{E} \left[\left|h_{ij} \right|^{1 + \eta} \mid \mathcal{X} \right]$ is symmetric because $h_{ij}$ is symmetric, and (iv) $\mathbb{E} \left|\mathbb{E} \left[\left|h_{ij} \right|^{1 + \eta} \mid \mathcal{X} \right]\right| \leq \mathbb{E} \left|h_{ij} \right|^{1 + \eta} < \infty$ by assumption: $\frac{1}{N} \sum^{n-1}_{i=1} \sum^{n}_{j=i+1} \mathbb{E} \left[\left|h_{ij} \right|^{1 + \eta} \mid \mathcal{X} \right]$ converges almost surely to a finite constant. And $N^{-\frac{1 + \eta}{2}}$ converges to 0, so the upper bound on $\mathbb{V}\text{ar} \left[(Ia) \mid \mathcal{X} \right]$ converges to 0, so by Chebyshev's inequality (applied conditionally on $\mathcal{X}$ and then made unconditional via dominated convergence), term (Ia) converges in probability to 0. For term (Ib), note that it is bounded in absolute expectation as such:
    \begin{align*}
        \mathbb{E} \left|(Ib) \right| &\leq \frac{1}{N} \sum^{n-1}_{i=1} \sum^{n}_{j=i+1}\left(\mathbb{E}\left|h_{ij} \mathds{1}[|h_{ij}| > \sqrt{N}]\right| + \mathbb{E} \left|\mathbb{E}[h_{ij} \mathds{1}[|h_{ij}| > \sqrt{N}] \mid \mathcal{X}]\right|\right) \\
        &\leq 2 \mathbb{E}\left|h_{ij} \mathds{1}[|h_{ij}| > \sqrt{N}]\right| \\
        &= 2 \mathbb{E}\left|\frac{h_{ij}^{1+\eta}}{h_{ij}^\eta} \mathds{1}[|h_{ij}| > \sqrt{N}]\right| \\
        &\leq \frac{2\mathbb{E}\left|h_{ij} \right|^{1 + \eta}}{\sqrt{N}^\eta}.
    \end{align*}
    The numerator of the upper bound on $\mathbb{E} \left|(Ib) \right|$ is finite by assumption and the denominator goes to infinity as $n$ goes to infinity. So by Markov's inequality, since $\mathbb{E} \left|(Ib) \right|$ converges to 0, term (Ib) converges in probability to 0 too.

    For term (II), we apply \cite{serfling_approximation_1980}'s SLLN for U-statistics again. Specifically, because (i) $\mathbb{E}[h_{ij} \mid \mathcal{X}]$ does not depend on $Y_{ij}$ or $n$, (ii) $X_i$ are i.i.d., (iii) $\mathbb{E}[h_{ij} \mid \mathcal{X}]$ is symmetric because $h_{ij}$ is symmetric, and (iv) $\mathbb{E} \left|\mathbb{E} \left[h_{ij} \mid \mathcal{X} \right]\right| \leq \mathbb{E} \left|h_{ij} \right| < \infty$ by assumption: $\frac{1}{N} \sum^{n-1}_{i=1} \sum^{n}_{j=i+1}\mathbb{E}[h_{ij} \mid \mathcal{X}]$ converges almost surely to $\mathbb{E}[h_{ij}]$, as required.

    \textbf{Step 2:} Simply note that $h_{ij}(\theta)$ is continuous by assumption, and is dominated by $\sup_{\theta \in \Theta} \left|h_{ij}(\theta) \right|$, which has finite expectation by Jensen's inequality and the moment assumption. So, by dominated convergence, $\mathbb{E}[h_{ij}(\theta)]$ is continuous too.

    \textbf{Step 3:} Decompose differently to as in step 1:
    \begin{align*}
         \frac{1}{N} \sum^{n-1}_{i=1} \sum^{n}_{j=i+1}h_{ij}(\theta) - \mathbb{E}[h_{ij}(\theta)]
         &= \underbrace{\frac{1}{N} \sum^{n-1}_{i=1} \sum^{n}_{j=i+1}h_{ij}(\theta_k) - \mathbb{E}[h_{ij}(\theta_k)]}_{\text{(i)}} \\
         &+ \underbrace{\frac{1}{N} \sum^{n-1}_{i=1} \sum^{n}_{j=i+1}\left(h_{ij}(\theta) - h_{ij}(\theta_k)\right)}_{\text{(ii)}}
         \\
         &+ \underbrace{\left(\mathbb{E}[h_{ij}(\theta_k)] - \mathbb{E}[h_{ij}(\theta)]\right)}_{\text{(iii)}},
    \end{align*}
    where $\theta \in B_k \coloneq \{\theta \in \Theta : \| \theta - \theta_k \| < r_k \}$ for some centre $\theta_k \in \Theta$ and radius $r_k > 0$. By compactness of $\Theta$, there are finitely many balls $B_k$, say $K$ balls in total. Because term (i) depends on $\theta$ only through $k$:
    $$
    \sup_{\theta \in \Theta} \left|(i)\right| = \max_{k \leq K} \left|(i)\right| = o_P(1),
    $$
    since term (i) is $o_P(1)$ for each $k$ from step 1 and the maximum of finitely many $o_P(1)$ terms is $o_P(1)$.
    
    For terms (ii) and (iii), define:
    $$
    \Delta_{ij}(\theta^*, r) \coloneq \sup_{\theta \in \Theta \ : \ \| \theta - \theta^* \| < r} \left|h_{ij}(\theta) - h_{ij}(\theta^*)\right|.
    $$
    Therefore:
    $$
    \sup_{\theta \in \Theta} \left|(ii)\right| \leq \max_{k \leq K} \frac{1}{N} \sum^{n-1}_{i=1} \sum^{n}_{j=i+1} \Delta_{ij}(\theta_k, r_k).
    $$
    For each $k$, apply step 1 to the kernel $\Delta_{ij}(\theta_k, r_k)$. This is legitimate because $\Delta_{ij}$ is built from $h_{ij}$ in a way that preserves symmetry and the $1 + \eta$ moment. So $\frac{1}{N} \sum^{n-1}_{i=1} \sum^{n}_{j=i+1} \Delta_{ij}(\theta_k, r_k) \xrightarrow{p} \mathbb{E} \left[\Delta_{ij}(\theta_k, r_k) \right]$. Now take arbitrary $\epsilon > 0$. Since $h_{ij}$ is continuous, $\Delta_{ij}(\theta^*, r) \to 0$ as $r \to 0$. Furthermore:
    $$
    \Delta_{ij}(\theta^*, r) \leq \sup_{\theta \in \Theta}|h_{ij}(\theta)| + |h_{ij}(\theta^*)| \leq 2\sup_{\theta \in \Theta}|h_{ij}(\theta)|.
    $$
    So $\mathbb{E}\left[\Delta_{ij}(\theta^*, r) \right] < \infty$ by assumption. Thus, by dominated convergence, $\mathbb{E}\left[\Delta_{ij}(\theta^*, r_n)\right] \to 0$ as $r_n \downarrow 0$ for a generic sequence $r_n$. Consequently, for any $\epsilon > 0$, there exists some $r(\theta^*)$ such that $\mathbb{E}\left[\Delta_{ij}(\theta^*, r) \right] < \epsilon$. This bounds the supremum of term (ii) by $2\epsilon$ and the supremum of term (iii) by $\epsilon$. 

    All together:
    $$
    \mathbb{P} \left[\sup_{\theta \in \Theta} \left|\frac{1}{N} \sum^{n-1}_{i=1} \sum^{n}_{j=i+1}h_{ij}(\theta) - \mathbb{E}[h_{ij}(\theta)]\right| > 4\epsilon\right] \to 0.
    $$
    Then the result follows because $\epsilon$ was arbitrary.

    \textbf{Step 4:} To upgrade from convergence in probability to almost sure convergence when $\eta = 1$, note that we already have almost sure convergence of term (II) in step 1. For term (I), note that when $\eta=1$:
    \begin{align*}
        \mathbb{E}[(I)^2] &= \mathbb{E}[\mathbb{V}\text{ar} \left[(I) \mid \mathcal{X} \right]] \\
        &= \mathbb{E}\left[\frac{1}{N^2} \sum^{n-1}_{i=1} \sum^{n}_{j=i+1} \mathbb{V}\text{ar}[h_{ij} \mid \mathcal{X}]\right] \\
        &\leq\mathbb{E}\left[\frac{1}{N^2} \sum^{n-1}_{i=1} \sum^{n}_{j=i+1} \mathbb{E}[h_{ij}^2 \mid \mathcal{X}]\right] \\
        &=\frac{\mathbb{E}[h_{ij}^2]}{N}.
    \end{align*}
    Therefore, by Chebyshev's inequality for each $n\geq2$:
    $$
    \mathbb{P}[|(I)| \geq \epsilon] \leq \frac{\mathbb{E}[h_{ij}^2]}{\epsilon^2} \frac{1}{N}.
    $$
    Thus, by Borel-Cantelli, we have almost sure convergence of term (I) provided $\sum_{n=2}^\infty\frac{1}{N} < \infty$:
    $$
    \sum_{n=2}^\infty\frac{2}{n(n-1)} = 2\sum_{n=2}^\infty\left(\frac{1}{n-1} - \frac{1}{n}\right) = 2 < \infty,
    $$
    as required. Finally, to upgrade from pointwise almost sure convergence to uniform almost sure convergence, reproduce the argument in step 3.
    
\end{proof}

\begin{proof}[Proof of Theorem 1]
 \vspace{1ex}  % forces a line break
We establish the conditions of Theorem 2.1 in \cite{newey_chapter_1994}, which states that $\hat{\theta}$ converges in probability to $\theta_0$ provided:
\begin{enumerate}
    \item[(2.1.i)] $Q_0 (\theta) \coloneq \mathbb{E} \big[ Y_{ij} \log \Phi_2^{\rho}\big(W_{ij}^{'}\beta, W_{ji}^{'}\beta\big) + (1 - Y_{ij})\log \big(1 - \Phi_2^{\rho}\big(W_{ij}^{'}\beta, W_{ji}^{'}\beta\big)\big) \big]$ has a unique maximum at $\theta_0$, 
    \item[(2.1.ii)] $\Theta$ is compact,
    \item[(2.1.iii)] $Q_0 (\theta)$ is continuous,
    \item[(2.1.iv)] $\hat{Q}_n(\theta) \coloneq \frac{1}{N}\sum^{n-1}_{i=1} \sum^{n}_{j=i+1} \big[Y_{ij} \log \Phi_2^{\rho}\big(W_{ij}^{'}\beta, W_{ji}^{'}\beta\big) + (1 - Y_{ij})\log \big(1 - \Phi_2^{\rho}\big(W_{ij}^{'}\beta, W_{ji}^{'}\beta\big)\big) \big]$ converges \textit{uniformly} in probability to $Q_0 (\theta)$.
\end{enumerate}
    Before establishing these conditions, we provide a lower bound on $\Phi_2^{\rho}\big(W_{ij}^{'}\beta, W_{ji}^{'}\beta\big)$ that is used repeatedly.

    \textbf{Lower bound on $\Phi_2^{\rho}\big(W_{ij}^{'}\beta, W_{ji}^{'}\beta\big)$:} The lower bound we want to show is:
    \begin{equation}\label{P-lower-bound}
        \Phi_2^{\rho}\big(a, b\big) \geq \frac{\phi(a) \ \Phi(z_{ab}) \ \text{exp}\left(-\frac{1}{2} - \frac{\rho^2}{2(1-\rho^2)} \right)}{1 + |a| + \frac{|\rho|}{\sqrt{1 - \rho^2}} \frac{\phi(z_{ab})}{\Phi(z_{ab})}},
    \end{equation}
    where $z_{ab} \coloneq \frac{b - \rho a}{\sqrt{1 - \rho^2}}$. To this end, decompose the joint law into a marginal times conditional and then make the substitution $u = a - s$, noting that the integrand is positive so we can discard all $u$ outside $u \in [0,1]$ to bound:
    \begin{align*}
     \Phi_2^{\rho}\big(a, b\big) &= \int_{-\infty}^a \phi(s) \ \Phi\left(\frac{b - \rho s}{\sqrt{1 - \rho^2}}\right) \ ds \\
     &= \int_{0}^\infty \phi(a-u) \ \Phi\left(\frac{b - \rho a}{\sqrt{1 - \rho^2}} + \frac{\rho u}{\sqrt{1 - \rho^2}}\right) \ du \\
     &\geq \int_{0}^1 \phi(a-u) \ \Phi\left(\frac{b - \rho a}{\sqrt{1 - \rho^2}} + \frac{\rho u}{\sqrt{1 - \rho^2}}\right) \ du.
   \end{align*}
   $\phi(a-u) = \phi(a) \ \text{exp}(au - \frac{u^2}{2}) \geq \phi(a) \ \text{exp}(-\frac{1}{2}) \ \text{exp}(- |a|u)$ since $u \in [0,1]$ here. Furthermore, via a second-order Taylor expansion of $\log \Phi(z_{ab} - t)$ around $z_{ab}$ with $t = \frac{|\rho| u}{\sqrt{1 - \rho^2}}$, and using that $(\log\Phi)''> -1$ from step 1 of the proof of Proposition 1:
   $$
   \log \Phi\left(\frac{b - \rho a}{\sqrt{1 - \rho^2}} + \frac{\rho u}{\sqrt{1 - \rho^2}}\right) \geq \log \Phi(z_{ab}) - \frac{|\rho|u}{\sqrt{1 - \rho^2}} \frac{\phi(z_{ab})}{\Phi(z_{ab})} - \frac{\rho^2 u^2}{2(1 - \rho^2)}.
   $$
   Exponentiating and using $u \in [0,1]$ again:
   $$
   \Phi\left(\frac{b - \rho a}{\sqrt{1 - \rho^2}} + \frac{\rho u}{\sqrt{1 - \rho^2}}\right) \geq \Phi(z_{ab}) \ \text{exp}\left(- \frac{|\rho|u}{\sqrt{1 - \rho^2}} \frac{\phi(z_{ab})}{\Phi(z_{ab})}\right) \ \text{exp}\left( - \frac{\rho^2}{2(1 - \rho^2)}\right).
   $$
   Multiplying these lower bounds:
   $$
   \phi(a-u) \ \Phi\left(\frac{b - \rho a}{\sqrt{1 - \rho^2}} + \frac{\rho u}{\sqrt{1 - \rho^2}}\right) \geq \phi(a) \ \Phi(z_{ab}) \ \text{exp}\left(- \frac{1}{2} - \frac{\rho^2}{2(1 - \rho^2)}\right) \ \text{exp}\left( - Mu\right),
   $$
   with $M \coloneq |a| + \frac{|\rho|}{\sqrt{1 - \rho^2}} \frac{\phi(z_{ab})}{\Phi(z_{ab})}\geq 0$. Then, we can plug this in to obtain:
   \begin{align*}
     \Phi_2^{\rho}\big(a, b\big) &\geq \phi(a) \ \Phi(z_{ab}) \ \text{exp}\left(- \frac{1}{2} - \frac{\rho^2}{2(1 - \rho^2)}\right)\int_{0}^1 \text{exp}\left( - Mu\right) \ du \\
     &= \phi(a) \ \Phi(z_{ab}) \ \text{exp}\left(- \frac{1}{2} - \frac{\rho^2}{2(1 - \rho^2)}\right) \frac{1 - \text{exp}(-M)}{M} \\
     &\geq \phi(a) \ \Phi(z_{ab}) \ \text{exp}\left(- \frac{1}{2} - \frac{\rho^2}{2(1 - \rho^2)}\right) \frac{1}{1+M},
   \end{align*}
   which is the bound we wanted to show.

    \textbf{Showing (2.1.i):}
    \begin{align*}
        Q_0(\theta) - Q_0(\theta_0) &= \mathbb{E} \big[ Y_{ij} \log \Phi_2^{\rho}\big(W_{ij}^{'}\beta, W_{ji}^{'}\beta\big) + (1 - Y_{ij})\log \big(1 - \Phi_2^{\rho}\big(W_{ij}^{'}\beta, W_{ji}^{'}\beta\big)\big) \big] \\
        &- \mathbb{E} \big[ Y_{ij} \log \Phi_2^{\rho_0}\big(W_{ij}^{'}\beta_0, W_{ji}^{'}\beta_0\big) + (1 - Y_{ij})\log \big(1 - \Phi_2^{\rho_0}\big(W_{ij}^{'}\beta_0, W_{ji}^{'}\beta_0\big)\big) \big] \\
        &= \mathbb{E} \left[ \Phi_2^{\rho_0}\big(W_{ij}^{'}\beta_0, W_{ji}^{'}\beta_0\big) \log \frac{\Phi_2^{\rho}\big(W_{ij}^{'}\beta, W_{ji}^{'}\beta\big)}{\Phi_2^{\rho_0}\big(W_{ij}^{'}\beta_0, W_{ji}^{'}\beta_0\big)} \right.\\
        &\quad\quad+ \left.(1 - \Phi_2^{\rho_0}\big(W_{ij}^{'}\beta_0, W_{ji}^{'}\beta_0\big))\log \frac{1 - \Phi_2^{\rho}\big(W_{ij}^{'}\beta, W_{ji}^{'}\beta\big)}{1 - \Phi_2^{\rho_0}\big(W_{ij}^{'}\beta_0, W_{ji}^{'}\beta_0\big)} \right] \\
        &= - \mathbb{E} \left[\text{KL}\left(\text{Bernoulli}(\Phi_2^{\rho_0}\big(W_{ij}^{'}\beta_0, W_{ji}^{'}\beta_0\big)) \mid \text{Bernoulli}(\Phi_2^{\rho}\big(W_{ij}^{'}\beta, W_{ji}^{'}\beta\big)) \right) \right],
    \end{align*}
    where the second equality follows from the law of iterated expectation in combination with the model assumption \ref{ass:I1} and the third equality follows from the definition of the KL-divergence between two Bernoulli random variables. It is well-known that $\text{KL} \geq 0$ with $\text{KL} = 0$ iff $\Phi_2^{\rho_0}\big(W_{ij}^{'}\beta_0, W_{ji}^{'}\beta_0\big) = \Phi_2^{\rho}\big(W_{ij}^{'}\beta, W_{ji}^{'}\beta\big)$. Therefore, after taking expectations, $Q_0(\theta_0) \geq Q_0(\theta)$ with strict equality iff $\mathbb{P}\left[\Phi_2^{\rho_0}\big(W_{ij}^{'}\beta_0, W_{ji}^{'}\beta_0\big) = \Phi_2^{\rho}\big(W_{ij}^{'}\beta, W_{ji}^{'}\beta\big) \right] = 1$. That is, $Q_0(\theta)$ is uniquely maximized at $\theta_0$ under assumption \ref{ass:C2}, as required.

    One subtlety is that we need:
    $$
    \mathbb{E} \left| Y_{ij} \log \Phi_2^{\rho}\big(W_{ij}^{'}\beta, W_{ji}^{'}\beta\big) + (1 - Y_{ij})\log \big(1 - \Phi_2^{\rho}\big(W_{ij}^{'}\beta, W_{ji}^{'}\beta\big)\big) \right| < \infty,
    $$
    for every $\theta \in \Theta$. To this end, take logs of \eqref{P-lower-bound} to obtain:
    $$
    \log\Phi_2^{\rho}\big(a, b\big) \geq \log\phi(a) \ + \log\Phi(z_{ab}) -\frac{1}{2} - \frac{\rho^2}{2(1-\rho^2)}  - \log( 1 + M).
    $$
    Then note that $\log\phi(a) = -\frac{a^2}{2} - \log\sqrt{2\pi}$ and $\log\Phi(z_{ab}) \geq \log\phi(z_{ab}) - \log(1 + |z_{ab}|)$. The latter comes from taking logs of $\lambda(z) < 1 + |z|$ from step 1 of the proof of Proposition 1. In the first line below we plug both these in and flip the inequality; in line two we use that $\log(1 + |z_{ab}|) \leq |z_{ab}|$, $\log(1 + M) \leq M$ since $M \geq 0$, and plug in $M$; in line three we use $\lambda(z) < 1 + |z|$ again; and for the final line we plug in $z_{ab}$ and use that $|x| \leq 1 + x^2$ to produce a bunch of constants and $a^2$ and $b^2$ terms that can be grouped:
    \begin{align*}
        -\log\Phi_2^{\rho}\big(a, b\big) &\leq \frac{a^2}{2} + \log\sqrt{2\pi} +\frac{z_{ab}^2}{2} + \log\sqrt{2\pi} + \log(1 + |z_{ab}|) +\frac{1}{2} + \frac{\rho^2}{2(1-\rho^2)}  + \log( 1 + M) \\
        &\leq \log2\pi + \frac{1}{2} + \frac{\rho^2}{2(1-\rho^2)} + \frac{a^2}{2} + \frac{z_{ab}^2}{2} + |z_{ab}| + |a| + \frac{|\rho|}{\sqrt{1 - \rho^2}} \frac{\phi(z_{ab})}{\Phi(z_{ab})} \\
        &\leq \log2\pi + \frac{1}{2} + \frac{\rho^2}{2(1-\rho^2)} + \frac{a^2}{2} + \frac{z_{ab}^2}{2} + |z_{ab}| + |a| + \frac{|\rho|}{\sqrt{1 - \rho^2}} + \frac{|\rho|}{\sqrt{1 - \rho^2}} |z_{ab}| \\
        &\leq C_\delta(1 + a^2 + b^2),
    \end{align*}
    with $C_\delta$ chosen appropriately. Now insert $a = W_{ij}^{'}\beta$ and $b = W_{ji}^{'}\beta$. $|a| \leq \|W_{ij}\| \|\beta\| \leq \bar{B} \|W_{ij}\|$ and $|b| \leq \|W_{ji}\| \|\beta\| \leq \bar{B} \|W_{ji}\|$ for some $\bar{B} < \infty$ by the compactness assumption \ref{ass:C1}, plus Cauchy-Schwarz. So certainly:
    $$
    -\log\Phi_2^{\rho}\big(W_{ij}^{'}\beta, W_{ji}^{'}\beta\big) \leq C R_{ij}^2,
    $$
    where $R_{ij} \coloneq 1 + \|W_{ij}\| + \|W_{ji}\|$, with $C$ chosen appropriately again. We need a similar bound for $-\log(1 - \Phi_2^{\rho}\big(W_{ij}^{'}\beta, W_{ji}^{'}\beta\big))$. Simply note that $1 - \Phi_2^{\rho}\big(W_{ij}^{'}\beta, W_{ji}^{'}\beta\big) \geq 1 - \Phi\big(W_{ij}^{'}\beta\big)$, so $-\log(1 - \Phi_2^{\rho}\big(W_{ij}^{'}\beta, W_{ji}^{'}\beta\big)) \leq -\log(1 - \Phi\big(W_{ij}^{'}\beta\big)) = -\log \Phi\big(-W_{ij}^{'}\beta\big) \leq C(1 + \bar{B}^2 \|W_{ij}\|^2) \leq C R_{ij}^2$. Combined:
    $$
    \left| Y_{ij} \log \Phi_2^{\rho}\big(W_{ij}^{'}\beta, W_{ji}^{'}\beta\big) + (1 - Y_{ij})\log \big(1 - \Phi_2^{\rho}\big(W_{ij}^{'}\beta, W_{ji}^{'}\beta\big)\big) \right| \leq C R_{ij}^2,
    $$
    so its expectation is finite under assumption \ref{ass:C3}, as required.

    \textbf{Showing (2.1.ii):} Assumed via \ref{ass:C1}.
    
    \textbf{Showing (2.1.iii) and (2.1.iv):} We apply Lemma 1 with $h_{ij}(\theta) = Y_{ij} \log \Phi_2^{\rho}\big(W_{ij}^{'}\beta, W_{ji}^{'}\beta\big) + (1 - Y_{ij})\log \big(1 - \Phi_2^{\rho}\big(W_{ij}^{'}\beta, W_{ji}^{'}\beta\big)\big)$. $h_{ij}(\theta)$ is clearly symmetric because $\Phi_2^{\rho}\big(a, b\big) = \Phi_2^{\rho}\big(b, a\big)$. And $h_{ij}(\theta)$ is continuous in $\theta$ because $\Phi_2^{\rho}\big(W_{ij}^{'}\beta, W_{ji}^{'}\beta\big)$ is continuous in $\beta$ and $\rho$ and lives in $(0,1)$, thus avoiding $\log0$. Finally, since $|h_{ij}(\theta)| \leq CR_{ij}^2$ for all $\theta$ from above, $\mathbb{E}\left[\sup_{\theta \in \Theta} \left|h_{ij}(\theta) \right|^{1+\eta} \right] \leq C \mathbb{E} \left[ R_{ij}^{2(1 + \eta)}\right] < \infty$ by \ref{ass:C3} for some $\eta \in (0, \frac{\delta}{2}]$. So the conditions of Lemma 1 hold and therefore $\hat{Q}_n(\theta)$ converges uniformly in probability to $Q_0 (\theta)$, with $Q_0 (\theta)$ continuous, as required.

\end{proof}

\begin{lemma}
    Let $\psi(y, x, x'; \theta)$ be a symmetric kernel. Assume \ref{ass:I1} holds and write $\psi_{ij}(\theta) \coloneq \psi(Y_{ij}, X_i, X_j; \theta)$. If $\mathbb{E} \left[\psi_{ij} \mid \mathcal{X}\right] = 0$ and $\mathbb{E} \|\psi_{ij} \|^{2+\eta} < \infty$ for some $\eta > 0$, then:
    $$
    \frac{1}{\sqrt{N}}\sum^{n-1}_{i=1} \sum^{n}_{j=i+1} \psi_{ij} \xrightarrow{d} \mathcal{N}(0, \mathbb{E}[\psi_{ij}\psi_{ij}']).
    $$
\end{lemma}

\begin{proof}[Proof of Lemma 2]
 \vspace{1ex}
 The proof follows three steps:
 \begin{enumerate}
     \item Establish two almost sure limits.
     \item Fix a realization of $\mathcal{X}$ and apply the Lindeberg–Feller CLT to the conditional law
     \item De-condition via dominated convergence.
 \end{enumerate}

 \textbf{Step 1:} Consider the kernels
    \begin{align*}
    K_1(x, x') &\coloneqq \mathbb{E}\left[\psi_{ij}\psi_{ij}' \mid X_i = x,\ X_j = x'\right], \\
    K_2(x, x') &\coloneqq \mathbb{E}\left[\|\psi_{ij}\|^{2+\eta} \mid X_i = x,\ X_j = x'\right].
    \end{align*}
    Applying the strong law of large numbers for U-statistics (Theorem 5.4.A) in \cite{serfling_approximation_1980} again to both kernels, since (i) both are independent of $Y_{ij}$ and $n$, (ii) $X_i$ are i.i.d., (iii) $\psi_{ij}$ is symmetric by assumption, and (iv) $\mathbb{E} \|\psi_{ij} \|^{2+\eta} < \infty$ by assumption, we have that:
    \begin{align*}
    \frac{1}{N}\sum^{n-1}_{i=1} \sum^{n}_{j=i+1}K_1(X_i, X_j) &\xrightarrow{a.s.} \mathbb{E}\left[\psi_{ij}\psi_{ij}'\right], \\
    \frac{1}{N}\sum^{n-1}_{i=1} \sum^{n}_{j=i+1}K_2(X_i, X_j) &\xrightarrow{a.s.} \mathbb{E}\|\psi_{ij}\|^{2+\eta}.
    \end{align*}

    \textbf{Step 2:} Conditional on $\mathcal{X}$, $\psi_{ij}$ are independent across dyads via assumption \ref{ass:I1}. And $\mathbb{E} \left[\psi_{ij} \mid \mathcal{X}\right] = 0$ by assumption. So by the Lindeberg-Feller CLT (Proposition 2.27) in \cite{van_der_vaart_asymptotic_1998} applied to a fixed realization of $\mathcal{X}$, the conclusion of Lemma 2 holds conditional on $\mathcal{X}$ provided:
    \begin{enumerate}
        \item[(i)] $\sum^{n-1}_{i=1} \sum^{n}_{j=i+1} \mathbb{E} \left[\|\frac{1}{\sqrt{N}}\psi_{ij} \|^2 \mathds{1}[\|\frac{1}{\sqrt{N}}\psi_{ij} \| > \epsilon] \mid \mathcal{X} \right] \xrightarrow{a.s.} 0$ for every $\epsilon > 0$,
        \item[(ii)] $\sum^{n-1}_{i=1} \sum^{n}_{j=i+1} \mathbb{C}\text{ov} \left[\frac{1}{\sqrt{N}}\psi_{ij}  \mid \mathcal{X} \right] \xrightarrow{a.s.} \mathbb{E}\left[\psi_{ij}\psi_{ij}'\right]$.
    \end{enumerate}
    To show (i), bound the left hand side as such:
    \begin{align*}
        \text{LHS} & \leq \frac{1}{\epsilon^\eta} \sum^{n-1}_{i=1} \sum^{n}_{j=i+1} \mathbb{E} \left[\|\frac{1}{\sqrt{N}}\psi_{ij} \|^{2+\eta} \mid \mathcal{X} \right] \\
        &=\frac{N^{-\frac{\eta}{2}}}{\epsilon^\eta} \cdot \frac{1}{N} \sum^{n-1}_{i=1} \sum^{n}_{j=i+1} \mathbb{E} \left[\|\psi_{ij} \|^{2+\eta} \mid \mathcal{X} \right].
    \end{align*}
    Then, as $n \to \infty$, the sample average converges almost surely to a finite constant via the second limit in step 1, and $\frac{N^{-\frac{\eta}{2}}}{\epsilon^\eta} \to 0$ since $\eta > 0$. Because $\epsilon > 0$ was arbitrary, (i) holds. To show (ii), simply note that its left hand side equals $\frac{1}{N} \sum^{n-1}_{i=1} \sum^{n}_{j=i+1} \mathbb{E}\left[\psi_{ij}\psi_{ij}' \mid \mathcal{X}\right]$, which we know converges almost surely to $\mathbb{E}\left[\psi_{ij}\psi_{ij}'\right]$ from the first limit in step 1. So (ii) holds too, as required.

    \textbf{Step 3:} By the portmanteau theorem, the conditional convergence in distribution result from step 2 is equivalent to:
    $$
    \mathbb{E}\left[f\left(\frac{1}{\sqrt{N}} \sum^{n-1}_{i=1} \sum^{n}_{j=i+1}\psi_{ij}  \right) \mid \mathcal{X}\right] \to \mathbb{E} \left[f(Z)\right],
    $$
    for every bounded and continuous function $f$, where $Z \sim \mathcal{N}$. Thus, by the law of iterated expectation and dominated convergence theorem given $| \mathbb{E}[f(\cdot) \mid \mathcal{X}] | \leq \sup |f| < \infty$ since $f$ is bounded:
    \begin{align*}
        \mathbb{E}\left[f\left(\frac{1}{\sqrt{N}} \sum^{n-1}_{i=1} \sum^{n}_{j=i+1}\psi_{ij}  \right)\right] &= \mathbb{E}\left[\mathbb{E}\left[f\left(\frac{1}{\sqrt{N}} \sum^{n-1}_{i=1} \sum^{n}_{j=i+1}\psi_{ij}  \right) \mid \mathcal{X}\right] \right] \\
        &\to \mathbb{E}\left[\mathbb{E}\left[f\left(Z  \right)\right] \right] \\
        &=\mathbb{E}\left[f\left(Z  \right)\right].
    \end{align*}
    Therefore, the unconditional convergence in distribution result stated by Lemma 2 is achieved via another application of the portmanteau theorem.

\end{proof}

\begin{proof}[Proof of Theorem 2]
 \vspace{1ex}  % forces a line break
 We establish the conditions of Theorem 3.1 in \cite{newey_chapter_1994}, which states that if $\hat{\theta} \xrightarrow{p}\theta_0$ and:
\begin{enumerate}
    \item[(3.1.i)] $\theta_0 \in \text{int}(\Theta)$, 
    \item[(3.1.ii)] $\hat{Q}_n(\theta)$ is twice continuously differentiable on a neighborhood $\mathcal{N}_0$ of $\theta_0$,
    \item[(3.1.iii)] $\sqrt{N} \nabla_{\theta}\hat{Q}_n(\theta_0) \xrightarrow{d} \mathcal{N}(0, \Sigma)$,
    \item[(3.1.iv)] There exists $H(\theta)$ that is continuous at $\theta_0$ and satisfies $\sup_{\theta \in \mathcal{N}_0} \left\|\nabla_{\theta}^2\hat{Q}_n(\theta) - H(\theta)\right\| \xrightarrow{p} 0$,
    \item[(3.1.v)] $H(\theta_0)$ is nonsingular,
\end{enumerate}
then $\sqrt{N}(\hat{\theta} - \theta_0) \xrightarrow{d} \mathcal{N}(0, H(\theta_0)^{-1}\Sigma H(\theta_0)^{-1})$.

\textbf{Showing (3.1.i):} Assumed via \ref{ass:C1}.

\textbf{Showing (3.1.ii):} $\hat{Q}_n(\theta)$ is $C^\infty$ on all of $\Theta$ since $\Phi_2^{\rho}\big(W_{ij}^{'}\beta, W_{ji}^{'}\beta\big)$ is $C^\infty$ and lives in $(0,1)$, thus avoiding $\log0$. So $\hat{Q}_n(\theta)$ is certainly $C^2$ on a neighborhood.

\textbf{Showing (3.1.iii):} We apply Lemma 2 with:
\begin{align*}
    \psi_{ij}(\theta) &= \nabla_\theta\left[Y_{ij} \log \Phi_2^{\rho}\big(W_{ij}^{'}\beta, W_{ji}^{'}\beta\big) + (1 - Y_{ij})\log \big(1 - \Phi_2^{\rho}\big(W_{ij}^{'}\beta, W_{ji}^{'}\beta\big)\big) \right] \\[1ex]
    &=\frac{Y_{ij} - \Phi_2^{\rho}\big(W_{ij}^{'}\beta, W_{ji}^{'}\beta\big)}{\Phi_2^{\rho}\big(W_{ij}^{'}\beta, W_{ji}^{'}\beta\big)\big(1 - \Phi_2^{\rho}\big(W_{ij}^{'}\beta, W_{ji}^{'}\beta\big)\big)}
    \begin{pmatrix}
    \partial_a \Phi_2^{\rho}\big(W_{ij}^{'}\beta, W_{ji}^{'}\beta\big)\, W_{ij}
    + \partial_b \Phi_2^{\rho}\big(W_{ij}^{'}\beta, W_{ji}^{'}\beta\big)\, W_{ji} \\[1ex]
    \phi_2^{\rho}\big(W_{ij}^{'}\beta, W_{ji}^{'}\beta\big)
    \end{pmatrix}.
\end{align*}
$\psi_{ij}(\theta)$ is symmetric because $\Phi_2^{\rho}\big(W_{ij}^{'}\beta, W_{ji}^{'}\beta\big) = \Phi_2^{\rho}\big(W_{ji}^{'}\beta, W_{ij}^{'}\beta\big)$ and $\phi_2^{\rho}\big(W_{ij}^{'}\beta, W_{ji}^{'}\beta\big) = \phi_2^{\rho}\big(W_{ji}^{'}\beta, W_{ij}^{'}\beta\big)$. And $\mathbb{E} \left[\psi_{ij} \mid \mathcal{X}\right] = 0$ because $\mathbb{E} \left[Y_{ij} \mid \mathcal{X}\right] = \Phi_2^{\rho}\big(W_{ij}^{'}\beta, W_{ji}^{'}\beta\big)$ by \ref{ass:I1}. We remain to show that $\mathbb{E} \|\psi_{ij} \|^{2+\eta} < \infty$ for some $\eta > 0$. To this end, note that:
$$
\|\psi_{ij} \| \leq \max\left\{\frac{\|\nabla_\theta\Phi_2^{\rho}\big(W_{ij}^{'}\beta, W_{ji}^{'}\beta\big)\|}{\Phi_2^{\rho}\big(W_{ij}^{'}\beta, W_{ji}^{'}\beta\big)}, \frac{\|\nabla_\theta\Phi_2^{\rho}\big(W_{ij}^{'}\beta, W_{ji}^{'}\beta\big)\|}{1 - \Phi_2^{\rho}\big(W_{ij}^{'}\beta, W_{ji}^{'}\beta\big)}\right\}.
$$
We bound each term in the max. 

First:
$$
\frac{\|\nabla_\theta\Phi_2^{\rho}\big(W_{ij}^{'}\beta, W_{ji}^{'}\beta\big)\|}{\Phi_2^{\rho}\big(W_{ij}^{'}\beta, W_{ji}^{'}\beta\big)} \leq \frac{\partial_a \Phi_2^{\rho}\big(W_{ij}^{'}\beta, W_{ji}^{'}\beta\big)}{\Phi_2^{\rho}\big(W_{ij}^{'}\beta, W_{ji}^{'}\beta\big)} \|W_{ij}\| + \frac{\partial_b \Phi_2^{\rho}\big(W_{ij}^{'}\beta, W_{ji}^{'}\beta\big)}{\Phi_2^{\rho}\big(W_{ij}^{'}\beta, W_{ji}^{'}\beta\big)} \|W_{ji}\| + \frac{\phi_2^{\rho}\big(W_{ij}^{'}\beta, W_{ji}^{'}\beta\big)}{\Phi_2^{\rho}\big(W_{ij}^{'}\beta, W_{ji}^{'}\beta\big)}.
$$
From here, we use two more inequalities (plus an analogous one to \eqref{partial-a-over-P} for $\partial_b$):
\begin{equation}\label{partial-a-over-P}
    \frac{\partial_a\Phi_2^{\rho}\big(a, b\big)}{\Phi_2^{\rho}\big(a, b\big)} \leq C_\delta(1 + |a| + |b|),
\end{equation}
\begin{equation}\label{p-over-P}
    \frac{\phi_2^{\rho}\big(a, b\big)}{\Phi_2^{\rho}\big(a, b\big)} \leq C_{\delta'}(1 + |a| + |b|)^2,
\end{equation}
which, when plugged in, imply that $\frac{\|\nabla_\theta\Phi_2^{\rho}\big(W_{ij}^{'}\beta, W_{ji}^{'}\beta\big)\|}{\Phi_2^{\rho}\big(W_{ij}^{'}\beta, W_{ji}^{'}\beta\big)} \leq CR_{ij}^2$. To prove \eqref{partial-a-over-P}, rearrange \eqref{P-lower-bound} given that $\partial_a\Phi_2^{\rho}\big(a, b\big) = \phi(a) \Phi(z_{ab})$ to obtain:
\begin{align*}
    \frac{\partial_a\Phi_2^{\rho}\big(a, b\big)}{\Phi_2^{\rho}\big(a, b\big)} &\leq \left(1 + |a| + \frac{|\rho|}{\sqrt{1 - \rho^2}} \frac{\phi(z_{ab})}{\Phi(z_{ab})}\right) \text{exp}\left(\frac{1}{2} + \frac{\rho^2}{2(1-\rho^2)} \right) \\
    &\leq\left(1 + |a| + \frac{|\rho|}{\sqrt{1 - \rho^2}} \left(1 + \frac{|a| +|b|}{\sqrt{1 - \rho^2}}\right)\right) \text{exp}\left(\frac{1}{2} + \frac{\rho^2}{2(1-\rho^2)} \right) \\
    &\leq C_\delta(1 + |a| + |b|),
\end{align*}
with $C_\delta$ chosen appropriately, where for the second line we use that $\lambda(z) < 1 + |z|$ and the definition of $z_{ab}$ as in the proof of Theorem 1. To prove \eqref{p-over-P}, use $\phi_2^{\rho}\big(a, b\big) = \frac{1}{\sqrt{1 - \rho^2}} \phi(a) \phi(z_{ab}) = \frac{1}{\sqrt{1 - \rho^2}} \partial_a\Phi_2^{\rho}\big(a, b\big) \frac{\phi(z_{ab})}{\Phi(z_{ab})}$ from step 2 of the proof of Proposition 1. So:
\begin{align*}
    \frac{\phi_2^{\rho}\big(a, b\big)}{\Phi_2^{\rho}\big(a, b\big)} &= \frac{1}{\sqrt{1 - \rho^2}} \frac{\partial_a\Phi_2^{\rho}\big(a, b\big)}{\Phi_2^{\rho}\big(a, b\big)} \frac{\phi(z_{ab})}{\Phi(z_{ab})} \\
    &\leq \frac{1}{\sqrt{1 - \rho^2}} C_\delta(1 + |a| + |b|)\left(1 + \frac{|a| +|b|}{\sqrt{1 - \rho^2}}\right)\\
    &\leq C_{\delta'}(1 + |a| + |b|)^2,
\end{align*}
with $C_{\delta'}$ chosen appropriately, where for the second line we use \eqref{partial-a-over-P} and $\lambda(z) < 1 + |z|$ again. 

Second:
\begin{align*}
    \frac{\|\nabla_\theta\Phi_2^{\rho}\big(W_{ij}^{'}\beta, W_{ji}^{'}\beta\big)\|}{1-\Phi_2^{\rho}\big(W_{ij}^{'}\beta, W_{ji}^{'}\beta\big)} &\leq \frac{\partial_a \Phi_2^{\rho}\big(W_{ij}^{'}\beta, W_{ji}^{'}\beta\big)}{1-\Phi_2^{\rho}\big(W_{ij}^{'}\beta, W_{ji}^{'}\beta\big)} \|W_{ij}\| + \frac{\partial_b \Phi_2^{\rho}\big(W_{ij}^{'}\beta, W_{ji}^{'}\beta\big)}{1-\Phi_2^{\rho}\big(W_{ij}^{'}\beta, W_{ji}^{'}\beta\big)} \|W_{ji}\| \\
    &+ \frac{\phi_2^{\rho}\big(W_{ij}^{'}\beta, W_{ji}^{'}\beta\big)}{1-\Phi_2^{\rho}\big(W_{ij}^{'}\beta, W_{ji}^{'}\beta\big)}.
\end{align*}
From here, we use two more inequalities (plus an analogous one to \eqref{partial-a-over-1-minus-P} for $\partial_b$):
\begin{equation}\label{partial-a-over-1-minus-P}
    \frac{\partial_a\Phi_2^{\rho}\big(a, b\big)}{1-\Phi_2^{\rho}\big(a, b\big)} \leq 1 + |a|,
\end{equation}
\begin{equation}\label{p-over-1-minus-P}
    \frac{\phi_2^{\rho}\big(a, b\big)}{1-\Phi_2^{\rho}\big(a, b\big)} \leq C_{\delta}(1 + |a|),
\end{equation}
which, when plugged in, imply that $\frac{\|\nabla_\theta\Phi_2^{\rho}\big(W_{ij}^{'}\beta, W_{ji}^{'}\beta\big)\|}{1 - \Phi_2^{\rho}\big(W_{ij}^{'}\beta, W_{ji}^{'}\beta\big)} \leq CR_{ij}^2$. To prove \eqref{partial-a-over-1-minus-P}, use $1-\Phi_2^{\rho}\big(a, b\big) \geq 1 - \Phi(a)$, $\partial_a\Phi_2^{\rho}\big(a, b\big) = \phi(a)\Phi(z_{ab})$, $\Phi(z_{ab}) \in (0,1)$, and $\lambda(z) < 1 + |z|$ again:
\begin{align*}
    \frac{\partial_a\Phi_2^{\rho}\big(a, b\big)}{1-\Phi_2^{\rho}\big(a, b\big)} \leq \frac{\phi(a)\Phi(z_{ab})}{1 - \Phi(a)} \leq \frac{\phi(a)}{1 - \Phi(a)} = \lambda(-a) \leq 1 + |a|.
\end{align*}
Similarly, to prove \eqref{p-over-1-minus-P}, use the same facts plus $\phi_2^{\rho}\big(a, b\big) = \frac{1}{\sqrt{1 - \rho^2}} \phi(a) \phi(z_{ab})$:
\begin{align*}
    \frac{\phi_2^{\rho}\big(a, b\big)}{1-\Phi_2^{\rho}\big(a, b\big)} \leq \frac{1}{\sqrt{1 - \rho^2}}\frac{\phi(a)\phi(z_{ab})}{1 - \Phi(a)} = \frac{\phi(z_{ab})}{\sqrt{1 - \rho^2}}\lambda(-a) \leq \frac{\phi_{\text{max}}}{\sqrt{1 - \rho^2}}(1 + |a|).
\end{align*}

All together, $\|\psi_{ij} \| \leq C R_{ij}^2$, so $\mathbb{E} \|\psi_{ij} \|^{2+\eta} \leq C \mathbb{E} \left[ R_{ij}^{4 + 2\eta}\right] < \infty$ by \ref{ass:N1} for some $\eta \in (0, \frac{\delta}{2}]$, as required.

The final requirement for this step is that:
$$
\Sigma = \mathbb{E}[\psi_{ij}\psi_{ij}'] = \mathbb{E} \left[\frac{\nabla_{\theta}\Phi_2^{\rho}\big(W_{ij}^{'}\beta, W_{ji}^{'}\beta\big)\left(\nabla_{\theta}\Phi_2^{\rho}\big(W_{ij}^{'}\beta, W_{ji}^{'}\beta\big)\right)'}{\Phi_2^{\rho}\big(W_{ij}^{'}\beta, W_{ji}^{'}\beta\big)(1-\Phi_2^{\rho}\big(W_{ij}^{'}\beta, W_{ji}^{'}\beta\big))} \right]
$$
is finite. To this end, split the term inside the expectation into partial fractions to bound its magnitude like so:
\begin{align*}
    \frac{\|\nabla_\theta\Phi_2^{\rho}\big(W_{ij}^{'}\beta, W_{ji}^{'}\beta\big)\|^2}{\Phi_2^{\rho}\big(W_{ij}^{'}\beta, W_{ji}^{'}\beta\big)(1-\Phi_2^{\rho}\big(W_{ij}^{'}\beta, W_{ji}^{'}\beta\big))} &= \|\nabla_\theta\Phi_2^{\rho}\big(W_{ij}^{'}\beta, W_{ji}^{'}\beta\big)\| \left[\frac{\|\nabla_\theta\Phi_2^{\rho}\big(W_{ij}^{'}\beta, W_{ji}^{'}\beta\big)\|}{\Phi_2^{\rho}\big(W_{ij}^{'}\beta, W_{ji}^{'}\beta\big)}\right. \\
    &\quad\quad\quad\quad\quad\quad\quad\quad\quad\quad\quad+ \left.\frac{\|\nabla_\theta\Phi_2^{\rho}\big(W_{ij}^{'}\beta, W_{ji}^{'}\beta\big)\|}{1-\Phi_2^{\rho}\big(W_{ij}^{'}\beta, W_{ji}^{'}\beta\big)}\right] \\
    &\leq \|\nabla_\theta\Phi_2^{\rho}\big(W_{ij}^{'}\beta, W_{ji}^{'}\beta\big)\| \left[C R_{ij}^2 + C R_{ij}^2\right],
\end{align*}
where for the final line we use the bounds from above. Then note that:
\begin{align*}
    \|\nabla_\theta\Phi_2^{\rho}\big(W_{ij}^{'}\beta, W_{ji}^{'}\beta\big)\| &\leq \partial_a \Phi_2^{\rho}\big(W_{ij}^{'}\beta, W_{ji}^{'}\beta\big) \|W_{ij}\| + \partial_b \Phi_2^{\rho}\big(W_{ij}^{'}\beta, W_{ji}^{'}\beta\big) \|W_{ji}\| + \phi_2^{\rho}\big(W_{ij}^{'}\beta, W_{ji}^{'}\beta\big) \\
    &=\phi(a)\Phi(z_{ab}) \|W_{ij}\| + \phi(b)\Phi(z_{ba}) \|W_{ji}\| + \frac{1}{\sqrt{1 - \rho^2}} \phi(a) \phi(z_{ab}) \\
    &\leq \phi_{\text{max}}\|W_{ij}\| + \phi_{\text{max}}\|W_{ji}\| + \frac{\phi_{\text{max}}^2}{\sqrt{1 - \rho^2}} \\
    &\leq CR_{ij}.
\end{align*}
So $\|\Sigma \| \leq C \mathbb{E} \left[ R_{ij}^{3}\right] < \infty$ under \ref{ass:N1}, as required.

\textbf{Showing (3.1.iv):} We apply Lemma 1 with:
\begin{align*}
    h_{ij}(\theta) &= \nabla_\theta^2\left[Y_{ij} \log \Phi_2^{\rho}\big(W_{ij}^{'}\beta, W_{ji}^{'}\beta\big) + (1 - Y_{ij})\log \big(1 - \Phi_2^{\rho}\big(W_{ij}^{'}\beta, W_{ji}^{'}\beta\big)\big) \right] \\[1ex]
    &=\frac{Y_{ij} - \Phi_2^{\rho}\big(W_{ij}^{'}\beta, W_{ji}^{'}\beta\big)}{\Phi_2^{\rho}\big(W_{ij}^{'}\beta, W_{ji}^{'}\beta\big)\big(1 - \Phi_2^{\rho}\big(W_{ij}^{'}\beta, W_{ji}^{'}\beta\big)\big)}\nabla_\theta^2\Phi_2^{\rho}\big(W_{ij}^{'}\beta, W_{ji}^{'}\beta\big) \\
    &- \left[\frac{Y_{ij}}{\Phi_2^{\rho}\big(W_{ij}^{'}\beta, W_{ji}^{'}\beta\big)^2} + \frac{1 - Y_{ij}}{\big(1 - \Phi_2^{\rho}\big(W_{ij}^{'}\beta, W_{ji}^{'}\beta\big)\big)^2} \right]\nabla_{\theta}\Phi_2^{\rho}\big(W_{ij}^{'}\beta, W_{ji}^{'}\beta\big)\left(\nabla_{\theta}\Phi_2^{\rho}\big(W_{ij}^{'}\beta, W_{ji}^{'}\beta\big)\right)'.
\end{align*}
$h_{ij}(\theta)$ is symmetric because $\Phi_2^{\rho}\big(a, b\big) = \Phi_2^{\rho}\big(b, a\big)$, $\nabla_\theta\Phi_2^{\rho}\big(a, b\big) = \nabla_\theta\Phi_2^{\rho}\big(b, a\big)$, and $\nabla_\theta^2\Phi_2^{\rho}\big(a, b\big) = \nabla_\theta^2\Phi_2^{\rho}\big(b, a\big)$ (see below for the derivation of $\nabla_\theta^2\Phi_2^{\rho}\big(a, b\big)$).  And $h_{ij}(\theta)$ is continuous in $\theta$ because $\Phi_2^{\rho}\big(W_{ij}^{'}\beta, W_{ji}^{'}\beta\big)$ is $C^\infty$ and lives in $(0,1)$, thus avoiding $\log0$. We remain to bound $\mathbb{E}\left[\sup_{\theta \in \Theta} \|h_{ij}(\theta) \|^{1+\eta} \right]$. To this end, note that the second term in the expression for $h_{ij}(\theta)$ is bounded by:
$$
\max\left\{\frac{\|\nabla_\theta\Phi_2^{\rho}\big(W_{ij}^{'}\beta, W_{ji}^{'}\beta\big)\|}{\Phi_2^{\rho}\big(W_{ij}^{'}\beta, W_{ji}^{'}\beta\big)}, \frac{\|\nabla_\theta\Phi_2^{\rho}\big(W_{ij}^{'}\beta, W_{ji}^{'}\beta\big)\|}{1 - \Phi_2^{\rho}\big(W_{ij}^{'}\beta, W_{ji}^{'}\beta\big)}\right\}^2 \leq C R_{ij}^4,
$$
using the bound on $\| \psi_{ij}\| $ from the (3.1.iii) step. For the first term in the $h_{ij}(\theta)$ expression, we need the second derivatives of $\Phi_2^{\rho}\big(W_{ij}^{'}\beta, W_{ji}^{'}\beta\big)$:
    \begin{align*}
    \partial_a^2\Phi_2^{\rho}\big(a, b\big) &= -a\phi(a)\Phi(z_{ab}) - \frac{\rho}{\sqrt{1 - \rho^2}}\phi(a)\phi(z_{ab})=-a\partial_a\Phi_2^{\rho}\big(a, b\big) - \rho\phi_2^{\rho}\big(a, b\big), \\
    \partial_{ab}\Phi_2^{\rho}\big(a, b\big) &= \frac{1}{\sqrt{1 - \rho^2}}\phi(a)\phi(z_{ab}) = \phi_2^{\rho}\big(a, b\big), \\
    \partial_{a\rho}\Phi_2^{\rho}\big(a, b\big) &= \frac{1}{\sqrt{1 - \rho^2}}\phi(a)\phi(z_{ab})\left[\frac{\rho}{\sqrt{1 - \rho^2}}z_{ab} - a \right] = \phi_2^{\rho}\big(a, b\big)\left[\frac{\rho}{\sqrt{1 - \rho^2}}z_{ab} - a \right],\\
    \partial_\rho^2\Phi_2^{\rho}\big(a, b\big) &= \phi_2^{\rho}\big(a, b\big) \partial_\rho\log\phi_2^{\rho}\big(a, b\big) \\
    &=\phi_2^{\rho}\big(a, b\big) \left[\frac{\rho}{1 - \rho^2} + \frac{ab(1 - \rho^2) - \rho (a^2 - 2\rho ab + b^2)}{(1 - \rho^2)^2} \right].
    \end{align*}
    Notice how $\partial_{a\rho}\Phi_2^{\rho}\big(a, b\big)$ is bounded in absolute value by $C \phi_2^{\rho}\big(a, b\big) R_{ij}$ and $\partial_\rho^2\Phi_2^{\rho}\big(a, b\big)$ is bounded in absolute value by $C \phi_2^{\rho}\big(a, b\big) R_{ij}^2$. So:
    \begin{align*}
        \frac{\|\nabla_\theta^2\Phi_2^{\rho}\big(W_{ij}^{'}\beta, W_{ji}^{'}\beta\big)\|}{\Phi_2^{\rho}\big(W_{ij}^{'}\beta, W_{ji}^{'}\beta\big)} &\leq \frac{|a|\partial_a\Phi_2^{\rho}\big(a, b\big) \|W_{ij}\|^2}{\Phi_2^{\rho}\big(W_{ij}^{'}\beta, W_{ji}^{'}\beta\big)} + \frac{\phi_2^{\rho}\big(a, b\big) \|W_{ij}\|^2}{\Phi_2^{\rho}\big(W_{ij}^{'}\beta, W_{ji}^{'}\beta\big)} \\
        &+ \frac{|b|\partial_b\Phi_2^{\rho}\big(a, b\big) \|W_{ji}\|^2}{\Phi_2^{\rho}\big(W_{ij}^{'}\beta, W_{ji}^{'}\beta\big)} + \frac{\phi_2^{\rho}\big(a, b\big) \|W_{ji}\|^2}{\Phi_2^{\rho}\big(W_{ij}^{'}\beta, W_{ji}^{'}\beta\big)} \\
        &+ 2 \frac{\phi_2^{\rho}\big(a, b\big) \|W_{ij}\| \|W_{ji}\|}{\Phi_2^{\rho}\big(W_{ij}^{'}\beta, W_{ji}^{'}\beta\big)} + 2C \frac{\phi_2^{\rho}\big(a, b\big)}{\Phi_2^{\rho}\big(W_{ij}^{'}\beta, W_{ji}^{'}\beta\big)} R_{ij}(\|W_{ij}\| + \|W_{ji}\|) \\
        &+ C\frac{\phi_2^{\rho}\big(a, b\big)}{\Phi_2^{\rho}\big(W_{ij}^{'}\beta, W_{ji}^{'}\beta\big)} R_{ij}^2.
    \end{align*}
    Then we can use $\frac{\partial_a\Phi_2^{\rho}\big(a, b\big) }{\Phi_2^{\rho}\big(W_{ij}^{'}\beta, W_{ji}^{'}\beta\big)} \leq CR_{ij}$, $\frac{\partial_b\Phi_2^{\rho}\big(a, b\big) }{\Phi_2^{\rho}\big(W_{ij}^{'}\beta, W_{ji}^{'}\beta\big)} \leq CR_{ij}$, and $\frac{\phi_2^{\rho}\big(a, b\big)}{\Phi_2^{\rho}\big(W_{ij}^{'}\beta, W_{ji}^{'}\beta\big)} \leq CR_{ij}^2$ from above to see that:
    $$
    \frac{\|\nabla_\theta^2\Phi_2^{\rho}\big(W_{ij}^{'}\beta, W_{ji}^{'}\beta\big)\|}{\Phi_2^{\rho}\big(W_{ij}^{'}\beta, W_{ji}^{'}\beta\big)} \leq CR_{ij}^4,
    $$
    and likewise for $\frac{\|\nabla_\theta^2\Phi_2^{\rho}\big(W_{ij}^{'}\beta, W_{ji}^{'}\beta\big)\|}{1-\Phi_2^{\rho}\big(W_{ij}^{'}\beta, W_{ji}^{'}\beta\big)}$. Therefore, all combined, $\|h_{ij}(\theta)\| \leq CR_{ij}^4$ for all $\theta \in \Theta$, and thus $\mathbb{E}\left[\sup_{\theta \in \Theta} \|h_{ij}(\theta) \|^{1+\eta} \right] < \infty$ under assumption \ref{ass:N1} provided $\eta \leq \frac{\delta}{4}$. In sum, via Lemma 1, $H(\theta) \coloneq \mathbb{E}[h_{ij}(\theta)]$ is continuous and:
    $$
    \sup_{\theta \in \mathcal{N}_0} \left\|\nabla_{\theta}^2\hat{Q}_n(\theta) - H(\theta)\right\| = \sup_{\theta \in \mathcal{N}_0} \left\|\frac{1}{N}\sum^{n-1}_{i=1} \sum^{n}_{j=i+1}h_{ij}(\theta) - \mathbb{E}[h_{ij}(\theta)]\right\| \xrightarrow{p} 0,
    $$
    as required.

    \textbf{Showing (3.1.v):} Simply note that:
    $$H
    (\theta_0) =\mathbb{E}[\mathbb{E}[h_{ij}(\theta_0)\mid \mathcal{X}]] = - \mathbb{E} \left[\frac{\nabla_{\theta}\Phi_2^{\rho_0}\big(W_{ij}^{'}\beta_0, W_{ji}^{'}\beta_0\big)\left(\nabla_{\theta}\Phi_2^{\rho_0}\big(W_{ij}^{'}\beta_0, W_{ji}^{'}\beta_0\big)\right)'}{\Phi_2^{\rho_0}\big(W_{ij}^{'}\beta_0, W_{ji}^{'}\beta_0\big)(1-\Phi_2^{\rho_0}\big(W_{ij}^{'}\beta_0, W_{ji}^{'}\beta_0\big))} \right] = -J,
    $$
    and $J$ is nonsingular by assumption \ref{ass:N2}.

    \textbf{Concluding point (i) of the Theorem:} By Theorem 3.1 in \cite{newey_chapter_1994}, $\sqrt{N}(\hat{\theta} - \theta_0) \xrightarrow{d} \mathcal{N}(0, H(\theta_0)^{-1}\Sigma H(\theta_0)^{-1}) = \mathcal{N}(0, J^{-1}J J^{-1}) = \mathcal{N}(0, J^{-1})$, as required.

    \textbf{Showing point (ii) of the Theorem:} We apply Lemma 1 with: 
    $$
    h_{ij}(\theta) = \frac{\nabla_{\theta}\Phi_2^{\rho}\big(W_{ij}^{'}\beta, W_{ji}^{'}\beta\big)\left(\nabla_{\theta}\Phi_2^{\rho}\big(W_{ij}^{'}\beta, W_{ji}^{'}\beta\big)\right)'}{\Phi_2^{\rho}\big(W_{ij}^{'}\beta, W_{ji}^{'}\beta\big)(1-\Phi_2^{\rho}\big(W_{ij}^{'}\beta, W_{ji}^{'}\beta\big))}.
    $$
    Yet again, $h_{ij}(\theta)$ is clearly symmetric and continuous. Furthermore, from the (3.1.iii) step, $\|h_{ij}(\theta) \| \leq CR_{ij}^3$ for all $\theta \in \Theta$, so $\mathbb{E}\left[\sup_{\theta \in \Theta} \|h_{ij}(\theta) \|^{1+\eta} \right] < \infty$ under assumption \ref{ass:N1} provided $\eta \leq \frac{\delta}{3}$. So by Lemma 1, $\mathbb{E}[h_{ij}(\theta)]$ is continuous and $\frac{1}{N}\sum^{n-1}_{i=1} \sum^{n}_{j=i+1}h_{ij}(\theta)$ converges uniformly in probability to $\mathbb{E}[h_{ij}(\theta)]$. Decompose:
    $$
    \hat{J} - J = \underbrace{\frac{1}{N} \sum^{n-1}_{i=1} \sum^{n}_{j=i+1}\left(h_{ij}(\hat{\theta}) - \mathbb{E}[h_{ij}(\hat{\theta})]\right)}_{\text{(a)}} + \underbrace{\mathbb{E}[h_{ij}(\hat{\theta})] - \mathbb{E}[h_{ij}(\theta_0)]}_{\text{(b)}}.
    $$
    Term (a) is $o_P(1)$ because uniform convergence in probability implies convergence in probability at $\hat{\theta}$; term (b) is $o_P(1)$ by $\hat{\theta} \xrightarrow{p} \theta_0$ (Theorem 1) under assumptions \ref{ass:I1} and \ref{ass:C1}-\ref{ass:C3} (with \ref{ass:C3} implied by \ref{ass:N1}) in combination with continuity of $\mathbb{E}[h_{ij}(\theta)]$. Therefore, $\hat{J} \xrightarrow{p} J$, as required.

\end{proof}

\begin{proof}[Proof of Theorem 3]
 \vspace{1ex}  % forces a line break

 We proceed in 7 steps:
\begin{enumerate}
    \item Show that various objects are bounded under the theorem's assumptions, such as the derivatives of $g(c,d)$. These bounds are used repeatedly in the next steps.
    \item Conditional on the individual covariates $\mathcal{X} = \{X_1, ..., X_n\}$, apply the functional CLT for independent but not necessarily identically distributed random variables in \cite{van_der_vaart_weak_1996} (Theorem 2.11.1) to $T_N(c) \coloneq \frac{1}{\sqrt{N}} \sum^{n-1}_{i=1} \sum^{n}_{j=i+1} e_{ij} r_{ij}(c)$. This will imply that $\{T_N(c) : c \in \mathcal{C}\}$ converges in distribution to a mean-zero Gaussian process, conditional on $\mathcal{X}$. Then show that we can drop the conditioning on $\mathcal{X}$.
    \item Let $\hat{T}_N(c) \coloneq \frac{1}{\sqrt{N}} \sum^{n-1}_{i=1} \sum^{n}_{j=i+1} \hat{e}_{ij}\hat{r}_{ij}(c)$ and show that $\sup_{c \in \mathcal{C}} |\hat{T}_N(c) - T_N(c)| = o_P(1)$. This relies on the fact $\hat{\gamma} - \gamma_0 = O_P(N^{-\frac{1}{2}})$ under \ref{ass:T1}-\ref{ass:T3}. The implication is that the collection of $\hat{T}_N(c)$ - the numerator of $t_N(c)$ - converges in distribution to the same Gaussian process as the collection of $T_N(c)$.
    \item Let $\hat{\Sigma}(c) \coloneq \frac{1}{N} \sum^{n-1}_{i=1} \sum^{n}_{j=i+1} \hat{\omega}_{ij} \hat{r}_{ij}(c)^2$ be the object inside the square root in the denominator of $t_N(c)$. And let $\Sigma(c) \coloneq \mathbb{E}[\omega_{ij} r_{ij}(c)^2]$ be the population counterpart. Then show $\sup_{c \in \mathcal{C}} |\hat{\Sigma}(c) - \Sigma(c)| = o_P(1)$ so that, by the continuous mapping theorem given that $\inf_{c \in \mathcal{C}} \Sigma(c) > 0$ by \ref{ass:T4}, the collection of $t_N(c)$ converges in distribution to the mean-zero Gaussian process $\mathcal{G}$ in part (i) of the theorem.
    \item Take the supremum over $c$ of part (i) to obtain the convergence result in part (ii). \cite{gaenssler_continuity_2007} guarantee that the distribution of $\sup_{c \in \mathcal{C}}\mathcal{G}(c)$ is continuous and strictly increasing.
    \item Repeat steps 2-5, conditional on the data $\mathcal{D} = (\mathcal{X}, \{Y_{ij}\}_{i < j})$, with the bootstrap statistic $t_N^{*\text{sup}}$ in place of $t_N^{\text{sup}}$. This is easier because, conditional on the data, the numerator of $t_N^{*\text{sup}}$ is exactly Gaussian. This allows us to use the simpler functional CLT of \cite{billingsley_convergence_1999} (Theorem 7.1). Doing so, we show that $t_N^{*\text{sup}}$ converges in distribution to $\sup_{c \in \mathcal{C}}\mathcal{G}(c)$ too (conditional on the data).
    \item Finally, show that $\hat{F}_N(x) \coloneq \mathbb{P}[t_N^{*\text{sup}} \leq x \mid \mathcal{D}]$ converges in probability to the distribution of $\sup_{c \in \mathcal{C}}\mathcal{G}(c)$, $F(x)$, for each $x$. From here, strict monotonicity of $F(x)$ ensures $c^*_{1-\alpha}$ converges in probability to $c_{1-\alpha}$, and a sandwich argument combined with part (ii) gives $\mathbb{P}[t_N^{\text{sup}} > c_{1-\alpha}^*] \to \alpha$, as required for part (iii).
\end{enumerate}
 
\textbf{Step 1:} We prove the following bounds for all $c, c_1, c_2 \in [\ubar{C}, \bar{C}]$ and $d, d_1, d_2 \in \mathbb{R}$, suppressing $ij$ subscripts and writing $e(y, z) = \phi(z) \frac{y - \Phi(z)}{\Phi(z)(1 - \Phi(z))}$:
\begin{align}
    &|g(c_1, d) - g(c_2, d)| \leq \frac{1}{\sqrt{\pi}}|c_1 - c_2|\label{g-c-lipschitz} \\
    &|g(c, d_1) - g(c, d_2)| \leq |d_1 - d_2|\label{g-d-lipschitz} \\
    &|d| \leq g(c,d) \leq |d| + \frac{c}{\sqrt{\pi}}\label{g-bound} \\
    &\| \Pi(c_1) - \Pi(c_2) \| \leq C_\Pi |c_1 - c_2|\label{Pi-lipschitz} \\
    &\sup_{c\in \mathcal{C}}\|\Pi(c) \| < \infty\label{Pi-bound}\\
    &|r(c_1) - r(c_2)| \leq L(V) |c_1 - c_2|\label{r-lipschitz}\\
    &\left| e\right| \leq 1 + |z| \label{e-bound}\\
    &\left|\frac{\partial e}{\partial z}\right| \leq C_1 < \infty\label{e-first-derivative}\\
    &\left|\frac{\partial^2 e}{\partial z^2}\right| \leq C_2 < \infty\label{e-second-derivative}\\
    &\mathbb{E}[e^4 \mid \mathcal{X}] \leq C_3 < \infty\label{e-fourth-moment}\\
    &| \hat{\omega} - \omega| \leq C_4 \|V\| \|\hat{\gamma} - \gamma_0\|\label{omega-lipschitz},
\end{align}
with $C_\Pi \coloneq \frac{1}{\sqrt{\pi}} \| \mathbb{E}[\omega VV']^{-1}\|\mathbb{E}[\omega \|V\|]$ and $L(V) \coloneq \frac{1}{\sqrt{\pi}} + C_\Pi \|V\|$.

For \eqref{g-c-lipschitz}, take $c_2 > c_1$ without loss of generality and note that $g(c,d) = c\psi(\frac{d}{c})= \frac{ce^{-\frac{d^2}{c^2}}}{\sqrt{\pi}} + d \ \text{erf}(\frac{d}{c})$ is continuous in $c$ on $[\ubar{C}, \bar{C}]$ and differentiable in $c$ on $(c_1, c_2)$. So, by the mean value theorem, there exists $c^* \in (c_1, c_2)$ such that $g(c_1, d) - g(c_2,d) = \partial_c g(c^*, d) (c_1 - c_2)$. Then take absolute values and notice that $\partial_c g(c, d) = \psi(\frac{d}{c}) - \frac{d}{c}\ \text{erf}(\frac{d}{c}) = \frac{e^{-\frac{d^2}{c^2}}}{\sqrt{\pi}} \in (0, \frac{1}{\sqrt{\pi}}]$ for all $c$.

For \eqref{g-d-lipschitz}, do the same as \eqref{g-c-lipschitz} but with $\partial_d$ instead of $\partial_c$. Namely, $\partial_d g(c, d) = \psi'(\frac{d}{c}) = \text{erf}(\frac{d}{c}) \in [-1, 1]$ for all $d$, so the result follows from the mean value theorem again.

For \eqref{g-bound}, use that $g(c,d) = \mathbb{E}_W|d - W|$ with $W \sim \mathcal{N}(0, \frac{c^2}{2})$. So, by Jensen's inequality, $g(c,d) \geq |\mathbb{E}[d - W]| = |d|$. And, by the triangle inequality plus the formula for the mean of the folded normal distribution, $g(c,d) \leq \mathbb{E}_W|d| + \mathbb{E}_W|W| = |d| + \frac{c}{\sqrt{\pi}}$.

For \eqref{Pi-lipschitz}, simply compute:
$$
\| \Pi(c_1) - \Pi(c_2) \| \leq \| \mathbb{E}[\omega VV']^{-1}\| \mathbb{E}[\omega \|V\| |g(c_1, d) - g(c_2, d)|],
$$
with $\| \mathbb{E}[\omega VV']^{-1}\| < \infty$ by the nonsingularity assumption \ref{ass:T2}. Then plug in \eqref{g-c-lipschitz} to obtain the desired result.

For \eqref{Pi-bound}, similarly to \eqref{Pi-lipschitz}, compute:
\begin{align*}
    \| \Pi(c) \| &\leq \| \mathbb{E}[\omega VV']^{-1}\| \mathbb{E}[\omega \|V\| |g(c, d)|] \\
    &\leq \| \mathbb{E}[\omega VV']^{-1}\|\mathbb{E}\left[\omega \|V\||d| + \omega \|V\| \frac{c}{\sqrt{\pi}}\right].
\end{align*}
Then note that the result follows because $\| \mathbb{E}[\omega VV']^{-1}\| < \infty$ by \ref{ass:T2} again, $\omega$ is bounded by $\frac{2}{\pi}$, $\sup_c c = \bar{C} < \infty$, and both $\mathbb{E}\|V\|$ and $\mathbb{E}[\|V\| |d|]$ are finite given finite second moments of $W$ (implied by \ref{ass:T3}).

For \eqref{r-lipschitz}, simply use the triangle inequality:
$$
|r(c_1) - r(c_2)| \leq |g(c_1, d) - g(c_2, d)| + \| \Pi(c_1) - \Pi(c_2) \| \|V\|,
$$
and then plug in \eqref{g-c-lipschitz} and \eqref{Pi-lipschitz}.

For \eqref{e-bound}, consider the cases of $y=1$ and $y=0$ separately: $e(1, z) = \frac{\phi(z)}{\Phi(z)} = \lambda(z)$ and $e(0, z) = -\frac{\phi(z)}{1 - \Phi(z)} = -\lambda(-z)$. Then use that $\lambda(z) \leq 1 + |z|$ for all $z$ from Step 1 of the Proof of Proposition 1 to obtain the desired result.

For \eqref{e-first-derivative}, consider the cases of $y=1$ and $y=0$ separately again. Then $\frac{\partial e}{\partial z}(1, z)= \lambda'(z)$ and $\frac{\partial e}{\partial z}(0, z)= \lambda'(-z)$. From the Proof of Proposition 1, $\lambda'(z) \in (-1, 0)$ for all $z$, so $\left|\frac{\partial e}{\partial z}\right| < 1$, as required.

For \eqref{e-second-derivative}, differentiate again to obtain $\frac{\partial^2 e}{\partial z^2}(1, z)= \lambda''(z)$ and $\frac{\partial^2 e}{\partial z^2}(0, z)= -\lambda''(-z)$, so $\left|\frac{\partial^2 e}{\partial z^2}\right| \leq \sup_{z \in \mathbb{R}} \lambda''(z)$. We remain to show that this supremum is finite. From the Proof of Proposition 1, $\lambda''(z) > 0$, so we remain to establish an upper bound. Since $\lambda'(z) = -\lambda(z)(\lambda(z) + z) \in (-1, 0)$, $\lambda''(z) = \lambda(z)[(\lambda(z) + z)^2 - \lambda'(z) - 1] < \lambda(z)(\lambda(z) + z)^2$. Case 1: $z \leq 0$. Then $\lambda''(z) < \lambda(z) + z \leq 1 + |z| + z = 1 - z + z = 1$. Case 2: $z \geq 0$. Then $\lambda(z) \leq 2 \phi(z)$. And since $\lambda(z)$ is decreasing, $\lambda(z) + z \leq \lambda(0) + z < 1 + z$. So $\lambda''(z) = \lambda(z)(\lambda(z) + z)^2 < 2\phi(z)(1+z)^2 \leq 8 \phi(1)$, where the final inequality follows from the fact $2\phi(z)(1+z)^2$ is uniquely maximized when $z=1$. So $\left|\frac{\partial^2 e}{\partial z^2}\right| \leq 8\phi(1) \approx 1.94$, as required.

For \eqref{e-fourth-moment}, note that since $y$ is a Bernoulli random variable with parameter $\Phi(V^{'}\gamma)$, $\mathbb{E}[(y - \Phi(V^{'}\gamma))^4] = \Phi(V^{'}\gamma)(1-\Phi(V^{'}\gamma))\left[\Phi(V^{'}\gamma)^3 + (1-\Phi(V^{'}\gamma))^3 \right]$. So:
\begin{align*}
    \mathbb{E}[e^4 \mid \mathcal{X}] &= \left(\frac{\phi(V^{'}\gamma)}{\Phi(V^{'}\gamma)(1-\Phi(V^{'}\gamma))}\right)^4\mathbb{E}[(y - \Phi(V^{'}\gamma))^4] \\
    &=\frac{\phi(V^{'}\gamma)^4\left[\Phi(V^{'}\gamma)^3 + (1-\Phi(V^{'}\gamma))^3 \right]}{\Phi(V^{'}\gamma)^3(1-\Phi(V^{'}\gamma))^3}.
\end{align*}
Because $\phi(-z)=\phi(z)$ and $\Phi(-z) = 1- \Phi(z)$, $\mathbb{E}[e^4 \mid \mathcal{X}]$ is an even function. So it suffices to bound $\mathbb{E}[e^4 \mid \mathcal{X}]$ on $[0, \infty)$. Specifically, we show that there exists some $z_0 \in (1,\infty)$ such that $\mathbb{E}[e^4 \mid z] \leq 1$ for all $|z| \geq z_0$. Then, combined with the fact $\mathbb{E}[e^4 \mid z]$ is continuous so attains a maximum on the compact set $[-z_0, z_0]$ by the extreme value theorem, $\mathbb{E}[e^4 \mid z]$ is bounded by either $1$ or its maximum on $[-z_0, z_0]$. To show $\mathbb{E}[e^4 \mid z] \leq 1$ for all $|z| \geq z_0$, first note that $\Phi(z)^3 + (1-\Phi(z))^3 < (\Phi(z) + 1-\Phi(z))^3 = 1$, $\frac{1}{\Phi(z)^3} \leq 8$ since $z \geq 0$, and $\frac{1}{(1-\Phi(z))^3} \leq \frac{(z^2 + 1)^3}{z^3\phi(z)^3}$ using the lower bound on $\frac{1-\Phi(z)}{\phi(z)}$ in \cite{gordon_values_1941}. So:
$$
\mathbb{E}[e^4 \mid z] \leq \frac{8\phi(z)(z^2 + 1)^3}{z^3}.
$$
Then use that $(z^2 + 1)^3 \leq 8z^6$ for $z \geq 1$ and $e^t \geq \frac{t^2}{2}$ with $t = \frac{z^2}{2}$ inside the definition of $\phi(z)$ to obtain:
$$
\mathbb{E}[e^4 \mid z] \leq 64\phi(z)z^3 \leq \frac{512}{\sqrt{2\pi}}\frac{1}{z}.
$$
Finally, as $z \to \infty$ this upper bound on $\mathbb{E}[e^4 \mid z]$ goes to 0, so there must exist some $z_0 \in (1, \infty)$ such that $\mathbb{E}[e^4 \mid z] \leq 1$ for all $|z| \geq z_0$, as required.

For \eqref{omega-lipschitz}, note that $\omega$ is continuous and differentiable, so by the mean value theorem and Cauchy-Schwarz inequality:
$$
| \hat{\omega} - \omega| \leq \sup_z |\omega'(z)| |V'(\hat{\gamma} - \gamma_0)| \leq \sup_z |\omega'(z)| \|V \| \|\hat{\gamma} - \gamma_0 \|.
$$
Since $\frac{\partial e}{\partial z}(y, z)= y\lambda'(z)+ (1-y)\lambda'(-z)$ from the proof of \eqref{e-first-derivative}, we have that $\omega(z) = - \mathbb{E}[\frac{\partial e}{\partial z}(y, z) \mid z] = -[\Phi(z)\lambda'(z)+ (1-\Phi(z))\lambda'(-z)]$. So:
$$
\omega'(z) = -\phi(z)[\lambda'(z) - \lambda'(-z)] - [\Phi(z)\lambda''(z) - (1 - \Phi(z))\lambda''(-z)].
$$
Since $\lambda'(z) \in (-1, 0)$ from the Proof of Proposition 1, $\lambda'(z) - \lambda'(-z) \in (-1, 1)$, so the first term is bounded in absolute value by $\phi(z) \leq \phi(0)$. Moreover, since $\lambda''(z) > 0$ and $\lambda''(z) < 8\phi(1)$ from the proof of \eqref{e-second-derivative}, the second term is bounded in absolute value by $8\phi(1)$. So, overall, $\sup_z |\omega'(z)| \leq \phi(0) + 8 \phi(1) \approx 2.33 < \infty$, as required.

\textbf{Step 2:} Conditional on $\mathcal{X}$, $Z_{N,ij}(c) \coloneq \frac{1}{\sqrt{N}} e_{ij} r_{ij}(c)$ are independent across dyads with $\mathbb{E}[Z_{N,ij}(c) \mid \mathcal{X}] = \frac{1}{\sqrt{N}}r_{ij}(c) \mathbb{E}[e_{ij} \mid \mathcal{X}]=0$ via assumption \ref{ass:T1}. So by the functional CLT for independent but not necessarily identically distributed stochastic processes (Theorem 2.11.1) in \cite{van_der_vaart_weak_1996} applied to a fixed realization of $\mathcal{X}$, and indexed by the totally bounded semi-metric space $(\mathcal{C}, |\cdot|)$, we have that:\footnote{Note that we discard their measurability statement because $Z_{N,ij}(c)$ are finite sums of continuous functions of $c$ on a compact interval.}
$$
\sum^{n-1}_{i=1} \sum^{n}_{j=i+1}\left(Z_{N,ij}(c) - \mathbb{E}[Z_{N,ij}(c) \mid \mathcal{X}] \right) = T_N(c),
$$
is such that:
$$
\{T_N(c) : c \in \mathcal{C}\} \xrightarrow{d} \mathcal{T},
$$
in $\ell^\infty(\mathcal{C})$, conditional on $\mathcal{X}$, where $\mathcal{T}$ is the mean-zero Gaussian process with covariances $\mathbb{E}[\omega_{ij}r_{ij}(c_1) r_{ij}(c_2)]$, provided:
    \begin{enumerate}
        \item[(i)] $\sum^{n-1}_{i=1} \sum^{n}_{j=i+1} \mathbb{E} \left[\sup_{c \in \mathcal{C}}|Z_{N,ij}(c) |^2 \mathds{1}[\sup_{c \in \mathcal{C}}|Z_{N,ij}(c) | > \eta] \mid \mathcal{X} \right] \xrightarrow{a.s.} 0$ for every $\eta > 0$,
        \item[(ii)] $\sup_{|c_1 - c_2| < \delta_n}\sum^{n-1}_{i=1} \sum^{n}_{j=i+1} \mathbb{E} \left[(Z_{N,ij}(c_1) - Z_{N,ij}(c_2))^2 \mid \mathcal{X} \right] \xrightarrow{a.s.} 0$ for every $\delta_n \downarrow 0$,
        \item[(iii)] $\int_0^{\delta_n}\sqrt{\log N(s, \mathcal{C}, d_n)} \ ds \xrightarrow{a.s.} 0$ for every $\delta_n \downarrow 0$, where $N(s, \mathcal{C}, d_n)$ is the covering number and $d_n(c_1, c_2) = \sqrt{\sum^{n-1}_{i=1} \sum^{n}_{j=i+1}(Z_{N,ij}(c_1) - Z_{N,ij}(c_2))^2}$,
        \item[(iv)] $\sum^{n-1}_{i=1} \sum^{n}_{j=i+1} \mathbb{E} \left[Z_{N,ij}(c_1) Z_{N,ij}(c_2) \mid \mathcal{X} \right] \xrightarrow{a.s.} \mathbb{E}[\omega_{ij}r_{ij}(c_1) r_{ij}(c_2)]$ for each fixed pair $(c_1, c_2) \in \mathcal{C} \times \mathcal{C}$.
    \end{enumerate}

    To show (i), first use \eqref{g-bound} to bound:
    \begin{align*}
        |Z_{N,ij}(c)| &= \frac{1}{\sqrt{N}}\left|\phi(V_{ij}^{'}\gamma) \frac{Y_{ij} - \Phi(V_{ij}^{'}\gamma)}{\Phi(V_{ij}^{'}\gamma)(1 - \Phi(V_{ij}^{'}\gamma))} \right| \cdot |g(c, d_{ij}) - \Pi(c)'V_{ij} | \\
        &\leq \frac{1}{\sqrt{N}}\left| \frac{\phi(V_{ij}^{'}\gamma)}{\Phi(V_{ij}^{'}\gamma)(1 - \Phi(V_{ij}^{'}\gamma))} \right| \left( |g(c, d_{ij})| + | \Pi(c)'V_{ij} |\right) \\
        &\leq \frac{1}{\sqrt{N}} \frac{\phi(V_{ij}^{'}\gamma)}{\Phi(V_{ij}^{'}\gamma)(1 - \Phi(V_{ij}^{'}\gamma))} \left( |d_{ij}| + \frac{c}{\sqrt{\pi}} + \| \Pi(c) \| \|V_{ij}\|\right),
    \end{align*}
    and thus \eqref{Pi-bound} to bound:
    $$
    \sup_{c \in \mathcal{C}}|Z_{N,ij}(c)| \leq \frac{1}{\sqrt{N}} \frac{\phi(V_{ij}^{'}\gamma)}{\Phi(V_{ij}^{'}\gamma)(1 - \Phi(V_{ij}^{'}\gamma))}\bar{R}_{ij},
    $$
    with $\bar{R}_{ij} \coloneq |d_{ij}| + \frac{\bar{C}}{\sqrt{\pi}} + \sup_{c \in \mathcal{C}}\| \Pi(c) \| \|V_{ij}\|$. Similarly:
    $$
    \sup_{c \in \mathcal{C}}|Z_{N,ij}(c)|^2 \leq \frac{1}{N} \left(\frac{\phi(V_{ij}^{'}\gamma)}{\Phi(V_{ij}^{'}\gamma)(1 - \Phi(V_{ij}^{'}\gamma))}\right)^2\bar{R}_{ij}^2.
    $$
    Plugging in these inequalities and using that $\mathbb{E}[e_{ij}^2 \mid \mathcal{X}] = \omega_{ij}$, we have that the left hand side of the Lindeberg condition (i) is bounded by:
    \begin{align*}
        &\sum^{n-1}_{i=1} \sum^{n}_{j=i+1}\frac{1}{N} \omega_{ij}\bar{R}_{ij}^2 \mathds{1}\left[\frac{1}{\sqrt{N}} \frac{\phi(V_{ij}^{'}\gamma)}{\Phi(V_{ij}^{'}\gamma)(1 - \Phi(V_{ij}^{'}\gamma))}\bar{R}_{ij} > \eta\right] \\
        &\leq \frac{1}{N}\sum^{n-1}_{i=1} \sum^{n}_{j=i+1} \omega_{ij}\bar{R}_{ij}^2 \mathds{1}\left[\frac{\phi(V_{ij}^{'}\gamma)}{\Phi(V_{ij}^{'}\gamma)(1 - \Phi(V_{ij}^{'}\gamma))}\bar{R}_{ij} > M\right],
    \end{align*}
    for every fixed $M$ satisfying $\eta \sqrt{N} \geq M$. For each $M$, we can apply the strong law of large numbers for U-statistics (Theorem 5.4.A) in \cite{serfling_approximation_1980} to the sample average on the right hand side. Since (a) $\omega_{ij}$ and $\bar{R}_{ij}$ do not depend on $Y_{ij}$ or $n$, (b) $X_i$ are i.i.d., (c) $V_{ij}$ and $|d_{ij}|$ are symmetric, and (d) $\mathbb{E} \left|\omega_{ij}\bar{R}_{ij}^2 \mathds{1}\left[\frac{\phi(V_{ij}^{'}\gamma)}{\Phi(V_{ij}^{'}\gamma)(1 - \Phi(V_{ij}^{'}\gamma))}\bar{R}_{ij} > M\right]\right| < \infty$ because $\omega_{ij}$ is bounded and $\mathbb{E}[\bar{R}_{ij}^2] < \infty$ under assumption \ref{ass:T3}, the sample average converges almost surely to:
    $$
    \mathbb{E} \left[\omega_{ij}\bar{R}_{ij}^2 \mathds{1}\left[\frac{\phi(V_{ij}^{'}\gamma)}{\Phi(V_{ij}^{'}\gamma)(1 - \Phi(V_{ij}^{'}\gamma))}\bar{R}_{ij} > M\right]\right].
    $$
    Finally, let $M \to \infty$ and use the dominated convergence theorem to establish that the $\lim \sup$ as $n \to \infty$ of the left hand side of the Lindeberg condition equals:
    $$
    \lim_{M \to \infty}\mathbb{E} \left[\omega_{ij}\bar{R}_{ij}^2 \mathds{1}\left[\frac{\phi(V_{ij}^{'}\gamma)}{\Phi(V_{ij}^{'}\gamma)(1 - \Phi(V_{ij}^{'}\gamma))}\bar{R}_{ij} > M\right]\right] = 0,
    $$
    thus proving (i), as required.

    To show (ii), first plug \eqref{r-lipschitz} into the left hand side and use $\mathbb{E}[e_{ij}^2 \mid \mathcal{X}] = \omega_{ij}$ to obtain the upper bound:
    \begin{align*}
        \text{LHS}&\leq\sup_{|c_1 - c_2| < \delta_n} \frac{1}{N} \sum^{n-1}_{i=1} \sum^{n}_{j=i+1} \omega_{ij} L(V_{ij})^2 (c_1 - c_2)^2 \\
        &\leq \delta_n^2 \frac{1}{N} \sum^{n-1}_{i=1} \sum^{n}_{j=i+1} \omega_{ij} L(V_{ij})^2.
    \end{align*}
    Then, analogous to (i), apply the SLLN in \cite{serfling_approximation_1980} to the sample average on the right hand side given (a) $\omega_{ij}$ and $L(V_{ij})$ do not depend on $Y_{ij}$ or $n$, (b) $X_i$ are i.i.d., (c) $V_{ij}$ is symmetric, and (d) $\mathbb{E} \left|\omega_{ij}L(V_{ij})^2\right| < \infty$ because $\omega_{ij}$ is bounded and $\mathbb{E}[L(V_{ij})^2] < \infty$ under \ref{ass:T3}. So the left hand side of (ii) is $O(\delta_n^2)$, so converges almost surely to 0 as $\delta_n \downarrow 0$, as required.

    To show (iii), note that by \eqref{r-lipschitz}:
    \begin{align*}
        d_n(c_1, c_2)^2 &=\sum^{n-1}_{i=1} \sum^{n}_{j=i+1}(Z_{N,ij}(c_1) - Z_{N,ij}(c_2))^2 \\
        &\leq |c_1 - c_2|^2 \hat{K}_n^2,
    \end{align*}
    where $\hat{K}_n^2 \coloneq \frac{1}{N} \sum^{n-1}_{i=1} \sum^{n}_{j=i+1} e_{ij}^2 L(V_{ij})^2$. So $d_n(c_1, c_2) \leq |c_1 - c_2| \hat{K}_n$, and therefore $|c_1 - c_2| \leq \frac{s}{\hat{K}_n} \implies d_n(c_1, c_2) \leq s$. That is, if $\mathcal{C}$ can be covered by $M$ balls of $|\cdot|$-radius $\frac{s}{\hat{K}_n}$, then $\mathcal{C}$ can be covered by $M$ balls of $d_n$-radius $s$. And $\mathcal{C}$ has length $\bar{C} - \ubar{C}$, so:
    $$
    N(s, \mathcal{C}, d_n) \leq N\left(\frac{s}{\hat{K}_n}, \mathcal{C}, |\cdot|\right) \leq \frac{(\bar{C} - \ubar{C})\hat{K}_n}{2s} \leq 1 + \frac{(\bar{C} - \ubar{C})\hat{K}_n}{s}.
    $$
    We want to show that $\hat{K}_n^2 \xrightarrow{a.s.} \mathbb{E}[\omega_{ij}L(V_{ij})^2]$. To this end, apply Lemma 1 with $h_{ij} = e_{ij}^2 L(V_{ij})^2$ and $\eta=1$ such that we obtain almost sure convergence. $h_{ij}$ is clearly symmetric because $V_{ij}$ is symmetric. And $h_{ij}$ is continuous. Finally, $\mathbb{E}[h_{ij}^2] = \mathbb{E}[\mathbb{E}[e_{ij}^4\mid \mathcal{X}] L(V_{ij})^4] \leq C_3 \mathbb{E}[L(V_{ij})^4]< \infty$ using \eqref{e-fourth-moment} and \ref{ass:T3}.\footnote{Note the absence of the supremum because we take $\Theta =\{\gamma_0\}$ here.} To conclude (iii), fix any $M$ such that $M^2 > \mathbb{E}[\omega_{ij}L(V_{ij})^2]$, so that on the probability one event (by the almost sure convergence result above) where $M > \hat{K}_n$: $N(s, \mathcal{C}, d_n) \leq 1 + \frac{(\bar{C} - \ubar{C})M}{s}$. So:
    $$
    \int_0^{\delta_n}\sqrt{\log N(s, \mathcal{C}, d_n)} \ ds \leq \int_0^{\delta_n}\sqrt{\log \left(1 + \frac{(\bar{C} - \ubar{C})M}{s}\right)} \ ds \xrightarrow{a.s.} 0,
    $$
    for every $\delta_n \downarrow 0$, as required.

    To show (iv), apply the SLLN in \cite{serfling_approximation_1980} to:
    $$
    \sum^{n-1}_{i=1} \sum^{n}_{j=i+1} \mathbb{E} \left[Z_{N,ij}(c_1) Z_{N,ij}(c_2) \mid \mathcal{X} \right] = \frac{1}{N} \sum^{n-1}_{i=1} \sum^{n}_{j=i+1} \omega_{ij} r_{ij}(c_1)r_{ij}(c_2).
    $$
    Since (a) $\omega_{ij}$ and $r_{ij}(c)$ do not depend on $Y_{ij}$ or $n$, (b) $X_i$ are i.i.d., (c) $V_{ij}$ and $g(c, d_{ij})$ are symmetric (because $\psi(\cdot)$ is even), and (d) $\mathbb{E} \left|\omega_{ij} r_{ij}(c_1)r_{ij}(c_2)\right| \leq \frac{2}{\pi}\mathbb{E}[\bar{R}_{ij}^2]< \infty$ under \ref{ass:T3}, the sample average converges almost surely to $\mathbb{E}[\omega_{ij}r_{ij}(c_1) r_{ij}(c_2)]$ for each fixed pair $(c_1, c_2) \in \mathcal{C} \times \mathcal{C}$, as required.

    To conclude Step 2, we show that we can drop the conditioning on $\mathcal{X}$, i.e. go from $\{T_N(c) : c \in \mathcal{C}\} \xrightarrow{d} \mathcal{T}$ in $\ell^\infty(\mathcal{C})$ conditional on $\mathcal{X}$ to $\{T_N(c) : c \in \mathcal{C}\} \xrightarrow{d} \mathcal{T}$ in $\ell^\infty(\mathcal{C})$ unconditionally. To this end, as in the Proof of Lemma 2, use the portmanteau theorem to express the conditional convergence in distribution as:\footnote{Write $T_N = (T_N(c))_{c\in\mathcal{C}}$.}
    $$
    \mathbb{E}\left[f(T_N) \mid \mathcal{X}\right] \to \mathbb{E} \left[f(\mathcal{T})\right],
    $$
    for every bounded and continuous function $f$. Thus, by the law of iterated expectation and dominated convergence theorem given $| \mathbb{E}[f(\cdot) \mid \mathcal{X}] | \leq \sup |f| < \infty$ since $f$ is bounded:
    \begin{align*}
        \mathbb{E}\left[f(T_N)\right] = \mathbb{E}\left[\mathbb{E}\left[f(T_N) \mid \mathcal{X}\right] \right] \to \mathbb{E}\left[\mathbb{E}\left[f\left(\mathcal{T}  \right)\right] \right] =\mathbb{E}\left[f\left(\mathcal{T}  \right)\right].
    \end{align*}
    Therefore, unconditional convergence in distribution is achieved via another application of the portmanteau theorem.

    \textbf{Step 3:} We begin this step by arguing that $\hat{\gamma} - \gamma_0 = O_P(N^{-\frac{1}{2}})$ under \ref{ass:T1}-\ref{ass:T3}. This is achieved by repeating the steps in the proofs of Theorem 1 and Theorem 2 with the TU log-likelihood $\ell_{ij}(\gamma) = Y_{ij} \log \Phi\big(V_{ij}^{'}\gamma\big) + (1 - Y_{ij})\log \big(1 - \Phi\big(V_{ij}^{'}\gamma\big)\big)$ in place of the NTU log-likelihood $\ell_{ij}(\beta, \rho) = Y_{ij} \log \Phi_2^{\rho}\big(W_{ij}^{'}\beta, W_{ji}^{'}\beta\big) + (1 - Y_{ij})\log \big(1 - \Phi_2^{\rho}\big(W_{ij}^{'}\beta, W_{ji}^{'}\beta\big)\big)$. This is much easier: we are simply reproducing the probit consistency and asymptotic normality proofs in \cite{newey_chapter_1994}, but accounting for dyadic data via our Lemma 1 and Lemma 2. For identification, the same Kullback-Leibler argument as in the Proof of Theorem 1 gives a unique maximizer of the population objective iff $\Phi(V_{ij}'\gamma) \neq \Phi(V_{ij}'\gamma_0)$ with positive probability, which is ensured by nonsingularity of $\mathbb{E}[\omega_{ij} V_{ij} V_{ij}']$ via \ref{ass:T2}. For the log-likelihood envelope, use \eqref{e-bound} to obtain:
    $$
    |\ell_{ij}(\gamma)| \leq C_0(1 + \|\gamma\|^2 \|V_{ij}\|^2),
    $$
    so that $\mathbb{E}|\ell_{ij}(\gamma)|^2 < \infty$ under \ref{ass:T3}. For the score CLT, apply Lemma 2 given:
    $$
    \mathbb{E}\|e_{ij}V_{ij}\|^{2 + \eta} \leq C_3^{\frac{2+\eta}{4}} \mathbb{E}\|V_{ij}\|^{2+\eta} < \infty
    $$
    for $\eta \leq 2$ using \eqref{e-fourth-moment} and \ref{ass:T3}. The limit covariance is $\mathbb{E}[\omega_{ij} V_{ij} V_{ij}']$, again nonsingular by \ref{ass:T2}. Finally, the Hessian $H_{ij}(\gamma) = \frac{\partial e_{ij}}{\partial z_{ij}} V_{ij} V_{ij}'$ is bounded as such using \eqref{e-first-derivative}:
    $$
    \sup_\gamma \|H_{ij}(\gamma) \| \leq C_1 \|V_{ij}\|^2.
    $$
    So Lemma 1 with $\eta =1$ can be applied to the Hessian to establish uniform convergence given $\mathbb{E}\|V_{ij}\|^4$ under \ref{ass:T3}, as required.
    
    Next, we decompose:
    \begin{align*}
         \hat{T}_N(c) - T_N(c)
         &= \underbrace{\frac{1}{\sqrt{N}} \sum^{n-1}_{i=1} \sum^{n}_{j=i+1}[\hat{e}_{ij} - e_{ij}]g(c, \hat{d}_{ij})}_{\text{(I)}} \\
         &+ \underbrace{\frac{1}{\sqrt{N}} \sum^{n-1}_{i=1} \sum^{n}_{j=i+1}e_{ij}[g(c, \hat{d}_{ij}) - g(c, d_{ij})]}_{\text{(II)}}
         \\
         &+ \underbrace{\frac{1}{\sqrt{N}} \sum^{n-1}_{i=1} \sum^{n}_{j=i+1}e_{ij}\Pi(c)'V_{ij}}_{\text{(III)}}.
    \end{align*}
    The goal is to show that $\text{(I)} + \text{(III)} =o_P(1)$ and $\text{(II)} = o_P(1)$, uniformly in $c$. 
    
    First, study (I). Plugging a Taylor expansion of $e(y_{ij},z_{ij})$ in $z_{ij}$ about $\hat{z}_{ij}$ into (I) yields:
    \begin{align*}
        \text{(I)} &= \underbrace{\frac{1}{\sqrt{N}} \sum^{n-1}_{i=1} \sum^{n}_{j=i+1}\frac{\partial e_{ij}}{\partial z_{ij}} V_{ij}'(\hat{\gamma} - \gamma_0)g(c, \hat{d}_{ij})}_{\text{(Ia)}} + \underbrace{\frac{1}{\sqrt{N}} \sum^{n-1}_{i=1} \sum^{n}_{j=i+1}\frac{1}{2}\frac{\partial^2 e_{ij}}{\partial z_{ij}^2}(\cdot) \left[V_{ij}'(\hat{\gamma} - \gamma_0)\right]^2g(c, \hat{d}_{ij})}_{\text{(Ib)}}.
    \end{align*}
    By the triangle inequality, \eqref{g-bound}, and \eqref{e-second-derivative}, Term (Ib) is bounded in absolute value by:
    $$
    \frac{C_2}{2} \sqrt{N} \|\hat{\gamma} - \gamma_0\|^2 \frac{1}{N} \sum^{n-1}_{i=1} \sum^{n}_{j=i+1} \|V_{ij}\|^2\left(|\hat{d}_{ij}| + \frac{\bar{C}}{\sqrt{\pi}}\right).
    $$
    The sample average is not a U-statistic because $\hat{d}_{ij}$ depends on $\hat{\gamma}$ which depends on $n$, so we cannot apply \cite{serfling_approximation_1980} to it. Instead, substitute in $|\hat{d}_{ij}| \leq |d_{ij}| + \|W_{ij} - W_{ji}\| \|\hat{\gamma} - \gamma_0\|$ to obtain the upper bound:
    \begin{align*}
        &\frac{C_2}{2} \sqrt{N} \|\hat{\gamma} - \gamma_0\|^2 \frac{1}{N} \sum^{n-1}_{i=1} \sum^{n}_{j=i+1} \|V_{ij}\|^2|d_{ij}| \\
        &+\frac{C_2}{2} \sqrt{N} \|\hat{\gamma} - \gamma_0\|^3 \frac{1}{N} \sum^{n-1}_{i=1} \sum^{n}_{j=i+1} \|V_{ij}\|^2\|W_{ij} - W_{ji}\| \\
        &+ \frac{C_2}{2}\frac{\bar{C}}{\sqrt{\pi}} \sqrt{N} \|\hat{\gamma} - \gamma_0\|^2 \frac{1}{N} \sum^{n-1}_{i=1} \sum^{n}_{j=i+1} \|V_{ij}\|^2.
    \end{align*}
    Each term now contains a U-statistic, so we can apply \cite{serfling_approximation_1980} three times given \ref{ass:T3}. This implies that the sample averages are all $O_P(1)$, which, combined with the fact $\hat{\gamma} - \gamma_0 = O_P(N^{-\frac{1}{2}})$ from above, gives that the upper bound is $\sqrt{N} \cdot O_P(N^{-1})\cdot O_P(1) + \sqrt{N} \cdot O_P(N^{-\frac{3}{2}}) \cdot O_P(1) + \sqrt{N} \cdot O_P(N^{-1}) \cdot O_P(1) = O_P(N^{-\frac{1}{2}})$. For term (Ia), we run into the same issue of $\hat{d}_{ij}$ depending on $\hat{\gamma}$. So we bound the following term instead, using the triangle inequality, \eqref{g-d-lipschitz}, and \eqref{e-first-derivative}:
    \begin{align*}
        &\left|\frac{1}{\sqrt{N}} \sum^{n-1}_{i=1} \sum^{n}_{j=i+1}\frac{\partial e_{ij}}{\partial z_{ij}} V_{ij}'(\hat{\gamma} - \gamma_0)\left[g(c, \hat{d}_{ij}) - g(c, d_{ij})\right]\right| \\
        &\leq C_1 \sqrt{N} \|\hat{\gamma} - \gamma_0\| \frac{1}{N} \sum^{n-1}_{i=1} \sum^{n}_{j=i+1} |\hat{d}_{ij} - d_{ij}|\|V_{ij}\| \\
        &\leq C_1 \sqrt{N} \|\hat{\gamma} - \gamma_0\|^2 \frac{1}{N} \sum^{n-1}_{i=1} \sum^{n}_{j=i+1} \|W_{ij} - W_{ji}\|\|V_{ij}\|.
    \end{align*}
    Yet again by \cite{serfling_approximation_1980} given \ref{ass:T3}, the sample average is $O_P(1)$, so the whole term is $\sqrt{N} \cdot O_P(N^{-1})\cdot O_P(1) = o_P(1)$, and thus we can replace $g(c, \hat{d}_{ij})$ with $g(c, d_{ij})$ in term (Ia). From here, we plug $\frac{\partial e_{ij}}{\partial z_{ij}} = -\omega_{ij} + \left(\frac{\partial e_{ij}}{\partial z_{ij}} + \omega_{ij} \right)$ into (Ia). Then apply Lemma 1 with $\eta =1$ to the average $-\frac{1}{N} \sum^{n-1}_{i=1} \sum^{n}_{j=i+1}\omega_{ij} g(c, d_{ij})V_{ij}'$ to show that it converges almost surely to $-\mathbb{E}[\omega_{ij} g(c, d_{ij})V_{ij}']$ uniformly in $c$. Continuity in $c$ comes from the Lipschitz bound \eqref{g-c-lipschitz}, compactness is \ref{ass:T5}, and the bound on the expectation of the supremum comes from $\omega_{ij} \leq \frac{2}{\pi}$, \eqref{g-bound}, and \ref{ass:T3}. An almost identical argument guarantees that the average $\frac{1}{N} \sum^{n-1}_{i=1} \sum^{n}_{j=i+1}\left(\frac{\partial e_{ij}}{\partial z_{ij}} + \omega_{ij} \right) g(c, d_{ij})V_{ij}'$ converges almost surely to $\mathbb{E}[\left(\frac{\partial e_{ij}}{\partial z_{ij}} + \omega_{ij} \right) g(c, d_{ij})V_{ij}'] = \mathbb{E}[\mathbb{E}[\left(\frac{\partial e_{ij}}{\partial z_{ij}} + \omega_{ij} \right) \mid \mathcal{X}]g(c, d_{ij})V_{ij}'] = \mathbb{E}[\left(-\omega_{ij} + \omega_{ij} \right)g(c, d_{ij})V_{ij}'] = 0$ uniformly in $c$. Finally, the standard probit linearization gives:
    $$
    \sqrt{N}(\hat{\gamma} - \gamma_0) = (\mathbb{E}[\omega_{ij} V_{ij} V_{ij}'])^{-1}\frac{1}{\sqrt{N}} \sum^{n-1}_{i=1} \sum^{n}_{j=i+1}e_{ij}V_{ij} + o_P(1).
    $$
    All together, uniformly in $c$:
    \begin{align*}
        \text{(I)} &= -\mathbb{E}[\omega_{ij} g(c, d_{ij})V_{ij}'](\mathbb{E}[\omega_{ij} V_{ij} V_{ij}'])^{-1}\frac{1}{\sqrt{N}} \sum^{n-1}_{i=1} \sum^{n}_{j=i+1}e_{ij}V_{ij} + o_P(1) \\
        &= -\Pi(c)'\frac{1}{\sqrt{N}} \sum^{n-1}_{i=1} \sum^{n}_{j=i+1}e_{ij}V_{ij} + o_P(1) \\
        &= -\text{(III)}  + o_P(1).
    \end{align*}
    That is, $\text{(I)} + \text{(III)} = o_P(1)$, as required.

    Second, study (II). Plugging in a Taylor expansion of $g_{ij} = g(c, d_{ij})$ in $d_{ij}$ around $\hat{d}_{ij}$ and using that $\frac{\partial g_{ij}}{\partial d_{ij}} = \text{erf}(\frac{d_{ij}}{c})$ and $\left|\frac{\partial^2 g_{ij}}{\partial d_{ij}^2} \right| = \left|\frac{2}{c \sqrt{\pi}}\text{exp}(-(\frac{d_{ij}}{c})^2) \right| \leq \frac{2}{\ubar{C} \sqrt{\pi}}$ yields:
    \begin{align*}
        \text{(II)} &= \underbrace{\left[\frac{1}{N} \sum^{n-1}_{i=1} \sum^{n}_{j=i+1}e_{ij} \ \text{erf} \left(\frac{d_{ij}}{c}\right) (W_{ij} - W_{ji})'\right]\sqrt{N}(\hat{\gamma} - \gamma_0)}_{\text{(IIa)}} + \underbrace{\frac{1}{\sqrt{N}} \sum^{n-1}_{i=1} \sum^{n}_{j=i+1}e_{ij} R_{ij}(c)}_{\text{(IIb)}},
    \end{align*}
    with $R_{ij}(c)$ such that $|R_{ij}(c)| \leq \frac{1}{\ubar{C} \sqrt{\pi}}(d_{ij} - \hat{d}_{ij})^2$. We apply Lemma 1 with $\eta=1$ to the term in square brackets: continuity in $c$ is clear because the error function is continuous in $c$; the bound on the expectation of the supremum comes from $|\text{erf}(\cdot)| \leq 1$, $\mathbb{E}[e_{ij}^2 \mid \mathcal{X}] = \omega_{ij} \leq \frac{2}{\pi}$, and \ref{ass:T3}. So $\text{(IIa)} = o_P(1) \cdot O_P(1) = o_P(1)$, uniformly in $c$. An almost identical argument applies to (IIb), so combined we have that (II) is $o_P(1)$ uniformly in $c$, as required.

        \textbf{Step 4:} We want to show that $\sup_{c \in \mathcal{C}} |\hat{\Sigma}(c) - \Sigma(c)| = o_P(1)$. The difficulty is that $\hat{\Sigma}(c)$ is not a U-statistic because it depends on $\hat{\gamma}$ and $\hat{\Pi}(c)$. Again, we decompose:
    \begin{align*}
         \hat{\Sigma}(c) - \Sigma(c)
         &= \underbrace{\frac{1}{N} \sum^{n-1}_{i=1} \sum^{n}_{j=i+1}[\hat{\omega}_{ij} - \omega_{ij}]\hat{r}_{ij}(c)^2}_{\text{(I)}} \\
         &+ \underbrace{\frac{1}{N} \sum^{n-1}_{i=1} \sum^{n}_{j=i+1}\omega_{ij}[\hat{r}_{ij}(c)^2 - r_{ij}(c)^2]}_{\text{(II)}}
         \\
         &+ \underbrace{\frac{1}{N} \sum^{n-1}_{i=1} \sum^{n}_{j=i+1}\omega_{ij}r_{ij}(c)^2 - \mathbb{E}[\omega_{ij} r_{ij}(c)^2]}_{\text{(III)}}.
    \end{align*}
    
    To bound (I), we first bound $\hat{r}_{ij}(c)^2 = (r_{ij}(c) + [g(c, \hat{d}_{ij}) - g(c, d_{ij})] + [\Pi(c)'V_{ij} -\hat{\Pi}(c)'V_{ij}])^2$ using $(x+y+z)^2 \leq 3(x^2 + y^2 + z^2)$:
    \begin{align*}
        \hat{r}_{ij}(c)^2 &\leq 3r_{ij}(c)^2 + 3[g(c, \hat{d}_{ij}) - g(c, d_{ij})]^2 + 3[\Pi(c)'V_{ij} -\hat{\Pi}(c)'V_{ij}]^2 \\
        &\leq 3\bar{R}_{ij}^2 + 3 \| W_{ij} - W_{ji}\|^2 \|\hat{\gamma} - \gamma_0\|^2 + 3\| \hat{\Pi}(c) - \Pi(c)\|^2\|V_{ij}\|^2.
    \end{align*}
    Combining this with \eqref{omega-lipschitz} gives:
    \begin{align*}
        \sup_{c \in \mathcal{C}} |\text{(I)}| \leq 3 C_4 \, \|\hat{\gamma} - \gamma_0\| \Bigg(&\frac{1}{N} \sum^{n-1}_{i=1} \sum^{n}_{j=i+1} \|V_{ij}\| \bar{R}_{ij}^2 + \|\hat{\gamma} - \gamma_0\|^2 \frac{1}{N} \sum^{n-1}_{i=1} \sum^{n}_{j=i+1} \|V_{ij}\| \|W_{ij} - W_{ji}\|^2 \\
        &+ \sup_{c \in \mathcal{C}} \| \hat{\Pi}(c) - \Pi(c)\|^2 \frac{1}{N} \sum^{n-1}_{i=1} \sum^{n}_{j=i+1} \|V_{ij}\|^3 \Bigg).
    \end{align*}
    Each sample average is $O_P(1)$ by \cite{serfling_approximation_1980} under \ref{ass:T3}. We remain to show that $\sup_{c \in \mathcal{C}} \| \hat{\Pi}(c) - \Pi(c)\| = o_P(1)$. Using \eqref{g-d-lipschitz} and \eqref{omega-lipschitz} to replace $\omega_{ij}$ by $\hat{\omega}_{ij}$ and $d_{ij}$ by $\hat{d}_{ij}$ at a cost of $O_P(N^{-1/2})$ uniformly in $c$, the uniform convergence result of $\frac{1}{N} \sum^{n-1}_{i=1} \sum^{n}_{j=i+1}\omega_{ij} g(c, d_{ij})V_{ij}'$ to $\mathbb{E}[\omega_{ij} g(c, d_{ij})V_{ij}']$ from Step 3 gives:
    $$
    \sup_{c \in \mathcal{C}} \left\| \frac{1}{N} \sum^{n-1}_{i=1} \sum^{n}_{j=i+1}\hat{\omega}_{ij} g(c, \hat{d}_{ij})V_{ij} - \mathbb{E}[\omega_{ij} g(c, d_{ij})V_{ij}]\right\| = o_P(1).
    $$
    By the same argument without $c$, and since $\mathbb{E}[\omega_{ij}V_{ij}V_{ij}']$ is nonsingular by \ref{ass:T2}, the continuous mapping theorem gives:
    $$
    \left(\frac{1}{N} \sum^{n-1}_{i=1} \sum^{n}_{j=i+1} \hat{\omega}_{ij}V_{ij}V_{ij}'\right)^{-1} - \mathbb{E}[\omega_{ij}V_{ij}V_{ij}']^{-1} = o_P(1).
    $$
    Pre-multiplying gives $\sup_{c \in \mathcal{C}} \| \hat{\Pi}(c) - \Pi(c)\| = o_P(1)$, and thus:
    $$
    \sup_{c \in \mathcal{C}} |\text{(I)}| = O_P(N^{-1/2})\left(O_P(1) + O_P(N^{-1}) + o_P(1)\right) = o_P(1),
    $$
    as required.

        To bound (II), use that:
    \begin{align*}
        |\hat{r}_{ij}(c)^2 - r_{ij}(c)^2| &= |r_{ij}(c) - \hat{r}_{ij}(c)||2r_{ij}(c) + \hat{r}_{ij}(c) - r_{ij}(c)| \\
        &\leq (|g(c, d_{ij}) - g(c, \hat{d}_{ij})| + \| \Pi(c) - \hat{\Pi}(c)\| \|V_{ij}\|) \\
        &\quad \ (2\bar{R}_{ij} + |g(c, d_{ij}) - g(c, \hat{d}_{ij})| + \| \Pi(c) - \hat{\Pi}(c)\| \|V_{ij}\|) \\
        &\leq (\| W_{ij} - W_{ji}\| \|\hat{\gamma} - \gamma_0\| + \| \Pi(c) - \hat{\Pi}(c)\| \|V_{ij}\|) \\
        &\quad \ (2\bar{R}_{ij} + \| W_{ij} - W_{ji}\| \|\hat{\gamma} - \gamma_0\| + \| \Pi(c) - \hat{\Pi}(c)\| \|V_{ij}\|).
    \end{align*}
    Expanding the product and using $(x + y)^2 \leq 2x^2 + 2y^2$ gives:
    \begin{align*}
        |\hat{r}_{ij}(c)^2 - r_{ij}(c)^2| &\leq 2\bar{R}_{ij}\| W_{ij} - W_{ji}\| \|\hat{\gamma} - \gamma_0\| + 2\bar{R}_{ij}\|V_{ij}\| \| \hat{\Pi}(c) - \Pi(c)\| \\
        &\quad + 2\| W_{ij} - W_{ji}\|^2 \|\hat{\gamma} - \gamma_0\|^2 + 2\|V_{ij}\|^2 \| \hat{\Pi}(c) - \Pi(c)\|^2.
    \end{align*}
    Plugging this into (II), using $\omega_{ij} \leq \frac{2}{\pi}$, and noting that only $\| \hat{\Pi}(c) - \Pi(c)\|$ depends on $c$:
    \begin{align*}
        \sup_{c \in \mathcal{C}} |\text{(II)}| \leq \frac{4}{\pi} \Bigg(&\|\hat{\gamma} - \gamma_0\| \frac{1}{N} \sum^{n-1}_{i=1} \sum^{n}_{j=i+1} \bar{R}_{ij}\| W_{ij} - W_{ji}\| + \sup_{c \in \mathcal{C}} \| \hat{\Pi}(c) - \Pi(c)\| \frac{1}{N} \sum^{n-1}_{i=1} \sum^{n}_{j=i+1} \bar{R}_{ij}\|V_{ij}\| \\
        &+ \|\hat{\gamma} - \gamma_0\|^2 \frac{1}{N} \sum^{n-1}_{i=1} \sum^{n}_{j=i+1} \| W_{ij} - W_{ji}\|^2 + \sup_{c \in \mathcal{C}} \| \hat{\Pi}(c) - \Pi(c)\|^2 \frac{1}{N} \sum^{n-1}_{i=1} \sum^{n}_{j=i+1} \|V_{ij}\|^2 \Bigg).
    \end{align*}
    Again, each sample average is $O_P(1)$ by \cite{serfling_approximation_1980} under \ref{ass:T3}. Using that $\hat{\gamma} - \gamma_0 = O_P(N^{-\frac{1}{2}})$ and $\sup_{c \in \mathcal{C}} \| \hat{\Pi}(c) - \Pi(c)\| = o_P(1)$ then gives:
    $$
    \sup_{c \in \mathcal{C}} |\text{(II)}| = O_P(N^{-\frac{1}{2}}) + o_P(1) + O_P(N^{-1}) + o_P(1) = o_P(1),
    $$
    as required.

    Finally, $\sup_{c \in \mathcal{C}} |\text{(III)}| = o_P(1)$ via Lemma 1 with $\eta = 1$. Continuity of $\omega_{ij} r_{ij}(c)^2$ in $c$ follows from \eqref{g-c-lipschitz} and \eqref{r-lipschitz}. And its envelope satisfies $\sup_{c \in \mathcal{C}} \omega_{ij} r_{ij}(c)^2 \leq \frac{2}{\pi}\bar{R}_{ij}^2
    $, so the bound on the expectation of the supremum is guaranteed under \ref{ass:T3}.

    Combining the bounds on (I), (II) and (III),
    $$
    \sup_{c \in \mathcal{C}} |\hat{\Sigma}(c) - \Sigma(c)| = o_P(1),
    $$
    as required.

    \textbf{Step 5:} From Step 4, $\{t_N(c) : c \in \mathcal{C}\} \overset{d}{\to} \mathcal{G}$ in $C(\mathcal{C})$, where $\mathcal{G}$ is the mean-zero Gaussian process with unit variance, covariances $\frac{\mathbb{E}[\omega_{ij} r_{ij}(c_1) r_{ij}(c_2)]}{\sqrt{\mathbb{E}[\omega_{ij} r_{ij}(c_1)^2]\mathbb{E}[\omega_{ij} r_{ij}(c_2)^2]}}$, and continuous sample paths. Since we have continuous sample paths on the compact set $\mathcal{C}$ (\ref{ass:T5}), by the extreme value theorem both $\sup_{c \in \mathcal{C}} t_N(c)$ and $\sup_{c \in \mathcal{C}} \mathcal{G}(c)$ exist. Then, since the supremum is a continuous function under the supremum norm, we can apply the continuous mapping theorem to $\{t_N(c) : c \in \mathcal{C}\} \overset{d}{\to} \mathcal{G}$ to obtain:
    $$
    \sup_{c \in \mathcal{C}} t_N(c) \overset{d}{\to} \sup_{c \in \mathcal{C}}\mathcal{G}(c),
    $$
    as required.

    We also require that the distribution function of $\sup_{c \in \mathcal{C}}\mathcal{G}(c)$ is continuous and strictly increasing on $\mathbb{R}$. This follows from \cite{gaenssler_continuity_2007} who prove that the distribution function of the supremum of a mean-zero Gaussian process is continuous and strictly increasing as soon as it is not degenerate at zero. We have that $\mathbb{V}ar \ \mathcal{G}(c) = 1 > 0$ from above, i.e. non-degeneracy for all $c$, as required.

    \textbf{Step 6:} Fix the data $\mathcal{D} = (\mathcal{X}, \{Y_{ij}\}_{i < j})$. Then $\hat{T}_N^*(c) \coloneq \frac{1}{\sqrt{N}} \sum^{n-1}_{i=1} \sum^{n}_{j=i+1} \xi_{ij} \hat{e}_{ij} \hat{r}_{ij}(c)$, the numerator of $t_N^{*\text{sup}}$, is random only through $\xi_{ij} \sim \mathcal{N}(0, 1)$. For each $c$, $\hat{T}_N^*(c)$ is clearly normal because it is a sum of re-scaled normals. To obtain joint normality across $c$, use Cram\'er-Wold. Specifically, take an arbitrary number $J$ of points $c_1, ... , c_J$ and $\lambda \in \mathbb{R}^J$. Then:
    $$
    \lambda'\hat{T}_N^*(c) = \frac{1}{\sqrt{N}} \sum^{n-1}_{i=1} \sum^{n}_{j=i+1} \xi_{ij} \hat{e}_{ij} \left(\sum_{a=1}^J \lambda_a\hat{r}_{ij}(c_a) \right).
    $$
    Since $\lambda'\hat{T}_N^*(c)$ is a linear combination of independent standard normals with fixed coefficients for every $\lambda$, by Cram\'er-Wold $\{\hat{T}_N^*(c) : c \in \mathcal{C} \}$ is a mean-zero Gaussian process with covariances $\frac{1}{N} \sum^{n-1}_{i=1} \sum^{n}_{j=i+1} \hat{e}_{ij}^2 \hat{r}_{ij}(c_a) \hat{r}_{ij}(c_b)$ for $a, b = 1, ..., J$.

    Since we have exact normality for each $N$, we no longer rely on Theorem 2.11.1 in \cite{van_der_vaart_weak_1996} to show weak convergence as $n \to \infty$. Instead, we can simply invoke Theorem 7.1 in \cite{billingsley_convergence_1999} which states that $\{\hat{T}_N^*(c) : c \in \mathcal{C} \} \xrightarrow{d} \mathcal{T}$ in $\ell^\infty(\mathcal{C})$, conditional on $\mathcal{D}$, provided:
    \begin{enumerate}
        \item[(i)] $(\hat{T}_N^*(c_1), ..., \hat{T}_N^*(c_J)) \mid \mathcal{D} \xrightarrow{d} (\mathcal{T}(c_1), ..., \mathcal{T}(c_J))$ for every $J \geq 1$ and every $c_1, ..., c_J \in \mathcal{C}$,
        \item[(ii)] The collection of conditional probability distributions of $\hat{T}_N^* \mid \mathcal{D}$ is tight in $\ell^\infty(\mathcal{C})$.
    \end{enumerate}

        To show (i), recall from above that, conditional on $\mathcal{D}$:
    $$
    (\hat{T}_N^*(c_1), ..., \hat{T}_N^*(c_J)) \sim \mathcal{N}(0, \hat{\Sigma}_{\hat{e}}),
    $$
    where $\hat{\Sigma}_{\hat{e}}$ has $(a, b)$-th entry $\frac{1}{N} \sum^{n-1}_{i=1} \sum^{n}_{j=i+1} \hat{e}_{ij}^2 \hat{r}_{ij}(c_a) \hat{r}_{ij}(c_b)$, so its conditional characteristic function is $\exp(-\frac{1}{2}t'\hat{\Sigma}_{\hat{e}}t)$ for $t \in \mathbb{R}^J$. Below, we show that $\hat{\Sigma}_{\hat{e}} \xrightarrow{p} \Sigma$, where $\Sigma$ has $(a, b)$-th entry $\mathbb{E}[\omega_{ij} r_{ij}(c_a) r_{ij}(c_b)]$. Then, every subsequence has a further subsequence along which $\hat{\Sigma}_{\hat{e}} \to \Sigma$ almost surely. Along this subsequence, for almost every realization of the data, $\exp(-\frac{1}{2}t'\hat{\Sigma}_{\hat{e}}t) \to \exp(-\frac{1}{2}t'\Sigma t)$ for every $t \in \mathbb{R}^J$, which is the characteristic function of $\mathcal{N}(0, \Sigma)$. Therefore, by L\'evy's continuity theorem along this subsequence:
    $$
    d_{BL}\left(\mathcal{L}[(\hat{T}_N^*(c_1), ..., \hat{T}_N^*(c_J)) \mid \mathcal{D}], \mathcal{N}(0, \Sigma)\right) \xrightarrow{a.s.} 0,
    $$
    where $\mathcal{L}[\cdot \mid \mathcal{D}]$ denotes the conditional law given $\mathcal{D}$. Finally, since every subsequence has such a further subsequence:
    $$
    d_{BL}\left(\mathcal{L}[(\hat{T}_N^*(c_1), ..., \hat{T}_N^*(c_J)) \mid \mathcal{D}], \mathcal{N}(0, \Sigma)\right) \xrightarrow{p} 0,
    $$
    as required. So we remain to show that $\hat{\Sigma}_{\hat{e}} \xrightarrow{p} \Sigma$. In fact, we show the stronger result:
    $$
    \sup_{c_a, c_b \in \mathcal{C}} \left| \frac{1}{N} \sum^{n-1}_{i=1} \sum^{n}_{j=i+1} \hat{e}_{ij}^2 \hat{r}_{ij}(c_a) \hat{r}_{ij}(c_b) - \mathbb{E}[\omega_{ij} r_{ij}(c_a) r_{ij}(c_b)] \right| = o_P(1),
    $$
    which is used when we apply the continuous mapping theorem again below. To this end, decompose:
    \begin{align*}
         \hat{\Sigma}_{\hat{e}}(c_a, c_b) - \Sigma(c_a, c_b)
         &= \underbrace{\frac{1}{N} \sum^{n-1}_{i=1} \sum^{n}_{j=i+1}[\hat{e}_{ij}^2 - e_{ij}^2]\hat{r}_{ij}(c_a)\hat{r}_{ij}(c_b)}_{\text{(I)}} \\
         &+ \underbrace{\frac{1}{N} \sum^{n-1}_{i=1} \sum^{n}_{j=i+1}e_{ij}^2[\hat{r}_{ij}(c_a)\hat{r}_{ij}(c_b) - r_{ij}(c_a)r_{ij}(c_b)]}_{\text{(II)}}
         \\
         &+ \underbrace{\frac{1}{N} \sum^{n-1}_{i=1} \sum^{n}_{j=i+1}[e_{ij}^2 - \omega_{ij}]r_{ij}(c_a)r_{ij}(c_b)}_{\text{(III)}} \\
         &+ \underbrace{\frac{1}{N} \sum^{n-1}_{i=1} \sum^{n}_{j=i+1}\omega_{ij}r_{ij}(c_a)r_{ij}(c_b) - \mathbb{E}[\omega_{ij} r_{ij}(c_a)r_{ij}(c_b)]}_{\text{(IV)}}.
    \end{align*}

        To bound (I), first use the mean value theorem and \eqref{e-first-derivative} to obtain:
    $$
    |\hat{e}_{ij} - e_{ij}| \leq C_1 \|V_{ij}\| \|\hat{\gamma} - \gamma_0 \|.
    $$
    Next, by \eqref{e-bound} and the Cauchy-Schwarz inequality, $|e_{ij}| \leq 1 + \|V_{ij}\|\|\gamma_0\| \leq C'(1 + \|V_{ij}\|)$, where $C' \coloneq \max\{1, \|\gamma_0\|\} < \infty$. Combining these, and using $\|V_{ij}\|^2 \leq \|V_{ij}\|(1 + \|V_{ij}\|)$:
    \begin{align*}
        |\hat{e}_{ij}^2 - e_{ij}^2| &\leq |\hat{e}_{ij} - e_{ij}|(2|e_{ij}| + |\hat{e}_{ij} - e_{ij}|) \\
        &\leq \left(2C_1C'\|\hat{\gamma} - \gamma_0 \| + C_1^2\|\hat{\gamma} - \gamma_0 \|^2\right) \|V_{ij}\|(1 + \|V_{ij}\|).
    \end{align*}
    Moreover, since $|\hat{r}_{ij}(c_a)\hat{r}_{ij}(c_b)| \leq \frac{1}{2}(\hat{r}_{ij}(c_a)^2 + \hat{r}_{ij}(c_b)^2)$, the bound on $\hat{r}_{ij}(c)^2$ from Step 4 gives:
    $$
    \sup_{c_a, c_b \in \mathcal{C}} |\hat{r}_{ij}(c_a)\hat{r}_{ij}(c_b)| \leq 3\bar{R}_{ij}^2 + 3 \| W_{ij} - W_{ji}\|^2 \|\hat{\gamma} - \gamma_0\|^2 + 3\sup_{c \in \mathcal{C}}\| \hat{\Pi}(c) - \Pi(c)\|^2\|V_{ij}\|^2.
    $$
    Plugging both bounds into (I) yields:
    \begin{align*}
        \sup_{c_a, c_b \in \mathcal{C}} |\text{(I)}| &\leq 3\left(2C_1C'\|\hat{\gamma} - \gamma_0 \| + C_1^2\|\hat{\gamma} - \gamma_0 \|^2\right) \\
        &\quad \times \Bigg(\frac{1}{N} \sum^{n-1}_{i=1} \sum^{n}_{j=i+1}\|V_{ij}\|(1 + \|V_{ij}\|)\bar{R}_{ij}^2 \\
        &\quad + \|\hat{\gamma} - \gamma_0\|^2 \frac{1}{N} \sum^{n-1}_{i=1} \sum^{n}_{j=i+1}\|V_{ij}\|(1 + \|V_{ij}\|)\| W_{ij} - W_{ji}\|^2 \\
        &\quad + \sup_{c \in \mathcal{C}}\| \hat{\Pi}(c) - \Pi(c)\|^2 \frac{1}{N} \sum^{n-1}_{i=1} \sum^{n}_{j=i+1}\|V_{ij}\|^3(1 + \|V_{ij}\|)\Bigg).
    \end{align*}
    Each sample average has a fixed kernel that is a polynomial of degree at most four in $\|W_{ij}\|$ and $\|W_{ji}\|$, so each is $O_P(1)$ by \cite{serfling_approximation_1980} under \ref{ass:T3}. Since $\hat{\gamma} - \gamma_0 = O_P(N^{-\frac{1}{2}})$ and $\sup_{c \in \mathcal{C}}\| \hat{\Pi}(c) - \Pi(c)\| = o_P(1)$ from Step 4:
    $$
    \sup_{c_a, c_b \in \mathcal{C}} |\text{(I)}| = O_P(N^{-\frac{1}{2}}) \left(O_P(1) + O_P(N^{-1}) + o_P(1)\right) = o_P(1),
    $$
    as required.

        To bound (II), use the inequalities from the bound on (II) in Step 4. Specifically, for every $c \in \mathcal{C}$, $|r_{ij}(c)| \leq \bar{R}_{ij}$ and $|\hat{r}_{ij}(c) - r_{ij}(c)| \leq \| W_{ij} - W_{ji}\| \|\hat{\gamma} - \gamma_0\| + \|V_{ij}\| \sup_{c \in \mathcal{C}}\| \hat{\Pi}(c) - \Pi(c)\|$. Combining these with $(x + y)^2 \leq 2x^2 + 2y^2$ gives:
    \begin{align*}
        |\hat{r}_{ij}(c_a)\hat{r}_{ij}(c_b) - r_{ij}(c_a)r_{ij}(c_b)| &\leq |\hat{r}_{ij}(c_a) - r_{ij}(c_a)||\hat{r}_{ij}(c_b)| + |\hat{r}_{ij}(c_b) - r_{ij}(c_b)||r_{ij}(c_a)| \\
        &\leq 2\bar{R}_{ij}\| W_{ij} - W_{ji}\| \|\hat{\gamma} - \gamma_0\| + 2\bar{R}_{ij}\|V_{ij}\| \sup_{c \in \mathcal{C}}\| \hat{\Pi}(c) - \Pi(c)\| \\
        &\quad + 2\| W_{ij} - W_{ji}\|^2 \|\hat{\gamma} - \gamma_0\|^2 + 2\|V_{ij}\|^2 \sup_{c \in \mathcal{C}}\| \hat{\Pi}(c) - \Pi(c)\|^2.
    \end{align*}
    Multiplying by $e_{ij}^2$ and averaging:
    \begin{align*}
        \sup_{c_a, c_b \in \mathcal{C}} |\text{(II)}| &\leq 2\|\hat{\gamma} - \gamma_0\| \frac{1}{N} \sum^{n-1}_{i=1} \sum^{n}_{j=i+1} e_{ij}^2\bar{R}_{ij}\| W_{ij} - W_{ji}\| \\
        &\quad + 2\sup_{c \in \mathcal{C}}\| \hat{\Pi}(c) - \Pi(c)\| \frac{1}{N} \sum^{n-1}_{i=1} \sum^{n}_{j=i+1} e_{ij}^2\bar{R}_{ij}\|V_{ij}\| \\
        &\quad + 2\|\hat{\gamma} - \gamma_0\|^2 \frac{1}{N} \sum^{n-1}_{i=1} \sum^{n}_{j=i+1} e_{ij}^2\| W_{ij} - W_{ji}\|^2 \\
        &\quad + 2\sup_{c \in \mathcal{C}}\| \hat{\Pi}(c) - \Pi(c)\|^2 \frac{1}{N} \sum^{n-1}_{i=1} \sum^{n}_{j=i+1} e_{ij}^2\|V_{ij}\|^2.
    \end{align*}
    Each sample average has non-negative summands with finite expectation, since $\mathbb{E}[e_{ij}^2 \mid \mathcal{X}] = \omega_{ij} \leq \frac{2}{\pi}$ and \ref{ass:T3} holds, so each is $O_P(1)$ by Markov's inequality. Since $\hat{\gamma} - \gamma_0 = O_P(N^{-\frac{1}{2}})$ and $\sup_{c \in \mathcal{C}}\| \hat{\Pi}(c) - \Pi(c)\| = o_P(1)$ from Step 4:
    $$
    \sup_{c_a, c_b \in \mathcal{C}} |\text{(II)}| = O_P(N^{-\frac{1}{2}}) + o_P(1) + O_P(N^{-1}) + o_P(1) = o_P(1),
    $$
    as required.

        To bound (III), apply Lemma 1 with $\eta = 1$ and parameter $(c_a, c_b) \in \mathcal{C}^2$, which is compact by \ref{ass:T5}. Continuity of the kernel in $(c_a, c_b)$ follows from the Lipschitz bound \eqref{r-lipschitz}. Symmetry is obvious. Moreover:
    $$
    \sup_{c_a, c_b \in \mathcal{C}} \left|[e_{ij}^2 - \omega_{ij}]r_{ij}(c_a)r_{ij}(c_b)\right| \leq |e_{ij}^2 - \omega_{ij}| \bar{R}_{ij}^2,
    $$
    and, since $\mathbb{E}[e_{ij}^2 \mid \mathcal{X}] = \omega_{ij}$, \eqref{e-fourth-moment} gives:
    $$
    \mathbb{E}\left[(e_{ij}^2 - \omega_{ij})^2 \bar{R}_{ij}^4\right] = \mathbb{E}\left[\mathbb{V}\text{ar}[e_{ij}^2 \mid \mathcal{X}] \bar{R}_{ij}^4\right] \leq C_3\, \mathbb{E}\left[\bar{R}_{ij}^4\right] < \infty,
    $$
    under \ref{ass:T3}. Finally, $\mathbb{E}[[e_{ij}^2 - \omega_{ij}]r_{ij}(c_a)r_{ij}(c_b)] = \mathbb{E}[\mathbb{E}[e_{ij}^2 - \omega_{ij} \mid \mathcal{X}]r_{ij}(c_a)r_{ij}(c_b)] = 0$ for every $(c_a, c_b) \in \mathcal{C}^2$. So, by Lemma 1, $ \sup_{c_a, c_b \in \mathcal{C}} \left|\text{(III)}\right| = o_P(1)$, as required.

    Finally, for (IV), simply apply Lemma 1 with $\eta = 1$ again. Compactness, continuity, and symmetry are the same as with (III). For the envelope use:
    $$
    \sup_{c_a, c_b \in \mathcal{C}} \left|\omega_{ij}r_{ij}(c_a)r_{ij}(c_b)\right| \leq \frac{2}{\pi}\bar{R}_{ij}^2,
    $$
    and $\mathbb{E}[\bar{R}_{ij}^4] < \infty$ under \ref{ass:T3} again. So, by Lemma 1, $\sup_{c_a, c_b \in \mathcal{C}} \left|\text{(IV)}\right| = o_P(1)$, and therefore, all together, $\sup_{c_a, c_b \in \mathcal{C}} \left|\hat{\Sigma}_{\hat{e}}(c_a, c_b) - \Sigma(c_a, c_b)\right| = o_P(1)$, as required.

    Next, we show (ii) of Theorem 7.1 in \cite{billingsley_convergence_1999}. To this end, we apply Problem 2.4.11 in \cite{karatzas_brownian_1998} which guarantees tightness if there exists positive constants $\alpha, \beta, \nu$, and $C$ such that:\footnote{$\alpha, \beta, \nu$ must be universal constants, but $C$ is allowed to depend on $\bar{C}$.}
    \begin{enumerate}
        \item[(a)] $\sup_N \mathbb{E}[|\hat{T}_N^*(\ubar{C})|^\nu \mid \mathcal{D}] < \infty$,
        \item[(b)] $\sup_N \mathbb{E}[|\hat{T}_N^*(c_a) - \hat{T}_N^*(c_b)|^\alpha \mid \mathcal{D}] \leq C|c_a - c_b|^{1 + \beta}$ for all $c_a, c_b \in \mathcal{C}$.
    \end{enumerate}

    For (a), we know that $\hat{T}_N^*(\ubar{C}) \mid \mathcal{D} \sim \mathcal{N}(0, \hat{\Sigma}_{\hat{e}}(\ubar{C}, \ubar{C}))$, so $\mathbb{E}[|\hat{T}_N^*(\ubar{C})|^2 \mid \mathcal{D}] = \hat{\Sigma}_{\hat{e}}(\ubar{C}, \ubar{C})$. From above, we know that this converges to $\Sigma(\ubar{C}, \ubar{C})$, which is finite. This is true for all $N$, so $\sup_N \mathbb{E}[|\hat{T}_N^*(\ubar{C})|^\nu \mid \mathcal{D}]$ with $\nu = 2$ is finite, as required.

    For (b), note that:
    \begin{align*}
        \mathbb{E}[|\hat{T}_N^*(c_a) - \hat{T}_N^*(c_b)|^2 \mid \mathcal{D}] &= \mathbb{E}\left[\left(\frac{1}{\sqrt{N}} \sum^{n-1}_{i=1} \sum^{n}_{j=i+1} \xi_{ij} \hat{e}_{ij} [\hat{r}_{ij}(c_a) - \hat{r}_{ij}(c_b)]\right)^2 \mid \mathcal{D}\right] \\
        &= \frac{1}{N} \sum^{n-1}_{i=1} \sum^{n}_{j=i+1} \hat{e}_{ij}^2 [\hat{r}_{ij}(c_a) - \hat{r}_{ij}(c_b)]^2 \\
        &\leq |c_a - c_b|^2 \frac{1}{N} \sum^{n-1}_{i=1} \sum^{n}_{j=i+1} \hat{e}_{ij}^2 \hat{L}(V_{ij})^2,
    \end{align*}
    where for the second line we write the squared double sum as a quadruple sum and then use that $\mathbb{E}[\xi_{ij} \xi_{kl}] = 1$ since $\xi_{ij}$ are i.i.d. $\mathcal{N}(0,1)$, independent of $\mathcal{D}$, and for the third line we use \eqref{r-lipschitz} with $\hat{L}(V_{ij}) \coloneq \frac{1}{\sqrt{\pi}} + \hat{C}_\Pi \|V_{ij}\|$ in place of $L(V_{ij}) \coloneq \frac{1}{\sqrt{\pi}} + C_\Pi \|V_{ij}\|$, where $\hat{C}_\Pi \coloneq \frac{1}{\sqrt{\pi}} \| (\frac{1}{N}\sum^{n-1}_{i=1} \sum^{n}_{j=i+1}\hat{\omega}_{ij} V_{ij}V_{ij}')^{-1}\|\frac{1}{N}\sum^{n-1}_{i=1} \sum^{n}_{j=i+1}\hat{\omega}_{ij} \|V_{ij}\|$. If we can show that:
    $$
    \frac{1}{N} \sum^{n-1}_{i=1} \sum^{n}_{j=i+1} \hat{e}_{ij}^2 \hat{L}(V_{ij})^2 \xrightarrow{p} \mathbb{E}[\omega_{ij}L(V_{ij})^2],
    $$
    then, for all large $N$, (b) holds with $\alpha=2$, $\beta=1$, and $C=(\mathbb{E}[\omega_{ij}L(V_{ij})^2] + 1)^2$. For this convergence in probability result, decompose:
    \begin{align*}
         \frac{1}{N} \sum^{n-1}_{i=1} \sum^{n}_{j=i+1} \hat{e}_{ij}^2 \hat{L}(V_{ij})^2 - \mathbb{E}[\omega_{ij}L(V_{ij})^2]
         &= \underbrace{\frac{1}{N} \sum^{n-1}_{i=1} \sum^{n}_{j=i+1} \hat{e}_{ij}^2 \hat{L}(V_{ij})^2 - \frac{1}{N} \sum^{n-1}_{i=1} \sum^{n}_{j=i+1} \hat{e}_{ij}^2 L(V_{ij})^2}_{\text{(I)}} \\
         &+ \underbrace{\frac{1}{N} \sum^{n-1}_{i=1} \sum^{n}_{j=i+1} \hat{e}_{ij}^2 L(V_{ij})^2 - \frac{1}{N} \sum^{n-1}_{i=1} \sum^{n}_{j=i+1} e_{ij}^2 L(V_{ij})^2}_{\text{(II)}}
         \\
         &+ \underbrace{\frac{1}{N} \sum^{n-1}_{i=1} \sum^{n}_{j=i+1} e_{ij}^2 L(V_{ij})^2 - \frac{1}{N} \sum^{n-1}_{i=1} \sum^{n}_{j=i+1} \omega_{ij} L(V_{ij})^2}_{\text{(III)}} \\
         &+ \underbrace{\frac{1}{N} \sum^{n-1}_{i=1} \sum^{n}_{j=i+1} \omega_{ij} L(V_{ij})^2 - \mathbb{E}[\omega_{ij}L(V_{ij})^2]}_{\text{(IV)}}.
    \end{align*}
    Bounding (I)-(IV) is a very similar exercise to above, so we do not repeat it. As always, combine the bounds from Step 1 with \cite{serfling_approximation_1980} and Lemma 1 to show each term is $o_P(1)$.

    We can finally conclude that $\{\hat{T}_N^*(c) : c \in \mathcal{C} \} \xrightarrow{d} \mathcal{T}$ in $\ell^\infty(\mathcal{C})$, conditional on $\mathcal{D}$, in probability. To get to $t_N^{*\text{sup}} \overset{d}{\to} \sup_{c \in \mathcal{C}}\mathcal{G}(c)$ conditional on $\mathcal{D}$, in probability, we simply apply the continuous mapping theorem again as in Steps 4 and 5 above. This is where $\sup_{c_a, c_b \in \mathcal{C}} \left|\hat{\Sigma}_{\hat{e}}(c_a, c_b) - \Sigma(c_a, c_b)\right| = o_P(1)$ is used. Again, the ``in probability'' statement comes from the fact we consider a further subsequence along which convergence occurs almost surely.

    \textbf{Step 7:} We begin this step by showing that for each $x \in \mathbb{R}$:
    $$
    \hat{F}_N(x) \coloneq \mathbb{P}[t_N^{*\text{sup}} \leq x \mid \mathcal{D}] \xrightarrow{p} F(x),
    $$
    where $F$ is the distribution function of $\sup_{c \in \mathcal{C}}\mathcal{G}(c)$. By Step 6:
    $$
    d_{BL}\left(\mathcal{L}[t_N^{*\text{sup}} \mid \mathcal{D}], \mathcal{L}\left[\sup_{c \in \mathcal{C}}\mathcal{G}(c)\right]\right) \xrightarrow{p} 0.
    $$
    So every subsequence has a further subsequence along which this convergence holds almost surely. Along this further subsequence, for almost every realization of the data, the conditional law of $t_N^{*\text{sup}}$ given $\mathcal{D}$ converges weakly to the law of $\sup_{c \in \mathcal{C}}\mathcal{G}(c)$. The boundary of $(-\infty, x]$ is $\{x\}$, and, since $F$ is continuous by Step 5, $\mathbb{P}[\sup_{c \in \mathcal{C}}\mathcal{G}(c) = x] = 0$. So, by the Portmanteau theorem, $\hat{F}_N(x) \to F(x)$ along this further subsequence for almost every realization of the data. Finally, since every subsequence has such a further subsequence, $\hat{F}_N(x) \xrightarrow{p} F(x)$, as required.

        Next, we show that:
    $$
    c_{1-\alpha}^* \coloneq \inf\{x: \hat{F}_N(x) \geq 1 - \alpha\} \xrightarrow{p} c_{1-\alpha} \coloneq F^{-1}(1-\alpha),
    $$
    where $F^{-1}$ is well-defined by continuity and strict monotonicity of $F$ from Step 5. Fix some $\epsilon > 0$. Since $F$ is strictly increasing and $F(c_{1-\alpha})=1-\alpha$:
    $$
    F(c_{1-\alpha} - \epsilon) < 1 -\alpha < F(c_{1-\alpha} + \epsilon).
    $$
    Take $\eta \coloneq \min\{(1- \alpha) - F(c_{1-\alpha} - \epsilon), F(c_{1-\alpha} + \epsilon) - (1- \alpha)\} > 0$. Then apply $\hat{F}_N(x) \xrightarrow{p} F(x)$ to $x = c_{1-\alpha} - \epsilon$ and $x = c_{1-\alpha} + \epsilon$. This gives:
    $$
    \mathbb{P}[|\hat{F}_N(c_{1-\alpha} - \epsilon) - F(c_{1-\alpha} - \epsilon)| < \eta] \to 1,
    $$
    $$
    \mathbb{P}[|\hat{F}_N(c_{1-\alpha} + \epsilon) - F(c_{1-\alpha} + \epsilon)| < \eta] \to 1.
    $$
    Call the intersection of these two events $E_N$, so that $\mathbb{P}[E_N] \to 1$. On $E_N$:
    $$
    \hat{F}_N(c_{1-\alpha} - \epsilon) < F(c_{1-\alpha} - \epsilon) + \eta \leq F(c_{1-\alpha} - \epsilon) + \left[(1- \alpha) - F(c_{1-\alpha} - \epsilon)\right] = 1 -\alpha,
    $$
    $$
    \hat{F}_N(c_{1-\alpha} + \epsilon) > F(c_{1-\alpha} + \epsilon) - \eta \geq F(c_{1-\alpha} + \epsilon) - \left[F(c_{1-\alpha} + \epsilon) - (1- \alpha)\right] = 1 -\alpha.
    $$
    So, together, on $E_N$:
    $$
    \hat{F}_N(c_{1-\alpha} - \epsilon) < 1 - \alpha < \hat{F}_N(c_{1-\alpha} + \epsilon).
    $$
    We now show that, on $E_N$, $c_{1-\alpha} - \epsilon < c_{1-\alpha}^* \leq c_{1-\alpha} + \epsilon$. For the upper bound, $\hat{F}_N(c_{1-\alpha} + \epsilon) > 1 - \alpha$, so $c_{1-\alpha} + \epsilon \in \{x: \hat{F}_N(x) \geq 1 - \alpha\}$, and hence $c_{1-\alpha}^* \leq c_{1-\alpha} + \epsilon$. For the lower bound, $\hat{F}_N$ is non-decreasing, so for every $x \leq c_{1-\alpha} - \epsilon$, $\hat{F}_N(x) \leq \hat{F}_N(c_{1-\alpha} - \epsilon) < 1 - \alpha$. Thus, every element of $\{x: \hat{F}_N(x) \geq 1 - \alpha\}$ is strictly greater than $c_{1-\alpha} - \epsilon$, so certainly $c_{1-\alpha}^* \geq c_{1-\alpha} - \epsilon$. Moreover, since $\hat{F}_N$ is a conditional distribution function, it is right-continuous and therefore the infimum is attained: $\hat{F}_N(c_{1-\alpha}^*) \geq 1 - \alpha$. This rules out $c_{1-\alpha}^* = c_{1-\alpha} - \epsilon$ since $\hat{F}_N(c_{1-\alpha} - \epsilon) < 1 - \alpha$, so $c_{1-\alpha} - \epsilon < c_{1-\alpha}^*$. In sum, $E_N \subseteq \{|c_{1-\alpha}^* - c_{1-\alpha}| \leq \epsilon\}$, and thus:
    $$
    \mathbb{P}\left[|c_{1-\alpha}^* - c_{1-\alpha}| > \epsilon\right] \leq \mathbb{P}[E_N^c] \to 0.
    $$
    Since $\epsilon > 0$ was arbitrary, $c_{1-\alpha}^* \xrightarrow{p} c_{1-\alpha}$, as required.

    Finally, we conclude that $\mathbb{P}[t_N^{\text{sup}} > c_{1-\alpha}^*] \to \alpha$. Again, fix $\epsilon > 0$. Note that:
    $$
    \mathbb{P}[t_N^{\text{sup}} > c_{1-\alpha}^*] \leq \mathbb{P}[t_N^{\text{sup}} > c_{1-\alpha} - \epsilon] + \mathbb{P}[|c_{1-\alpha}^* - c_{1-\alpha}| > \epsilon],
    $$
    $$
    \mathbb{P}[t_N^{\text{sup}} > c_{1-\alpha} + \epsilon] \leq \mathbb{P}[t_N^{\text{sup}} > c_{1-\alpha}^*] + \mathbb{P}[|c_{1-\alpha}^* - c_{1-\alpha}| > \epsilon].
    $$
    As $N \to \infty$, $\mathbb{P}\left[|c_{1-\alpha}^* - c_{1-\alpha}| > \epsilon\right] \to 0$ from above. Moreover, by Step 5, $t_N^{\text{sup}} \xrightarrow{d} \sup_{c \in \mathcal{C}}\mathcal{G}(c)$, whose distribution function $F$ is continuous, so $\mathbb{P}[t_N^{\text{sup}} > c_{1-\alpha} - \epsilon] \to 1 - F(c_{1-\alpha} - \epsilon)$ and $\mathbb{P}[t_N^{\text{sup}} > c_{1-\alpha} + \epsilon] \to 1 - F(c_{1-\alpha} + \epsilon)$. Hence:
    $$
    1 - F(c_{1-\alpha} + \epsilon) \leq \liminf_{N\to \infty}\mathbb{P}[t_N^{\text{sup}} > c_{1-\alpha}^*] \leq \limsup_{N\to \infty}\mathbb{P}[t_N^{\text{sup}} > c_{1-\alpha}^*] \leq 1 - F(c_{1-\alpha} - \epsilon).
    $$
    Since $\epsilon > 0$ was arbitrary, and $F$ is continuous at $c_{1-\alpha}$ with $F(c_{1-\alpha}) = 1 - \alpha$, letting $\epsilon \downarrow 0$ sends both the upper and lower bounds to $\alpha$, so:
    $$
    \liminf_{N\to \infty}\mathbb{P}[t_N^{\text{sup}} > c_{1-\alpha}^*] = \limsup_{N\to \infty}\mathbb{P}[t_N^{\text{sup}} > c_{1-\alpha}^*] = \lim_{N\to \infty}\mathbb{P}[t_N^{\text{sup}} > c_{1-\alpha}^*] = \alpha,
    $$
    as required.
    
\end{proof}

\newpage

\section{Specification Test Robustness Checks}

\setcounter{figure}{0}
\renewcommand{\thefigure}{A\arabic{figure}}

\begin{figure}[H]
  \centering
  \caption{Figure \ref{fig:supLM_test_power_curves} with $n=50$}
  \label{fig:supLM_test_power_curves_n50}
  \includegraphics[width=0.9\textwidth]{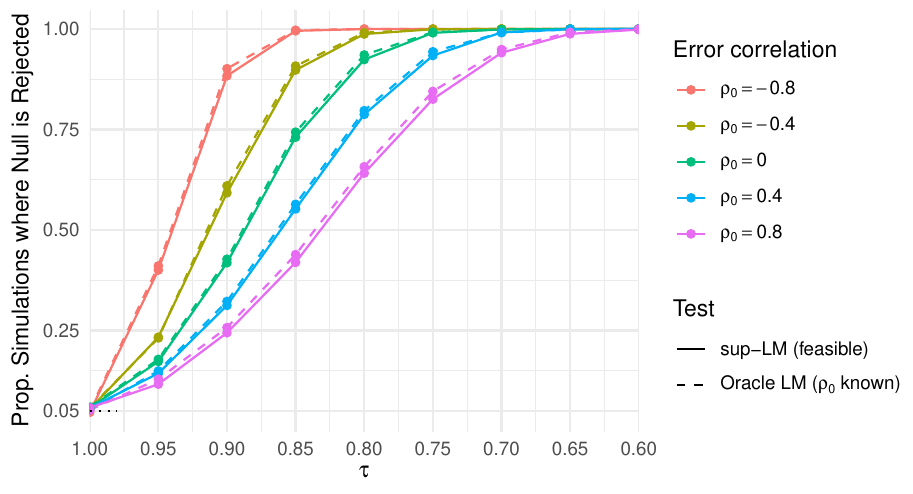}
  \fnote{This figure is constructed identically to Figure \ref{fig:supLM_test_power_curves}, except that now $n = 50$ instead of $n = 100$. The oracle test still hardly outperforms the sup-LM test and, although the power curves are flatter, power is still 100\% for all values of $\rho$ provided $\tau \leq 0.6$.}
\end{figure}

\begin{figure}[H]
  \centering
  \caption{Figure \ref{fig:supLM_test_power_curves} with supremum search set $[-0.6, 0.6]$}
  \label{fig:supLM_test_power_curves_restricted}
  \includegraphics[width=0.9\textwidth]{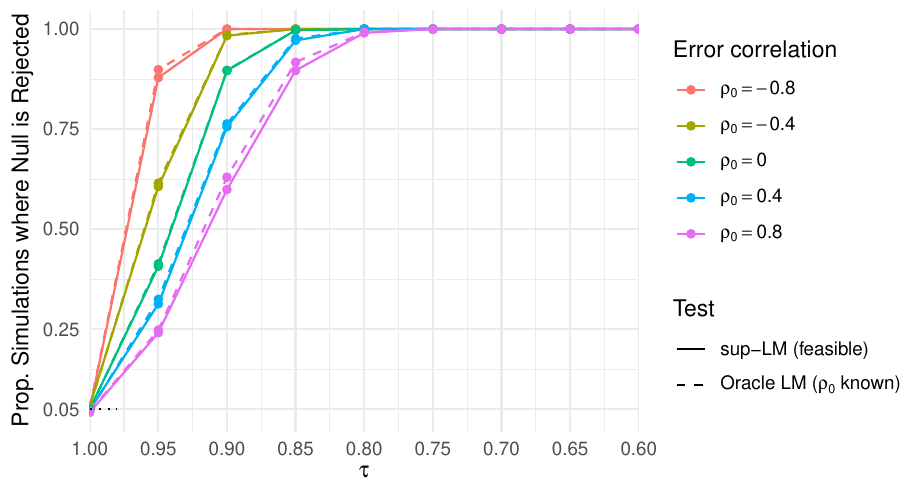}
  \fnote{This figure is constructed identically to Figure \ref{fig:supLM_test_power_curves}, except that now $[\ubar{\rho}, \bar{\rho}] = [-0.6, 0.6]$ instead of $[\ubar{\rho}, \bar{\rho}] = [-0.9, 0.9]$. The figures are practically indistinguishable except that the power discrepancy between the sup-LM and oracle tests shrinks slightly for $\rho \in \{-0.4, 0, 0.4\}$ and increases slightly for $\rho \in \{-0.8, 0.8\}$.}
\end{figure}
\end{document}